\documentclass[a4paper,11pt]{article}
\pdfoutput=1

\usepackage{tcs}

\begin{document}

\title{Learning quantum Hamiltonians at any temperature \\ in polynomial time}
\author{
Ainesh Bakshi \\
\texttt{ainesh@mit.edu} \\
MIT
\and
Allen Liu \\
\texttt{cliu568@mit.edu} \\
MIT
\and
Ankur Moitra \\
\texttt{moitra@mit.edu} \\
MIT
\and
Ewin Tang \\
\texttt{ewin@berkeley.edu} \\
UC Berkeley
}
\date{}

\maketitle

\begin{abstract}
    We study the problem of learning a local quantum Hamiltonian $H$ given copies of its Gibbs state $\rho = e^{-\beta H}/\tr(e^{-\beta H})$ at a known inverse temperature $\beta>0$.
    Anshu, Arunachalam, Kuwahara, and Soleimanifar~\cite{aaks20} gave an algorithm to learn a Hamiltonian on $n$ qubits to precision $\eps$ with only polynomially many copies of the Gibbs state, but which takes exponential time.
    Obtaining a computationally efficient algorithm has been a major open problem~\cite{alhambra22,aa23}, with prior work only resolving this in the limited cases of high temperature~\cite{hkt21} or commuting terms~\cite{aaks21commuting}.
    We fully resolve this problem, giving a polynomial time algorithm for learning $H$ to precision $\eps$ from polynomially many copies of the Gibbs state at any constant $\beta > 0$.

    Our main technical contribution is a new \textit{flat} polynomial approximation to the exponential function, and a translation between multi-variate scalar polynomials and nested commutators. This enables us to formulate Hamiltonian learning as a polynomial system. We then show that solving a low-degree sum-of-squares relaxation of this polynomial system suffices to accurately learn the Hamiltonian.
\end{abstract}



\thispagestyle{empty}
\clearpage
\newpage

\microtypesetup{protrusion=false}
\tableofcontents{}
\thispagestyle{empty}
\microtypesetup{protrusion=true}
\clearpage
\setcounter{page}{1}

\section{Introduction}

Quantum computing has sparked a major interest in increasing the scale and control of quantum systems~\cite{gan14}.
This increased interest is accompanied with the need for better algorithms to characterize and verify these systems~\cite{cekkz21}.
A central computational task in controlling and verifying quantum systems is that of \emph{Hamiltonian learning}, where the goal is to estimate physical properties, namely the interaction strengths, of an interacting quantum many-body system from measurements~\cite{cpfsgb10,slp11,BAL19,aaks20}.
Formally, we consider $\qubits$ qubits (quantum particles with a local dimension of two) on a lattice.\footnote{
    Our results do not need the Hamiltonian to be \emph{geometrically local} as described here; we merely require it to be \emph{low-intersection} in the sense of Haah, Kothari, and Tang~\cite{hkt21}. So, our algorithm will still work if the locality structure of the qubits is, say, an expander graph.
}
The resulting system is characterized by a Hamiltonian, a $2^\qubits \times 2^\qubits$ complex Hermitian matrix of the form
$H = \sum_{a=1}^\terms \lambda_a E_a$,
where a term $E_a$ encodes an interaction on at most $\locality$ of the particles, and the coefficient $\lambda_a \in [-1, 1]$ is the strength of the corresponding interaction.
We assume that the system has reached thermal equilibrium at a known inverse temperature $\beta$, in which case it is in the Gibbs state with density matrix
$$\rho = \frac{e^{-\beta H}}{\Tr e^{-\beta H}}.$$
The density matrix is normalized by the partition function, $\Tr e^{-\beta H}$, which ensures that $\rho$ has trace one.
The goal of the Hamiltonian learning problem is to estimate the $\lambda_a$'s, given the ability to prepare copies of the Gibbs state.
\begin{problem*}[Hamiltonian learning, \cref{prob:main} (informal)]
    Consider $\qubits$ qubits on a constant-dimensional lattice.
    Let $H = \sum_{a=1}^\terms \lambda_a E_a \in \mathbb{C}^{2^\qubits \times 2^\qubits}$ be a Hamiltonian whose terms $E_a$ are known, distinct, non-identity Pauli operators supported on at most $k$ qubits that are local with respect to the lattice.
    Further suppose the coefficients $\lambda_a \in \R$ satisfy $\abs{\lambda_a} \leq 1$.
    Given copies of the corresponding Gibbs state $\rho$ at a known inverse temperature $\beta > 0$, and $\eps>0$, find estimates $\tilde{\lambda}_a$ such that $\abs{\tilde{\lambda}_a - \lambda_a} \leq \eps$ for all $a \in [\terms]$.
\end{problem*}
We are interested in both the number of copies of $\rho$ that we need, which is called the sample complexity, as well as the running time of the algorithm.
Of particular interest to us is Hamiltonian learning in the \emph{low-temperature} regime, where $\beta$ is an arbitrarily large constant.

\paragraph{Motivation.}
As alluded to above, this problem is of fundamental importance in science and engineering.
For example, in pursuit of understanding the phenomena like topological order and superconductivity that are studied in condensed matter physics, experimentalists carefully design systems which exhibit this exotic behavior.
In particular, analog quantum simulators are tuned to obey poorly-understood Hamiltonians like the Fermi--Hubbard model for experimental exploration~\cite{gb17,ecfwsg11,hcbbbd23}.\footnote{
    See also \cite[Section 5.4.2]{ghls15} for a description of this work aimed towards theoretical computer scientists.
}
For these experiments, a natural goal is to learn the interactions which give rise to various phenomena~\cite{wgfc14,kbevz21}.
Intractability of computation is a major barrier to resolving open problems like finding the phase diagram of the 2D Fermi--Hubbard model, so having better algorithmic tools is of key importance in this domain~\cite{leblanc15}.
This problem also arises when engineering quantum systems: a major challenge in building near-term quantum devices is being able to validate them---certify that they implement the desired Hamiltonian---and understand sources of error~\cite{slp11,shnbdu14}.
Quantum devices with $100$ or more qubits are challenging to simulate classically, but quantum Hamiltonian learning has emerged as an alternative strategy for benchmarking devices by combining quantum resources and classical learning techniques~\cite{cekkz21,gcc24}.

The low-temperature setting is of particular interest because because quantum phenomena are most prominent at zero or near-zero temperature~\cite{afov08}, precisely where high-temperature series expansions fail~\cite{leblanc15}.\footnote{
    Morally, these expansions fail precisely because of the non-local quantum correlations we'd like to understand!
}
In some sense, this is the only relevant setting for analog quantum simulators, since models at high-temperature can be solved with a classical computer~\cite[Chapter 8]{ohz06}, without needing to resort to a quantum simulation.
More generally, low temperature is the \emph{computationally} interesting regime, since quantum advantage is a low-temperature phenomenon: ``temperature scaling'' laws show that quantum annealers can only achieve large speedups over classical computers when $\beta$ scales with system size~\cite{amh17}.

\paragraph{Prior work.}
Despite its importance, the computational complexity of Hamiltonian learning from Gibbs states is not well understood.
Anshu, Arunachalam, Kuwahara, and Soleimanifar gave the first polynomial sample complexity bounds for this task in 2020~\cite{aaks20}, attaining coefficient estimates using
\begin{equation}
    \frac{2^{\poly(\beta)} m^2 \log m}{\beta^c \epsilon^2}
    \tag*{\cite{aaks20}}
\end{equation}
copies of the Gibbs state~\cite[Remark 4.5]{hkt21}.
However, their work comes with a serious drawback: it is computationally inefficient.
In particular, they give a stochastic gradient descent algorithm and show that it converges to the true parameters in a small number of iterations, but actually computing an iterate involves evaluating a log-partition function, which is well-known to be computationally hard even for classical systems \cite{montanari2015computational}.

Prior work has obtained fast algorithms for Hamiltonian learning in limited regimes.
A follow-up paper of Anshu, Arunachalam, Kuwahara and Soleimanifar~\cite{aaks21commuting} shows that when the terms of $H$ commute, then a direct generalization of the classical algorithm learns the parameters efficiently.
Further, \cite{aaks20} notes that their suggested algorithm can be performed in polynomial time for sufficiently high temperature (small $\beta$), since in this regime the log-partition function can be evaluated, using that its multivariate Taylor series expansion converges rapidly.
Haah, Kothari, and Tang~\cite{hkt21} later gave an improved algorithm that achieves the sample and time complexity
\begin{equation}
    \frac{e^{\bigO{\beta}} \log \terms}{\beta^2 \epsilon^2} \quad \text{and} \quad \frac{ \terms  e^{\bigO{\beta}} \log \terms}{\beta^2 \epsilon^2}, \tag*{\cite{hkt21}}
\end{equation}
respectively, which they prove is tight up to the constant factor in the exponential, even in the classical case.

However, a central open question remains~\cite{aaks20,hkt21,alhambra22,aa23}: 

\begin{question}
\label{question:central-open-q}
\begin{center}
    \textit{Can Hamiltonian learning at low temperature be solved in time polynomial in $\qubits$?} 
\end{center}
\end{question}

In practice quantum many-body systems are run at low temperature, which is also when most macroscopic phenomena arise, so this is the most important regime for the problem.
However, as we discuss later, no strategies had been suggested for solving Hamiltonian learning at low temperature.
In fact, the situation is even more dire: all approaches to Hamiltonian learning used in prior settings fail catastrophically here, since reduction to sufficient statistics~\cite{aaks20}, efficient computation of the partition function~\cite{aaks20}, the approximate Markov property~\cite{kkb20}, and cluster expansion~\cite{hkt21} all provably fail for sufficiently large $\beta$.
This state of the literature reflects a broader scarcity of algorithmic tools known for understanding Hamiltonians outside of special settings like high temperature or one dimension.
So, a negative resolution to this question seemed plausible, or even likely.
Indeed, a recent survey on the complexity of learning quantum systems by Anshu and Arunachalam~\cite{aa23} discusses Hamiltonian learning and conclude by asking two questions:

\begin{question}[\cite{aa23}]
\label{question:ham-learning-approx-conditional-indp}
    Can we achieve Hamiltonian learning under the assumption that the Gibbs states satisfy an approximate conditional independence?\footnote{
    Approximate conditional independence is a property of Gibbs states which is proven to hold in 1D and conjectured to hold in general.
}
\end{question}
\begin{question}[\cite{aa23}]
\label{question:low-temp-gibbs-states-pseudo-random}
    Could low temperature Gibbs states be pseudorandom, which would explain the difficulty in finding a time efficient algorithm?
\end{question}

\subsection{Our results}

Surprisingly, we provide a positive resolution to \cref{question:central-open-q}. 
Our main result is a computationally efficient algorithm for Hamiltonian learning that works at all temperatures.
This is a fortunate development since, if learning were truly computationally hard in the low-temperature regime, then we could not understand the behavior of analog quantum simulators in precisely the regimes where they outperform classical simulators~\cite[Section 6.10]{Pre18}.
As a consequence our main result, we also resolve \cref{question:ham-learning-approx-conditional-indp} positively and \cref{question:low-temp-gibbs-states-pseudo-random} negatively. 

\begin{theorem}[Efficiently learning a quantum Hamiltonian, \cref{thm:main-ham-learning-theorem} (informal)]
Given $\eps>0$, $ \beta \geq \beta_c$, for a fixed universal constant $\beta_c>0$, and $\mathfrak{n}$ copies of the Gibbs state of a low-intersection Hamiltonian, $H= \sum_{a\in [\terms]} \lambda_a E_a $, there exists an algorithm that runs in time $\mathfrak{n}^{\mathcal{O}(1)}$ and outputs estimates $\braces{\hat{\lambda}_a}_{a\in[m]}$ such that with probability at least $99/100$, for all $a\in [m]$, $ \Abs{\lambda_a - \hat{\lambda}_a }\leq \eps$, whenever $\mathfrak{n} \ge \poly\parens{m , (1/\eps)^{2^{\mathcal{O}(\beta) }}  } $. 
\end{theorem}

\begin{remark}[On temperature]
For our algorithm, we only need to know an upper bound on $\beta$, as we can consider a Gibbs state at temperature $\beta$ to be a Gibbs state at temperature, say, $2\beta$ with Hamiltonian $H/2$.
Our requirement that $\beta > \beta_c$ is for simplicity, and $\beta_c$ can be any constant bounded away from zero.
In particular, we can take $\beta_c$ to be the temperature at which the high-temperature algorithm of \cite{hkt21} fails; so, when $\beta < \beta_c$, we can simply appeal to~\cite{hkt21} to achieve a sample and time complexity of $\frac{\log \terms}{\beta^2 \epsilon^2}$ and $\frac{ \terms \log \terms}{\beta^2 \epsilon^2}$, respectively.
\end{remark}

\begin{remark}[On locality]
    We do not try to optimize the set of measurements used by the algorithm.
    As written, the quantum part of the algorithm simply estimates $\tr(A \rho)$ for $A$ ranging across Pauli matrices with support size $2^{\mathcal{O}(\beta)} \log(1/\eps)$.
    However, we can observe that we only need those $A$ which are local with respect to the geometry of the Hamiltonian.
    The algorithm is still a global one, though, as the classical part of the algorithm uses observables across the entire space to extract any given term of the Hamiltonian.
    See \cref{rmk:local} for more details.
\end{remark}

As noted by prior work~\cite{aaks20}, Hamiltonian learning is a generalization of the classical and well-studied problem of learning undirected graphical models, specifically parameter learning of these models.
This classical problem requires $\frac{e^{\bigO{\beta}} \terms \log(\terms)}{\beta^2\eps^2}$ time (and there is an algorithm matching the lower bound), so exponential dependence on $\beta$ is necessary~\cite{hkt21}.\footnote{
    It is an interesting open question to improve our doubly exponential dependence to singly exponential.
}
Analogies with the classical setting turn out to be of limited use, since the non-commutativity and non-locality inherent in the quantum setting rules out generalizations of classical ideas.
However, with the classical setting we can identify barriers to designing a time-efficient algorithm.

A key challenge of time-efficient Hamiltonian learning is that we cannot work directly with the partition function.
The only previous approach to low temperature~\cite{aaks20} only used its copies of $\rho$ to estimate $\tr(E_a \rho)$ for all $a \in [\terms]$.
It is known in the classical literature that taking just these estimates and using them to compute the parameters $\lambda_a$ is as hard as computing the partition function~\cite{montanari2015computational}.
To avoid this barrier, we take a richer set of expectations $\tr(P \rho)$ that allows us to reduce learning to a tractable, but fairly involved, optimization problem instead.
Along the way we develop several new tools of independent interest, and ultimately give a semi-definite programming algorithm based on the sum-of-squares hierarchy.
Consequently, we show that sophisticated modern tools in optimization theory lead to a surprising resolution of the Hamiltonian learning problem.

\subsection{Technical overview}

The recipe for quantum Hamiltonian learning introduced by Anshu, Arunachalam, Kuwahara, and Soleimanifar~\cite{aaks20} is based on matching the local marginals of the Gibbs state $\rho$, which we can estimate with copies of $\rho$.
Specifically, for two Hamiltonians $H = \sum \lambda_a E_a$ and $H' = \sum \lambda_a' E_a$ with respective Gibbs states
    $$\rho = \frac{e^{-\beta H}}{\tr(e^{-\beta H} )}\mbox{ and } \rho' = \frac{e^{-\beta H'}}{\tr(e^{-\beta H'} )},$$
they show that $H = H'$ (and so $\lambda_a = \lambda_a'$ for all $a \in [\terms]$) if and only if $\rho$ and $\rho'$ are identical on local marginals i.e. $\tr\Paren{ E_a \rho } = \tr\Paren{ E_a \rho'}$ for all $a\in [m]$~\cite[Proposition 4]{aaks20}.
This does not imply a bound on sample complexity, because with copies of $\rho$, we can only compute $\tr(E_a \rho)$ approximately, with noise introduced from sampling error.
The key structural result of \cite{aaks20} is that this equivalence can be made robust, so that if $H'$ only approximately matches marginals, then the corresponding coefficients $\lambda_a'$ approximately match the true coefficients $\lambda_a$.

However, the last step of this algorithm is to invert the map $\{\lambda_a\}_{a \in [\terms]} \mapsto \{\tr(E_a \rho)\}_{a \in [\terms]}$, which is a computationally hard problem.
Formally, for a classical Hamiltonian,\footnote{
    A classical Hamiltonian is a Hamiltonian that is diagonal, i.e.\ its terms are tensor products of the identity and $\sigma_z$ (\cref{def:paulis}).
    For a classical Hamiltonian, the state $\rho$ is a sample from the Gibbs distribution, and $\tr(E_a \rho)$ is a $\locality$-point correlation function.
} $\tr\Paren{ E_a \rho }$ are \emph{sufficient statistics} of a graphical model and it is known that estimating the parameters of a graphical model from these sufficient statistics is computationally intractable~\cite{montanari2015computational}.
This doesn't mean that the problem is hopeless, but rather that to find a tractable algorithm, we should be looking for the opportunity to use a richer family of statistics.

\paragraph{Designing a new system of constraints.}
We interpret the previous argument as defining and then solving a constraint system in the set of unknowns, 
$\{\lambda_a'\}_{a \in [\terms]}$.  The structural result in~\cite{aaks20} shows that an approximate solution to this system will be close to the true parameters $\lambda_a$.
However, this system is computationally hard to solve. Our starting point is to define a larger set of constraints which $\{\lambda_a\}_{a \in [\terms]}$ must satisfy, which can be verified by measuring expectations of observables slightly less local than the terms $\{E_a\}_{a \in [\terms]}$.
Let $\locals_{\text{local}}$ be the set of Pauli matrices whose support is $K$-local for some large constant $K$.
We begin with the following system of constraints:
\begin{equation}
\label{eqn:intro-exp-constraints}
     \left \{
    \begin{aligned}
      & \forall a \in [\terms] & -1 \leq \lambda_a' &\leq 1 \\
      & & H' &= \sum_{a \in [m] } \lambda_a' \cdot E_a \\
      & \forall P, Q \in \locals_{\text{local}},
    &    \tr\Paren{ Q e^{-\beta H'} P e^{\beta H'}   \rho }  
    &  =    \tr\Paren{ P Q \rho }  \\
    \end{aligned}
  \right \},
\end{equation}
The constraints above are indeed satisfied for the true parameters ($\lambda'=\lambda$) since by assumption $\abs{ \lambda_a}\leq 1$  for all $a \in [m]$ and moreover 
\[
\tr\Paren{ Q e^{-\beta H} P e^{\beta H} \rho } =  \tr\Paren{ Q e^{-\beta H} P e^{\beta H} \frac{e^{-\beta H}}{\tr(e^{-\beta H})} }  = \tr\Paren{ PQ \rho } 
\]
which follows from the cyclic property of the trace. 
Two main challenges remain: Must a solution to this system be close to the true coefficients?
And how can we efficiently solve the system?
Eventually we will derive a convex relaxation for it that is based on
\begin{enumerate}
    \item[(A)] replacing the last constraint in \cref{eqn:intro-exp-constraints} which involves the matrix exponential with low degree polynomial constraints on the indeterminates ($\Set{\lambda_a'}_{a\in[m]}$) instead and
    \item[(B)] showing that any choice of  $\lambda'$ that satisfy the constraints must also approximately match the true coefficients $\lambda$.
\end{enumerate}
In general solving systems of polynomial equations is computationally hard, but because our analysis in (B) will be based on sum-of-squares proofs, there is by now standard machinery for turning it into an efficient algorithm (see~\cref{subsec:sos-framework} for detailed explanation). 

\paragraph{Identifying an equivalence between nested commutators and polynomials.}
Working towards the goal of replacing the term $\tr\Paren{ Q e^{-\beta H'}Pe^{\beta H'} \rho }$ with a low-degree polynomial in the variables $\lambda'$, we begin by recalling the Hadamard formula:\footnote{This can be derived from the Baker--Campbell--Hausdorff formula, \[
    \exp(A) \exp(B) = \exp\parens[\Big]{A + B + \frac{1}{2} [A,B] + \frac{1}{12}\parens{[A,[A,B]] +[B, [B,A]]} + \dots }.
\]}
\begin{equation}
\label{eqn:intro-commutator-series}
e^{-\beta H'} P e^{\beta H'} = \sum_{\ell = 0}^{\infty} (-\beta)^\ell \frac{ [H', P]_\ell }{ \ell! }, 
\end{equation} 
where $[H', P ]_\ell = [H',  [H', \ldots , [H', P]\ldots ]]$ is the $\ell$-th nested commutator. A natural first step is to truncate this series at $d$ terms and observe that 
\begin{equation*}
    \tr\Paren{ Q  \Paren{ \sum_{\ell=0}^d (-\beta)^\ell \frac{[H', P]_\ell}{\ell!} } \rho   } 
\end{equation*}
is a low-degree polynomial in the variables $\lambda'$. For instance, observe for the order-$2$ nested commutator, we have
\begin{equation*}
\begin{split}
    \tr\parens[\Big]{ \bracks[\Big]{\sum_{a \in [m]} \lambda_a' E_a , P}_2 \rho } & = \tr\parens[\Big]{ \bracks[\Big]{\sum_{a \in [m]} \lambda_a' E_a , \sum_{a \in [m]} \lambda_a' \bracks{E_a , P}}  \rho } \\
    & = \sum_{a , b \in [m]} \lambda_a'\lambda_{b}' \tr\Paren{  \left[ E_a , [E_{b}, P] \right] \rho},
\end{split}
\end{equation*}
which is a degree-$2$ polynomial in the $\lambda_i'$ indeterminates.
However, the series in \cref{eqn:intro-commutator-series} only converges quickly when $\beta$ is sufficiently small~\cite{hkt21}, so we cannot use it.\footnote{
    This expansion does converge after $\beta\norm{H}$ terms, but our running time is exponential in the degree, so this would be far too large.
}

Nevertheless, from this observation we can develop a general formalism for constructing polynomial approximations of evolutions of operators.
We observe that in the eigenbasis of $H'$, 
\begin{equation*}
    [H', P]_\ell = P \circ \Set{ \Paren{ \sigma_i - \sigma_j }^\ell }_{ij},
\end{equation*}
where $\Set{ \sigma_i }_{i \in [\dims]}$ are the eigenvalues of $H'$ and $\circ$ denotes the Hadamard product (\cref{def:hadamard}).
Similarly,
\begin{equation*}
    e^{-\beta H'} P e^{\beta H'} = P \circ \Set{ e^{-\beta (\sigma_i - \sigma_j)} }_{ij},
\end{equation*}
and thus we can focus our attention on designing polynomials that approximate the scalar quantity $e^{-\beta \Paren{\sigma_i - \sigma_j} }$ with low-degree polynomials in $\Paren{ \sigma_i - \sigma_j}$. Further, any degree-$d$ polynomial $p(z)= \sum_{\ell=0}^d c_\ell z^\ell$ can be extended to commutators as follows:
\begin{equation*}
    p(H'\mid P ) =  P \circ \Set{ p\Paren{\sigma_i - \sigma_j} }_{ij} = \sum_{\ell=0}^d c_\ell [H'
    ,  P]_\ell. 
\end{equation*}
This allows us to translate between matrix series expansions involving nested commutators and univariate polynomials.
We note that for technical reasons we need to extend our equivalence to nested commutators with two distinct operators $X,Y$ appearing in an arbitrary order, such as $[X, [Y, [X, \ldots ] \ldots ] , E_{a_1}]$ and bi-variate polynomials $p\Paren{x,y}$ (\cref{sec:poly-commutators}).
The translation between bi-variate polynomials and nested commutators incurs additive error depending on $[X,Y]$ due to re-ordering of the $X$ and $Y$ operators as expected (\cref{thm:polynomial-equivalence}). In our full algorithm (\cref{sec:algorithm-and-analysis}) we introduce an additional constraint to drive this additive error to zero. 
Focusing on the scalar polynomial approximation to the exponential function, we now formalize the notion of approximation that we require.

\paragraph{Constructing a new, flat approximation to the exponential.} 
Recall that we want a polynomial such that, working in the eigenbasis of $H'$,
\begin{equation}
    p(H' \mid P) = P \circ \Set{ p\Paren{\sigma_i - \sigma_j} }_{ij} \approx P \circ \Set{ e^{-\beta (\sigma_i - \sigma_j)} }_{ij} = e^{-\beta H'} P e^{\beta H'},
\end{equation}
where ``$\approx$'' denotes an unusual notion of approximation which, for the purposes of this discussion, we can consider to mean that the matrices are close in some norm.
The Taylor series approximation to the exponential would be \cref{eqn:intro-commutator-series}, which we established is too high degree.

Our key insight is that we choose a better polynomial approximation. 
We begin by observing that an operator with small support is approximately band-diagonal in the basis of eigenvectors of $H$, which is a property of local terms proved by Arad, Kuwahara, and Landau~\cite{akl16}.
We state a weak version of this here: let $P$ be a Pauli operator with support size $\bigO{1}$, and let $H = \sum_i d_i v_iv_i^\dagger$ be an eigendecomposition of $H$.
Then, considering $P$ in the eigenbasis of $H$,
\begin{equation}
    \abs{P_{ij}} = \abs{v_i^\dagger P v_j} \leq e^{-\bigOmega{\abs{d_i - d_j}}}. \tag{\cref{lem:akl}}
\end{equation}
A consequence of this is that entries of $P \circ \Set{ p\Paren{\sigma_i - \sigma_j} - e^{-\beta (\sigma_i - \sigma_j)} }_{ij}$, which denote the error of the polynomial approximation, are \emph{weighted inverse exponentially in $\sigma_i - \sigma_j$}.
Therefore, our polynomial approximation need not be equally good for all $\sigma_i - \sigma_j$; rather, our approximation should be $\eps$-good in a small range but is allowed to diverge at a sufficiently slow exponential rate outside that range.
We call this a \emph{flat} approximation.
In particular, given parameters $\beta \geq 0$, $0<\eps, \eta<1$, we construct $p$ such that
\begin{equation}
\label{eqn:intro-poly-approx}
    \begin{cases}
        \Abs{ p(z)  - e^{-\beta z} } \leq \eps & \hspace{0.2in} \textrm{if } z\in [-1, 1 ] \\
        \Abs{ p(z)  } \leq \max\Paren{1, e^{-\beta z} } \cdot e^{\eta \beta \abs{z}} & \hspace{0.2in} \textrm{if } z\notin [-1, 1]
    \end{cases}
\end{equation}
The key difficulty in satisfying the above constraints is satisfying $|p(z)| \leq e^{\eta \beta |z|}$ for $z \geq \beta$, as standard approximations like Taylor series truncations and Chebyshev series truncations fail this condition (\cref{rmk:taylor-fails}).
In \cref{sec:poly-approx}, we explicitly construct a degree-$ \Paren{ 2^{ \mathcal{O}(1/\eta) } \cdot \Paren{ \beta + \log(1/(\eps \eta)) } } $ polynomial that satisfies \cref{eqn:intro-poly-approx}.
This construction is inspired by the iterative ``peeling'' of the exponential used in proofs of Lieb-Robinson bounds~\cite{LiebRobinson72,Hastings10}.
We can write 
\[
e^{-\beta z} = \underbrace{e^{-\beta_c z} \cdots e^{-\beta_c z}}_{\beta/\beta_c}
\]
for a fixed small constant $\beta_c$ and then truncate the Taylor series expansion of $e^{-\beta_c z}$ at different scales for all of the $\beta/\beta_c = O(\beta)$ copies in the product so that the tails of the different truncations don't ``interfere".  

We show that when $p$ is a flat approximation of the above form for some sufficiently small $\eta$, then $Qp(H|P)\rho$ is a good approximation to $Qe^{-\beta H}Pe^{\beta H}\rho$.
In other words, the polynomial approximation is good when $H' = H$ and we right multiply by $\rho$; this is crucial for the polynomial system that we set up next to be feasible.

\paragraph{Formulating a polynomial system.}
We now have all the tools to describe a polynomial system that captures the Hamiltonian learning problem. The constraint system we describe in this section is an informal treatment of the system that appears in \cref{sec:algorithm-and-analysis}, and avoids several technical details.  
We show that using our flat approximation to the exponential, we can obtain a polynomial $p$ such that
$$\tr\Paren{ Q e^{-\beta H } P e^{\beta H} \rho } \approx \tr\Paren{ Q p\Paren{ H \mid P } \rho }.$$

Then, we can re-write \cref{eqn:intro-exp-constraints} as the following polynomial constraint system:
\begin{equation}
\label{eqn:poly-constraint-system}
     \left \{
    \begin{aligned}
    & \forall a \in [\terms] & -1 &\leq \lambda_a' \leq 1 \\
    & & H' &= \sum_{a \in [\terms]} \lambda_a' E_a \\
    & \forall P, Q \in \locals_{\text{local}},
    &    |\tr\Paren{ Q p\Paren{H' \mid P}   \rho }  
    &  -   \tr\Paren{ P Q \rho } | \leq \eps
    \end{aligned}
  \right \},
\end{equation}
and observe that the last constraint encodes a relaxation of the last constraint in \cref{eqn:intro-exp-constraints} and is satisfied when $H' = H$.
Further, all of the constraints are indeed succinctly representable as low-degree polynomials in the indeterminates, $\Set{ \lambda_a' }_{a\in [m]}$, as discussed earlier. 
Finally, the coefficients, such as $\tr\Paren{ [E_a, [E_{b}, P ]] }$, are expectations of the Gibbs state against a slightly larger set of local observables, which are the richer class of test functions we desired.
We can obtain estimates of these expectations through quantum measurements (\cref{sec:accessing-gibbs-states}).
Computing these estimates is the only quantum part of our algorithm, and the rest of the algorithm is entirely classical.

\paragraph{Feasibility of the polynomial system.} 
Recall that to show that the polynomial system in \eqref{eqn:poly-constraint-system} is feasible, we need to argue that 
$$\tr\Paren{ Q e^{-\beta H } P e^{\beta H} \rho } \approx \tr\Paren{ Q p\Paren{ H \mid P } \rho }$$
for all $P,Q$.  Working in the eigenbasis of $H$, let its eigenvalues be $\Set{\sigma_i}_{i \in [2^n]}$.  The key tool that we leverage is from \cite{akl16} (see \cref{lem:akl}) which roughly states that any local term $E$ must be approximately diagonal in the eigenbasis of $H$, with off-diagonal entries decaying as $|E_{ij}| \leq e^{- \Omega(|\sigma_i - \sigma_j|)}$.  Thus, we can decompose the matrices $Q, P$ into two parts \--- parts indexed by $i,j$ where $|\sigma_i - \sigma_j| \leq \beta$ and parts indexed by $i,j$ where $|\sigma_i - \sigma_j| \geq \beta$.  Then we use the fact that $p(x)$ is a good approximation to $e^{-x}$ on $[-\beta, \beta]$ to prove that the error on the first part is small. We then appeal to the exponential decay of the  off-diagonals to argue that the contribution from the second part in both $\tr\Paren{ Q e^{-\beta H } P e^{\beta H} \rho }$ and $tr\Paren{ Q p\Paren{ H \mid P } \rho }$ is small.  Our flat approximation to the exponential is designed to ensure that it does not overwhelm the exponential decay in the off-diagonal entries in $P, Q$ in any regime.

\paragraph{Efficiently optimizing polynomial systems.} 
Now that we know that our polynomial system is feasible, we consider a convex relaxation of this system. In particular, we consider a degree-$d$ sum-of-squares relaxation, which can be efficiently optimized by expressing it as a semi-definite program (see \cref{subsec:sos-framework} for details), with $d=  \log(1/\eps) \cdot  2^{\mathcal{O}(\beta)}$. Since we have $m$ variables and $ 2^{ \mathcal{O}(\beta) }$ constraints, and each constraint is a degree-$d$ polynomial, we can solve the degree-$2d$ sum-of-squares relaxation of \cref{eqn:poly-constraint-system}  
in $m^{\Paren{  \log(1/\eps) \cdot 2^{\mathcal{O}(\beta)} } }$ time. 
The main challenge in analyzing the sum-of-squares relaxation is to show that we can \textit{round} it to estimates $\Set{\tilde{\lambda}'_a }_{a \in [m]}$ such that they are close to the true parameters. 
Here, we adopt the so-called \textit{proofs-to-algorithms} philosophy, where we instead work with the dual object to the sum-of-squares relaxation, namely sum-of-squares proofs (see~\cite{barak2016proofs, fleming2019semialgebraic}, and references therein). This perspective states that if the true parameters are identifiable only using the sum-of-squares proof system, then we immediately obtain an efficient algorithm and we show that we can easily and accurately \textit{round} the solution.

We then provide a proof of identifiability, i.e. for all $a\in [m]$, the inequality $\Paren{ \lambda_a'  - \lambda_a} \leq \eps$ can be derived using the system of polynomial constraints and other basic inequalities that admit sum-of-squares proofs (we refer the reader to \cref{sec:sos-identifiability} for a detailed exposition). At a high level, the proof works by arguing that when $H' - H$ is large, there are witnesses $P, Q$ such that 
\[
|\tr(Qp(H'|P)\rho ) - \tr(Q p(H|P)\rho )|
\]
is large.   Since we know that $H$ is a feasible solution, this would imply that $H'$ cannot be a feasible solution so any feasible solution must have $H' - H$ be small.  The construction of the witnesses relies on an additional property of the polynomial $p$ that we construct, namely that it is strongly monotone (in some appropriate quantitative sense).  

For the identifiability proof, we crucially use an additional important property of local Hamiltonians.
It deals with the quantity $\tr(A^2 \rho)$, where $A = \sum_b \sigma_b P_b$ is a Hermitian linear combination of Pauli matrices with small support.
Thinking of $\rho$ as a distribution, $\tr(A^2 \rho)$ is a second moment term with respect to $\rho$; we can prove this is not much smaller than $\tr(A^2 \id/\dim) = \sum_b \sigma_b^2$, the second moment against the uniform distribution: for some constant $c > 0$,
\begin{align*}
    \tr(A^2 \rho) \geq c^{\bigO{\beta}} \max_b \sigma_b^2 . \tag{\cref{aaks-marginals}}
\end{align*}
Intuitively, this shows that $\rho$ is not close to zero in any local direction.
This was first shown by \cite{aaks20} for quasi-local operators; we adjust their proof to hold for just local operators and give a tighter bound.
We show that we can obtain a slightly weaker statement of this form in the sum-of-squares proof system by formulating it as a quadratic inequality. 
This inequality can be used to remove the dependence on $\rho$ in expressions appearing in the proof; for example, it is used to relate $\tr([H, H']^2 \rho)$ to the size of $[H, H']$ itself.

Finally, we observe that our identifiability proof  does not use the full power of a degree-$2d$ sum-of-squares relaxation and therefore, it should suffice to solve a significantly smaller semi-definite program. We show that we can execute our proof of identifiability by only appealing to a sparse subset of monomials of degree at most $2d$ and invoke a linearization theorem by Steurer and Teigel~\cite{st21} to obtain a final running time of $\poly(m)\cdot \Paren{1/\eps}^{2^{\mathcal{O}(\beta) } }$, as desired.

\subsection{Further related work}

\paragraph{Hamiltonian learning.}
Hamiltonian learning is a broad topic studied both in experimental and theoretical contexts.
This work fits into a large body of algorithmic research about learning properties of quantum states modeling physical systems~\cite{cpfsgb10,aa23}.
Here, we point to a few lines of related work in this field.

Hamiltonian learning often focuses on the real-time evolution setting, where one can allow the system to evolve with respect to $H$, applying the unitary $e^{-\ii Ht}$~\cite{slp11,wgfc14,htfs23}.
Some algorithms consider taking time derivatives (i.e.\ taking $t \to 0$), which are similar to small-$\beta$ algorithms in the Gibbs state setting~\cite{zylb21,hkt21,gcc24}.
There is some research on learning from (zero-temperature) ground states~\cite{QR19}, but the algorithmic work is limited because the ground state of a Hamiltonian need not determine the Hamiltonian.
We study the finite temperature case, which is both the typical temperature at which experiments are run and, in the $\beta \to \infty$ limit, a rich approximation to the much less computationally tractable ground state~\cite{alhambra22,ghls15}.

Though our algorithm is not practical, we use constraint systems that bear some similarity to the ``correlation matrix'' linear constraint systems analyzed heuristically and experimentally in prior work~\cite{BAL19,QR19}.
In fact, our constraint system contains these constraints for technical reasons.
Our work places these works on a rigorous basis, as we prove that, though the linear constraint systems might not uniquely identify the true Hamiltonian, adding more, similar constraints eventually fully constrains the Hamiltonian.

\paragraph{Bounding correlations in Gibbs states.}
Though classical Gibbs states have extremely good locality properties, these become much weaker in the quantum setting.
A series of works aims at bounding the non-locality in quantum Gibbs states with various different measures and in various different regimes~\cite{kb19,kkb20,kaa21}, often with the goal of concluding that simulating or learning these systems can be done time-efficiently.
It is an interesting open problem whether one can extract a new kind of ``locality'' statement from our algorithm, to understand how general our approach is for learning quantum systems.
Our polynomial approximation is inspired by proofs of the Lieb--Robinson bound~\cite{LiebRobinson72,Hastings10}, and can be viewed as a ``low-degree'' form of this bound.
This could be of independent interest.

\paragraph{Parameter learning of graphical models.}
There is a rich body of work on the problem of learning graphical models.
Our setting is that of learning Markov random fields; the literature on this topic focuses on the task of \emph{structure learning}, which in our setting corresponds to learning the terms $\{E_a\}_{a \in [\terms]}$, given the guarantee that they form an (unknown) dual interaction graph with bounded degree~\cite{bms13,Bresler2015,hkm17,km17}.
The problem we consider, learning the parameters with known terms, is easy in the classical setting~\cite[Appendix B]{hkt21}, because classical Gibbs states satisfy the Hammersley--Clifford theorem~\cite{HC71}, also known as the Markov property.
A consequence of the Markov property is that estimating a parameter on a $\locality$-body term can be done by computing conditional marginals on the support of this term.
It is not clear how to generalize this argument to the quantum setting, since the Markov property does not hold for low-temperature quantum Hamiltonians, even approximately~\cite{kkb20}.

\paragraph{The sum-of-squares meta-algorithm.}
The sum-of-squares hierarchy has been used to analyze several problems in quantum information, including best state separation~\cite{doherty2002distinguishing, brandao2011quasipolynomial, barak2012hypercontractivity, barak2017quantum}, optimizing fermionic Hamiltonians~\cite{hastings2022optimizing, hastings2023field}, and a quantum analogue of max-cut~\cite{parekh2021application, watts2023relaxations}. Additionally, the proofs-to-algorithms perspective, introduced in~\cite{barak2015dictionary,barak2016proofs}, has been extensively used to design efficient algorithms for several estimation and learning tasks. In particular, this perspective has led to efficient algorithms for robust learning~\cite{hopkins2018mixture,kothari2018robust, klivans2018efficient, bakshi2020outlier, bakshi2021robust, liu2021settling,bakshi2022robustly} and list-decodable learning~\cite{karmalkar2019list, raghavendra2020list, bakshi2021list}.  

\section{Background} \label{sec:prelim}

Throughout, $\log$ denotes the natural logarithm and $\ii = \sqrt{-1}$.
$\bigO{\cdot}$, $\bigTheta{\cdot}$, and $\bigOmega{\cdot}$ are big O notation, and we use the notation $f \lesssim g$ to mean $f = \bigO{g}$.  The notation $\bigOt{f}$ denotes $\bigO{f\polylog(f)}$.
For a parameter $t$, $\calO_t$ denotes big O notation where $t$ is treated as a constant; the same holds for the notation for polynomial scaling, $\poly_t(\cdot)$.
Everywhere, the binary operator $\cdot$ denotes the usual multiplication.
For a sequence $S \in \{0,1\}^*$, $\len(S)$ denotes its length.

\subsection{Linear algebra}

We work in the Hilbert space $\mathbb{C}^{\dims}$ corresponding to a system of $\qubits$ qubits, $\mathbb{C}^2 \otimes \dots \otimes \mathbb{C}^2$, so that $\dims = 2^\qubits$.
For a matrix $A$, we use $A^\dagger$ to denote its conjugate transpose and $\norm{A}$ to denote its operator norm; for a vector $v$, we use $\norm{v}$ to denote its Euclidean norm.
We will work with this Hilbert space, often considering it in the basis of (tensor products of) Pauli matrices.

\begin{definition}[Pauli matrices] \label{def:paulis}
    The Pauli matrices are the following $2 \times 2$ Hermitian matrices.
    \begin{equation*}
    \sigma_\id = \begin{pmatrix}
        1 & 0 \\ 0 & 1
    \end{pmatrix}, \qquad \sigma_x = \begin{pmatrix}
        0 & 1 \\
        1 & 0
    \end{pmatrix}, \qquad \sigma_y = \begin{pmatrix}
        0 & -\ii \\
        \ii & 0
    \end{pmatrix}, \qquad \sigma_z = \begin{pmatrix}
        1 & 0\\
        0& -1
    \end{pmatrix}.
    \end{equation*}
    These matrices are unitary and (consequently) involutory.
    Further, $\sigma_x \sigma_y = \ii \sigma_z$, $\sigma_y \sigma_z = \ii \sigma_x$, and $\sigma_z \sigma_x = \ii \sigma_y$, so the product of Pauli matrices is a Pauli matrix, possibly up to a factor of $\{\ii, -1, -\ii\}$.
    The non-identity Pauli matrices are traceless.
    We also consider tensor products of Pauli matrices, $P_1 \otimes \dots \otimes P_\qubits$ where $P_i \in \{\sigma_\id, \sigma_{\textup x}, \sigma_{\textup y}, \sigma_{\textup z}\}$ for all $i \in [\qubits]$.
    The set of such products of Pauli matrices, which we denote $\locals$, form an orthogonal basis for the vector space of $2^\qubits \times 2^\qubits$ (complex) Hermitian matrices under the trace inner product.
    The product of two elements of $\locals$ is an element of $\locals$, possibly up to a factor of $\{\ii, -1, -\ii\}$.
\end{definition}

\begin{definition}[Support of an operator]
    For an operator $P \in \mathbb{C}^{\dims \times \dims}$ on a system of $\qubits$ qubits, its \emph{support}, $\supp(P) \subset [\qubits]$ is the subset of qubits that $P$ acts non-trivially on.
    That is, $\supp(P)$ is the minimal set of qubits such that $P$ can be written as $P = O_{\supp(P)} \otimes \id_{[n] \setminus \supp(P)}$ for some operator $O$.
\end{definition}

So, for example, the support of a tensor product of Paulis, $P_1 \otimes \dots \otimes P_\qubits$ are the set of $i \in [\qubits]$ such that $P_i \neq \sigma_\id$.
A central object we consider is (nested) commutators of operators.

\begin{definition}[Commutator]\label{def:commutator}
    Given operators $A , B \in \mathbb{C}^{\dims \times \dims}$, the \emph{commutator} of $A$ and $B$ is defined as $[A,B] = AB - BA$.
    The \emph{nested commutator} of order $\ell$ is defined recursively as $[A, B]_k = [A, [A, B]_{k-1}]$, with $[A, B]_1 = [A, B]$.
\end{definition}

Pauli matrices behave straightforwardly under commutation: the commutator of two Pauli matrices is another Pauli matrix up to a scalar (see \cref{lem:pauli-commutator}).

Finally, we define the following notation to extract the piece of an operator that acts on a particular qubit.

\begin{definition}[Localizing an operator] \label{def:localizing-operator}
    For an operator $O \in \C^{\dims \times \dims}$, define
    \begin{align*}
        O_{(i)} = O - \tr_i(O) \otimes \frac{\id_i}{2},
        = O - \int \diff\mu_i(U) U^\dagger O U,
    \end{align*}
    where $\id_i$ denotes the identity operator on the $i$th qubit, $\tr_i$ denotes the partial trace operation with respect to the $i$th qubit, and $\mu_i$ denotes the Haar measure over the set of unitaries only supported on qubit $i$.
    In other words, $[\cdot]_{(i)}: (\mathbb{C}^{2 \times 2})^{\otimes \qubits}\to (\mathbb{C}^{2 \times 2})^{\otimes \qubits}$ is the linear map on operators that is the identity on every qubit but $i$, and on the $i$th qubit maps $M \mapsto M - \frac{1}{2}\tr(M)\id$ for $M \in \mathbb{C}^{2 \times 2}$.
\end{definition}

For a tensor product of Pauli matrices, $P \in \locals$, $P_{(i)}$ is $P$ when $i \not\in \supp(P)$ and $0$ otherwise.
So, for a linear combination of Paulis, $A = \sum_{P \in \locals} \lambda_P P$, applying this map restricts the sum to Pauli matrices that interact with qubit $i$:
\begin{equation} \label{eq:localizing-local}
    A_{(i)} = \sum_{P : i \in \supp(P)} \lambda_P P
\end{equation}
So, $\abs{\supp(A_{(i)})} \leq (\degree + 1)\locality$.

\begin{definition}[Projector onto eigenspaces]
    For a Hermitian matrix $X$ and an interval $\mathcal{I} \subset \mathbb{R}$, $\Pi_{\mathcal{I}}^{(X)}$ denotes the orthogonal projector onto the subspace spanned by eigenvectors of $X$ with eigenvalues in $\mathcal{I}$.
\end{definition}

We sometimes work in bases where matrices are diagonal, in which case Hadamard products become useful.
\begin{definition}[Hadamard product] \label{def:hadamard}
    For $A, B \in \mathbb{C}^{\dims \times \dims}$, their \emph{Hadamard product}, denoted $A \circ B$, satisfies $[A \circ B]_{ij} = A_{ij} B_{ij}$.
\end{definition}

\subsection{Hamiltonians of interacting systems}
\label{subsec:hamiltonian-of-interacting-system}
We begin by defining a Hamiltonian, which encodes the interaction forces between quantum particles in a physical system.

\begin{definition}[Hamiltonian] \label{def:hamiltonian}
    A \emph{Hamiltonian} is an operator $H \in \mathbb{C}^{\dims \times \dims}$ that we consider as a linear combination of local \emph{terms} $E_a$ with associated \emph{coefficients} $\lambda_a$, $H = \sum_{a=1}^\terms \lambda_a E_a$.
    For normalization, we assume $\norm{E_a} \leq 1$ and $\abs{\lambda_a} \leq 1$.
    This Hamiltonian is \emph{$\locality$-local} if every term $E_a$ satsifies $\abs{\supp(E_a)} \leq \locality$.
\end{definition}

We will only consider Hamiltonians whose terms $E_a$ are distinct, traceless product of Pauli matrices.
Other kinds of local Hamiltonians can be reduced to this setting by expanding out local terms, which are Hermitian matrices, into the basis of products of Pauli matrices.
This preserves the locality of the Hamiltonian and only inflates the number of terms by a factor exponential in the locality, which we think of as constant.
The assumption that $E_a$ are traceless is without loss of generality, since adding a multiple of the identity to the Hamiltonian does not affect its corresponding Gibbs state.

\begin{definition}[Low-intersection Hamiltonian~\cite{hkt21}]
\label{def:low-insersection-ham}
    For a $\locality$-local Hamiltonian $H = \sum \lambda_a E_a$ on a system of $\qubits$ qubits, its \emph{dual interaction graph} $\graph$ is an undirected graph with vertices labeled by $[\terms]$ and an edge between $a, b \in [\terms]$ if and only if
    \begin{align*}
        \supp(E_a) \cap \supp(E_b) \neq \varnothing.
    \end{align*}
    Let $\degree$ denote the maximum degree of this graph.
    We call $H$ \emph{low-intersection}\footnote{Also called a ``low-interaction''~\cite{htfs23} or ``sparsely interacting''~\cite{gcc24} Hamiltonian.} if $\locality$ and $\degree$ are constant.
\end{definition}

Throughout, we assume that our Hamiltonian is $\locality$-local and has a dual interaction graph of degree $\degree$.  We will treat $\locality$ and $\degree$ as constants throughout the paper.  This encompasses most Hamiltonians discussed in the literature, including ``geometrically local'' Hamiltonians, i.e.\ Hamiltonians whose qubits are thought of as being on a constant-dimensional lattice like $\mathbb{Z}^3$ and whose terms that are spatially local with respect to the lattice.
This class is more general than geometrically local Hamiltonians, as it can extend to qubits where locality is dictated by an expander graph.

We will be considering operators other than the terms, so we define a notion of locality with respect to the dual interaction graph that generalizes to such operators.
Concretely, what we need about $\ell$-local operators for $\ell \geq k$ is that (1) they contain the Hamiltonian terms $E_a$, (2) the dimension of the subspace spanned $O(\qubits)$, and (3) they contain nested commutators involving $k$-local operators.

\begin{definition}[Local operator with respect to the dual interaction graph] \label{def:dual-interaction-graph}
    Consider a $\locality$-local Hamiltonian $H = \sum_{a=1}^\terms \lambda_a E_a$ with dual interaction graph $\graph$.
    For $P$ a tensor product of Paulis, we say it is $\locality\ell$-$\graph$-local if there is some $S \subset [\terms]$ of size $\ell$ such that $\supp(P) \subset \cup_{a \in S} \supp(E_a)$ and $S$ is connected in $\graph$.

    We define $\locals_{\locality\ell}$ to be the set of $\locality\ell$-$\graph$-local Pauli matrices.
    Generally, we call an operator $M \in \C^{\dims \times \dims}$ $\locality \ell$-$\graph$-local if it is equal to a linear combination of elements in $\locals_{\locality\ell}$.
\end{definition}

By this definition, if $E_a$ and $E_b$ are terms of a Hamiltonian, then $E_b$ is $\locality$-$\graph$-local, and $[E_a, E_b]$ is $2\locality$-$\graph$-local.
We state the form of nested commutators of Pauli matrices below.

\begin{lemma} \label{lem:pauli-commutator}
    Let $P_1 \in \locals_{k_1}, \dots ,P_{a} \in \locals_{k_a}$, and $Q \in \locals_{\ell}$ (with respect to some background dual interaction graph $\graph$).
    Then the nested commutator is either zero or also a $\graph$-local Pauli matrix,
    \begin{align*}
        \frac{\ii^a}{2^a} [P_1,  [P_2, \dots [P_a, Q] \dots ]] \in \locals_{k_1 + \dots + k_a + \ell},
    \end{align*}
    possibly up to a factor of $\pm 1$ where $\ii = \sqrt{-1}$.
    
\end{lemma}
\begin{proof}
Let $a = 1$.
By properties given in \cref{def:paulis}, for $P \in \locals_{k_1}$ and $Q \in \locals_\ell$, $PQ$ is a tensor product of Pauli matrices, up to fourth root of unity.
Consequently, $[P, Q] = PQ - (PQ)^\dagger$ is either zero (if $PQ$ is Hermitian) or $2PQ$ (if $\ii PQ$ is Hermitian).
Further, $\supp([P, Q]) \subset \supp(PQ) \subset \supp(P) \cup \supp(Q)$, which shows that $[P, Q]$ is $(k_1 + \ell)$-$\graph$-local.
This proves the lemma for $a = 1$; the general statement follows by iterating this case.
\end{proof}

\subsection{Properties of local operators on quantum systems}

For a Hamiltonian describing a quantum system, we consider getting access to the associated state attained by thermalizing it at a particular inverse temperature $\beta$.
This is known as a Gibbs state.

\begin{definition}[Gibbs state]
    The \emph{Gibbs state} of the Hamiltonian $H = \sum \lambda_a E_a$ at inverse temperature $\beta > 0$ is given by
    \begin{align}
        \rho
        = \frac{\exp(-\beta H)}{\Tr \exp (-\beta H)}
        = \exp\biggl(-\beta \sum_a \lambda_a E_a\biggr) \biggr/ \Tr \exp\biggl(-\beta \sum_a \lambda_a E_a\biggr).
    \end{align}
\end{definition}

A useful piece of intuition is to think of $\rho$ as a distribution over its eigenspaces with probability proportional to the eigenvalue.
In that sense, we can think about expectations and variances against this distribution.

\begin{definition}[Expectation against the Gibbs state]
    For a Gibbs state $\rho$ of a Hamiltonian $H$ and an operator $A \in \mathbb{C}^{\dims \times \dims}$, we define $\angles{A} = \tr(A \rho)$.
\end{definition}

A key result in the prior work of Anshu, Arunachalam, Kuwahara, and Soleimanifar gives a lower bound on the variance of a local operator with respect to the energy distribution defined by the Gibbs state.
The authors prove this for the more general class of quasi-local operators, but we only need it for local operators, for which the result can be tightened.
We give this tighter version below; its proof is in \cref{app:aaks}.

\begin{theorem}[{\cite[Theorem 33]{aaks20}}] \label{aaks-marginals}
    Let $H$ be a $\locality$-local Hamiltonian with dual interaction graph $\graph$ with max degree $\degree$.
    Let $A = \sum_b \sigma_b P_b$ be a $\locality$'-local operator where $P_b$ are products of Pauli matrices and $-1 \leq \sigma_b \leq 1$ and whose dual interaction graph has max degree $\degree'$.
    For $\beta > 0$, let $\rho$ be the corresponding Gibbs state of $H$.
    Then
    \begin{align*}
        \angles{A^2} = \tr(A^2 \rho)
        \geq \max_{i \in [\qubits]} \parens[\Big]{c\tr(A_{(i)}^2/\dims)}^{6 + c'\beta}.
    \end{align*}
    Here, $c$ and $c'$ are positive constants that depend on $\locality$, $\degree$, $\locality'$, and $\degree'$.
\end{theorem}

After some manipulation we can conclude that, if a local operator has small variance with respect to $\rho$, then the operator itself must be close to zero.
Some may recognize this as the quantum analogue of the statement that bounded-degree graphical models have local marginals bounded away from zero (see e.g. \cite{Bresler2015}).

\begin{corollary}[``No small local marginals''] \label{lem:large-marginals}
    Let $H$ be a $\locality$-local Hamiltonian with dual interaction graph $\graph$ with max degree $\degree$.
    Let $A = \sum_{P \in \locals} \sigma_P P$ be a sum over Paulis with support at most $\locality'$ and with a dual interaction graph $\graph'$ with max degree $\degree'$, with coefficients $\sigma_P \in \R$.
    Then
    \begin{equation*}
        \tr(A^2 \rho) \geq \exp(- c_{\locality,\degree,\locality', \degree'} - c_{\locality,\degree,\locality', \degree'}' \beta ) \max_{Q \in \locals} \sigma_Q^2 
    \end{equation*}
    where $c_{\locality,\degree, \locality',\degree'}, c'_{\locality,\degree, \locality',\degree'}$ are constants depending only on $\locality, \degree, \locality', \degree'$.
\end{corollary}

\begin{proof}
Let $g^2 = \max_{Q \in \locals_\ell} \sigma_Q^2$.
\begin{align*}
    \tr(X^2 \rho) &= g^2 \tr((X/g)^2 \rho) \\
    &\geq g^2 \max_{i \in [\qubits]} \parens[\Big]{c\tr((X_{(i)}/g)^2/\dims)}^{6 + c'\beta} \tag*{by \cref{aaks-marginals}}\\
    &= g^2 \max_{i \in [\qubits]} \parens[\Big]{c\tr\parens[\Big]{\parens[\Big]{\sum_{\substack{P \in \locals_\ell \\ i \in \supp(P)}} \frac{\sigma_P}{g} P}^2/\dims}}^{6 + c'\beta} \tag*{by \cref{eq:localizing-local}} \\
    &= g^2 \max_{i \in [\qubits]} \parens[\Big]{c\sum_{\substack{P \in \locals_\ell \\ i \in \supp(P)}} \frac{\sigma_P^2}{g^2}}^{6 + c'\beta} \tag*{by \cref{def:paulis}}\\
    &\geq g^2 c^{6 + c'\beta}.
\end{align*}
This gives the desired statement.
\end{proof}

Another key statement we use about local operators is that a local operator is approximately diagonal when considered in the basis of eigenvectors of another local operator.

\begin{theorem}[{\cite[Theorem 2.1]{akl16}}]
    Let $H = \sum_{S \subset [\qubits]} h_S$ be a Hamiltonian where all its terms $h_S$ are positive semi-definite, supported on at most $\locality$ qubits, and the terms interacting with any particular site have bounded norm: $\sum_{S : i \in S} \norm{h_S} \leq \onespin$ for all $i \in [\qubits]$.
    Let $A$ be an operator, and define $R = \sum_{X \in \calC} \norm{h_X}$, where $\calC = \{S \subset [\qubits] \mid [h_S, A] \neq 0\}$ is the set of terms that don't commute with $A$.
    Then
    \begin{equation*}
        \norm{\Pi_{[\sigma, \infty)}^{(H)} A \Pi_{[0, \varsigma]}^{(H)}} \leq \norm{A} e^{-\frac{1}{2\onespin\locality}(\sigma - \varsigma - 2R)}.
    \end{equation*}
\end{theorem}

The assumption that the $h_S$ are PSD can be removed: if the $h_S$ are not PSD, we can apply the theorem to $h_S + \id \norm{h_S}$, which only affects the final bound by inflating $\onespin$ and $R$ by a factor of two.
Note that this just leads to shifting $H$ by a factor of the identity, and so only shifts the spectrum.

\begin{corollary}\label{lem:akl}
    Let $H = \sum_{S \subset [\qubits]} h_S$ be a Hamiltonian where all its terms $h_S$ are supported on at most $\locality$ qubits, and the terms interacting with any particular site has bounded norm: $\sum_{S : i \in S} \norm{h_S} \leq \onespin$ for all $i \in [\qubits]$.
    Let $A$ be an operator, and define $R = \sum_{X \in \calC} \norm{h_X}$, where $\calC = \{S \subset [\qubits] \mid [h_S, A] \neq 0\}$ is the set of terms that don't commute with $A$.
    Then
    \begin{equation*}
        \norm{\Pi_{[\sigma, \infty)}^{(H)} A \Pi_{[-\infty, \varsigma]}^{(H)}} \leq \norm{A} e^{-\frac{1}{4\onespin\locality}(\sigma - \varsigma - 4R)}.
    \end{equation*}
\end{corollary}

When $A$ satisfies $\abs{\supp(A)} \leq \locality'$, we can take $R$ to be the number of terms that intersect with the support of $A$.
In our setup, where the terms are Pauli matrices with a dual interaction graph with maximum degree $\degree$, we can take $\onespin = \degree + 1$.
If the support of $A$ is contained in the support of a term, we can take $R = \degree+1$; otherwise, we can take $R = \locality' \degree$.

\subsection{Bounds on nested commutators}

One key property about commutators $[A,B]$ is that they compose well when the inputs $A,B$ are a sum of local terms (recall \cref{lem:pauli-commutator}).  In this section, we will make these composition properties more precise when we have a nested commutator of the form
\[
[H_1, [H_2 , [ \dots [H_\ell , A] \dots ]]],
\]
where $H_1, \dots , H_\ell$ are local operators and $A$ has small support.
We use a cluster expansion argument to bound these terms.

\begin{definition}\label{def:cluster}
Let $M_1, \dots , M_\ell \in \mathbb{C}^{\dims \times \dims}$ be operators.  We say that the ordered $\ell$-tuple $(M_1, \dots,  M_\ell)$ forms a cluster if, for all $a$, the support of $M_{a+1}$ has nonempty intersection with $\supp(M_1) \cup \dots \cup \supp(M_{a})$.
\end{definition}

\begin{lemma}\label{lem:cluster-expansion}
Let $\mathcal{E} \subset \locals$ be a set of Pauli terms where every $P \in \mathcal{E}$ satisfies $\supp(P) \leq \locality$ and the dual interaction graph associated with $\mathcal{E}$ has max degree $\degree$.
Let $H_1, \dots , H_\ell$ be linear combinations of elements of $\mathcal{E}$, written as $H_i = \sum_{P \in \mathcal{E}} \lambda_{i,P} P$ for all $i$.  Let $A \in \mathcal{E}$.  Then we can write $[H_1, [H_2 , [ \dots [H_{\ell},  A] \dots ]]]$ in the form
\[
2^\ell  \sum_{\substack{ P_1, P_2, \dots , P_\ell \in \mathcal{E} \\ (A, P_{\ell}, \dots , P_1) \text{ is a cluster}} } c_{P_1, \dots , P_\ell} Q_{P_1, \dots , P_\ell} \prod_{j = 1}^\ell \lambda_{j,P_j},
\] 
where the constants $c_{P_1, \dots , P_m} \in \{0, \pm 1, \pm i\}$ and 
\begin{enumerate}[label=(\alph*)]
    \item $Q_{P_1, \dots , P_\ell} \in \locals_{(\ell+1)k}$ (as defined in \cref{def:dual-interaction-graph}) and is distance $\ell$ from $A$ in $\graph$-distance;
    \item The number of terms in the sum is at most  $\ell! (\degree+1)^\ell$.
\end{enumerate}
\end{lemma}
\begin{proof}
We can expand the nested commutator $[H_1, [H_2 , [ \dots [H_\ell , A] \dots ]]]$, using that the commutator is bilinear, into a sum of individual terms
\[
\sum_{P_1, \dots , P_\ell}  [P_1, [P_2 , [ \dots [P_\ell , A] \dots ]]] \prod_{a=1}^\ell \lambda_{a,P_a} \,.
\]
Now we argue about which of the terms in the sum above are nonzero.  Note that for the commutator to be nonzero, $P_\ell$ must intersect the support of $A$, $P_{\ell-1}$ must intersect $\supp(P_\ell) \cup \supp(A)$, and so on.  This condition is equivalent to $A, P_{\ell}, \dots , P_1$ being a cluster.  Now by \cref{lem:pauli-commutator}, we have
\[
 [P_1, [P_2 , [ \dots [P_\ell , A] \dots ]]] \in 2^\ell c_{P_1, \dots , P_\ell} \locals_{\ell+1}
\]
for some $c_{P_1, \dots , P_\ell} \in \{ 0, \pm 1, \pm i \}$.  Because $(A, P_\ell, \dots, P_1)$ is a cluster, we can also conclude that the elements of $\locals_{\ell + 1}$ that appear in the sum have distance $\leq \ell$ from $A$.  Now it remains to count the number of clusters.  When choosing $P_a$, we have $a$ choices for which of $(A, P_1, \ldots, P_{a-1})$, to intersect with, and $\degree+1$ choices for elements of $\mathcal{E}$ that intersect with that element of $\mathcal{E}$.  This gives an upper bound on the total number of clusters of
\[
(\degree+1) \cdot 2(\degree+1) \cdots \ell(\degree+1) = \ell! (\degree+1)^{\ell}.
\]
This completes the proof.
\end{proof}

We will also need the following lemma from \cite{hkt21} that counts the number of distinct multisets of terms that form a cluster (under some ordering).

\begin{lemma}[{\cite[Proposition 3.6]{hkt21}}]\label{lem:count-monomials}
Consider a set of terms $\mathcal{E} \subset \locals$ with dual interaction graph $\graph$ with max degree $\degree \geq 2$.
For a fixed $E_{a_1} \in \mathcal{E}$, the number of multisets of terms $E_{a_2}, \dots , E_{a_\ell} \in \mathcal{E}$ such that there is some ordering $\pi$ where $(E_{a_{\pi(1)}}, \dots , E_{a_{\pi(\ell)}})$ forms a cluster is at most $e\degree(1+e(\degree-1))^{\ell-1} \leq (3\degree)^\ell$.
\end{lemma}

As a consequence, we can bound the number of distinct elements in $\locals_{k\ell}$.

\begin{corollary}\label{coro:number-of-local-paulis}
We have $|\locals_{k\ell}| \leq m (10^{k}\degree)^{\ell }$.  Furthermore, the number of elements of $\locals_{k\ell}$ that have support intersecting with a fixed term $P \in \locals_k$ is at most $(10^{k}\degree)^{\ell + 1}$. 
\end{corollary}

\subsection{Sums-of-squares polynomials}

We will need some preliminaries about polynomials and the sum-of-squares (SoS) framework.  First, we introduce a notion of polynomials with sum-of-squares representations with bounded coefficients.  We maintain coefficient bounds because later, certain sampling and approximation errors are multiplied by these coefficients.  The general SoS framework is not concerned with such tight bounds on the coefficients and thus we need to make a few definition outside of the general framework.

\begin{definition}[Sum-of-squares polynomial]
    A polynomial $p \in \R[x_1, \dots , x_m ]$ is a \emph{sum-of-squares polynomial} if we can write $p = q_1^2 + \dots + q_k^2$ for some polynomials $q_1, \dots , q_k \in \R[x_1, \dots , x_m ]$.
\end{definition}

We sometimes abbreviate sum-of-squares as SoS.

\begin{definition}[Bounded polynomial]
\label{def:bound-polynomial}
\label{def:sos-bounded-polynomial}
    We say a polynomial $p(x_1, \dots , x_m) \in \R[x_1, \dots , x_m ]$ is \emph{$(d,C)$-bounded} if  
    \begin{enumerate}[label=(\alph*)]
        \item $p$ has degree at most $d$, and
        \item for each monomial of degree $d' \leq d$ in $p$, its coefficient has magnitude at most $C/(d'!)$. 
    \end{enumerate} 
    We say $p$ is a \emph{$(k,d,C)$-bounded sum-of-squares polynomial} if $p$ is a sum-of-squares polynomial, $p = q_1^2 + \dots + q_k^2$, and each of the $q_i$'s are $(d,C)$-bounded.
\end{definition}

\begin{claim}[Composition of bounded SoS polynomials]\label{claim:basic-composition-properties}
Let $p_1(x_1, x_2)$ be a $(k_1, d_1, C_1)$-bounded SoS polynomial and $p_2(x_1, x_2)$ be a $(k_2, d_2, C_2)$-bounded SoS polynomial.  Then
\begin{enumerate}[label=(\alph*)]
    \item $p_1 + p_2$ is a $(k_1 + k_2,\,\max(d_1, d_2),\,\max(C_1, C_2))$-bounded SoS polynomial;
    \item $p_1p_2$ is a $(k_1k_2,\,d_1 + d_2,\,(d_1 + d_2+1)2^{d_1 + d_2}C_1C_2)$-bounded SoS polynomial;
    \item For any $t \in [0,1]$, $p_1((1 - t)x_1 + ty_1, (1 - t)x_2+ ty_2) $ is a $(k_1, d_1, C_1)$-bounded SoS polynomial in $x_1,y_1,x_2,y_2$.
\end{enumerate}
\end{claim}
\begin{proof}
For this proof, we let $[x_1^ix_2^j]p(x_1, x_2)$ denote the coefficient corresponding to $x_1^ix_2^j$ in $p$.

The first statement is obvious as we can simply combine the two representations as sums of squares.  The second also follows immediately by taking the two representations as sums of squares and expanding the product.
For this, we use that, when $r(x_1, x_2)$ is $(d_1, C_1)$-bounded and $s(x_1, x_2)$ is $(d_2, C_2)$-bounded, $rs$ is $(d_1+d_2, (d_1 + d_2 + 1)2^{d_1 + d_2}C_1C_2)$-bounded.
\begin{multline*}
    [x_1^ix_2^j](rs) = \sum_{\substack{0 \leq i' \leq i \\ 0 \leq j' \leq j}} \abs{[x_1^{i'}x_2^{j'}]r}\abs{[x_1^{i-i'}x_2^{j-j'}]s} \\
    \leq \sum_{\substack{0 \leq i' \leq i \\ 0 \leq j' \leq j}} \frac{C_1}{(i'+j')!}\frac{C_2}{(i+j-i' - j')!}
    = \frac{C_1 C_2}{(i+j)!}\sum_{\substack{0 \leq i' \leq i \\ 0 \leq j' \leq j}} \binom{i+j}{i' + j'}
    \leq \frac{C_1 C_2}{(i+j)!}(i+j+1)2^{i+j}
\end{multline*}
For the final statement, write
$p_1(x_1, x_2) = q_1(x_1,x_2)^2 + \dots + q_{k_1}(x_1,x_2)^2$.  Now we simply substitute in the change of variables.  A coefficient of the new polynomial is
\begin{multline*}
    [x_1^{i_1}y_1^{j_1}x_2^{i_2}y_2^{j_2}]q_{\ell}((1 - t)x_1 + ty_1, (1 - t)x_2+ ty_2) \\
    = 
    [x_1^{i_1 + j_1}x_2^{i_2 + j_2}]q_\ell(x_1, x_2) \cdot \binom{i_1 + j_1}{i_1} (1 - t)^{i_1} t^{j_1} \binom{i_2 + j_2}{i_2} (1 - t)^{i_2} t^{j_2} \leq [x_1^{i_1 + j_1}x_2^{i_2 + j_2}]q_\ell(x_1, x_2)
\end{multline*}
and this shows that after the change of variables, $p_1((1 - t)x_1 + ty_1, (1 - t)x_2+ ty_2) $ is still a $(k_1, d_1, C_1)$-bounded SoS polynomial.
\end{proof}

\subsubsection{The sum-of-squares framework}
\label{subsec:sos-framework}

We now provide an overview of the sum-of-squares proof system.
We closely follow the exposition as it appears in the lecture notes of Barak~\cite{barak2016proofs}.   

\paragraph{Pseudo-Distributions.}
A discrete probability distribution over $\R^m$ is defined by its probability mass function, $D\from \R^m \to \R$, which must satisfy $\sum_{x \in \mathrm{supp}(D)} D(x) = 1$ and $D \geq 0$.
We extend this definition by relaxing the non-negativity constraint to merely requiring that $D$ passes certain low-degree non-negativity tests.
We call the resulting object a pseudo-distribution.

\begin{definition}[Pseudo-distribution]
A \emph{degree-$\ell$ pseudo-distribution} is a finitely-supported function $D:\R^m \rightarrow \R$ such that $\sum_{x} D(x) = 1$ and $\sum_{x} D(x) p(x)^2 \geq 0$ for every polynomial $p$ of degree at most $\ell/2$, where the summation is over all $x$ in the support of $D$.
\end{definition}
Next, we define the related notion of pseudo-expectation.
\begin{definition}[Pseudo-expectation]
The \emph{pseudo-expectation} of a function $f$ on $\R^m$ with respect to a pseudo-distribution $\mu$, denoted by $\pexpecf{\mu(x)}{f(x)}$,  is defined as
\begin{equation*}
    \pexpecf{\mu(x)}{f(x)} = \sum_{x} \mu(x) f(x).
\end{equation*}
\end{definition}
We use the notation $\pexpecf{\mu(x)}{(1,x_1, x_2,\ldots, x_m)^{\otimes \ell}}$ to denote the degree-$\ell$ moment tensor of the pseudo-distribution $\mu$.
In particular, each entry in the moment tensor corresponds to the pseudo-expectation of a monomial of degree at most $\ell$ in $x$. 

\begin{definition}[Constrained pseudo-distribution]
\label{def:constrained-pseudo-distributions}
Let $\calA = \Set{ p_1\geq 0 , p_2\geq0 , \dots, p_r\geq 0}$ be a system of $r$ polynomial inequality constraints of degree at most $d$ in $m$ variables.
Let $\mu$ be a degree-$\ell$ pseudo-distribution over $\mathbb{R}^m$.
We say that $\mu$ \emph{satisfies} $\calA$ at degree $\ell \ge1$ if for every subset $\calS \subset [r]$ and every sum-of-squares polynomial $q$ such that $\deg(q) + \sum_{i \in \calS } \max\Paren{ \deg(p_i), d} \leq \ell$, $\pexpecf{\mu}{ q \prod_{i \in \calS} p_i } \geq 0$.
Further, we say that $\mu$ \emph{approximately satisfies} the system of constraints $\calA$ if the above inequalities are satisfied up to additive error $\pexpecf{\mu}{ q \prod_{i \in \calS} p_i } \geq -2^{-n^{\ell} } \norm{q} \prod_{i \in \calS} \norm{p_i}$, where $\norm{\cdot}$ denotes the Euclidean norm of the coefficients of the polynomial, represented in the monomial basis.  
\end{definition}

Crucially, there's an efficient separation oracle for moment tensors of constrained pseudo-distributions. 
Below gives the unconstrained statement; the constraint statement follows analogously.

\begin{fact}[\cite{shor1987approach, nesterov2000squared, parrilo2000structured, grigoriev2001complexity}]
    \label{fact:sos-separation-efficient}
    For any $m,\ell \in \N$, the following convex set has a $m^{\bigO{\ell}}$-time weak separation oracle, in the sense of \cite{grotschel1981ellipsoid}:\footnote{
        A separation oracle of a convex set $S \subset \R^M$ is an algorithm that can decide whether a vector $v \in \R^M$ is in the set, and if not, provide a hyperplane between $v$ and $S$.
        Roughly, a weak separation oracle is a separation oracle that allows for some $\eps$ slack in this decision.
    }:
    \begin{equation*}
        \Set{  \pexpecf{\mu(x)} { (1,x_1, x_2, \ldots, x_m)^{\otimes \ell } } \Big\vert \text{ $\mu$ is a degree-$\ell$ pseudo-distribution over $\R^m$}}
    \end{equation*}
\end{fact}
This fact, together with the equivalence of weak separation and optimization \cite{grotschel1981ellipsoid} forms the basis of the sum-of-squares algorithm, as it allows us to efficiently approximately optimize over pseudo-distributions. 

Given a system of polynomial constraints, denoted by $ \calA$, we say that it is \emph{explicitly bounded} if it contains a constraint of the form $\{ \|x\|^2 \leq 1\}$. Then, the following fact follows from  \cref{fact:sos-separation-efficient} and \cite{grotschel1981ellipsoid}:

\begin{theorem}[Efficient optimization over pseudo-distributions]
    \label{fact:eff-pseudo-distribution}
There exists an $(m+r)^{O(\ell)} $-time algorithm that, given any explicitly bounded and satisfiable system $ \calA$ of $r$ polynomial constraints in $m$ variables, outputs a degree-$\ell$ pseudo-distribution that satisfies $ \calA$ approximately, in the sense of~\cref{def:constrained-pseudo-distributions}.\footnote{
    Here, we assume that the bit complexity of the constraints in $ \calA$ is $(m+t)^{O(1)}$.
}
\end{theorem}
\begin{remark}[Bit complexity and approximate satisfaction]
    \label{rmk:tedium}
    We will eventually apply this result to a constraint system that can be defined with numbers with $\log(t)$ bits, where $t$ is the sample complexity of the algorithm (scaling polynomially with $\terms$).
    Consequently, we can run this algorithm efficiently, and the errors incurred here (which is exponentially small in $\qubits$) can be thought of as a ``machine precision'' error, and is dominated by the sampling errors incurred elsewhere.
    We can therefore safely ignore precision issues in the rest of our proof.

    The pseudo-distribution $D$ found will satisfy $\calA$ only approximately, but provided the bit complexity of the sum-of-squares proof of $\calA \sststile{r'}{} B$, i.e.\ the number of bits required to write down the proof, is bounded by $\terms^{\bigO{\ell}}$ (assuming that all numbers in the input have bit complexity $\terms^{\bigO{1}}$), we can compute to sufficiently good error in polynomial time that the soundness will hold approximately.
    All of our sum-of-squares proofs will have this bit complexity.
\end{remark}

We now state some standard facts for pseudo-distributions, which extend facts that hold for standard probability distributions.
These can be found in the prior works listed above.

\begin{fact}[Cauchy--Schwarz for pseudo-distributions]
Let $f,g$ be polynomials of degree at most $d$ in the variables $x \in \R^m$. Then, for any degree-$d$ pseudo-distribution $\mu$,
$\pexpecf{\mu}{fg}  \leq \sqrt{ \pexpecf{\mu}{ f^2} } \cdot \sqrt{ \pexpecf{\mu}{ g^2} }$.
    \label{fact:pseudo-expectation-cauchy-schwarz}
\end{fact}

\begin{fact}[Hölder's inequality for pseudo-distributions] \label{fact:pseudo-expectation-holder}
Let $f,g$ be polynomials of degree at most $d$ in the variables $x \in \R^m$. 
Fix $t \in \N$. Then, for any degree-$dt$ pseudo-distribution $\mu$,
\begin{equation*}
    \pexpec{\mu}{ f^{t-1}  g} \leq \Paren{ \pexpec{\mu}{ f^t }  }^{\frac{t-1}{t}} \cdot  \Paren{  \pexpec{\mu }{ g^t }  }^{\frac{1}{t}}.  
\end{equation*}
In particular, when $t$ is even,
$\pexpec{\mu}{f}^t \leq \pexpec{\mu}{ f^t }$.
\end{fact}

\paragraph{Sum-of-squares proofs.}
Up to minor technical details, our algorithm is to set up a polynomial constraint system and then call \cref{fact:eff-pseudo-distribution} to get a pseudo-distribution $\mu$ over the indeterminates $\{\lambda_i' \mid i \in [\terms]\}$ which approximately satisfies the constraints.
With this pseudo-distribution, we will output $\pexpec{\mu}{\lambda'}$ as our estimated Hamiltonian coefficients.
To show that these estimates are close to the true coefficients $\lambda$, we use that, under $\mu$, the polynomial constraints hold.
That is, for a constraint $p \geq 0$, we have $\pexpec{\mu}{p} \geq 0$.
If we can use these constraints to derive that $\abs{\pexpec{\mu}{\lambda_a'} - \lambda_a} \leq \eps$ for every $a \in [\terms]$, then the algorithm is correct.
So, we provide such a proof: this proof will be in the sum-of-squares proof system.

Let $f_1, f_2, \ldots, f_r$ and $g$ be multivariate polynomials in the indeterminates $x \in \R^m$.
Given the constraints $\{f_1 \geq 0, \ldots, f_r \geq 0\}$, a \emph{sum-of-squares proof} of the identity $\{g \geq 0\}$ is a set of polynomials $\{p_S\}_{S \subseteq [r]}$ such that
\begin{equation*}
    g = \sum_{S \subseteq [r]} p^2_S \cdot \prod_{i \in S} f_i.
\end{equation*}
As its name suggests, the existence of such an SoS proof shows that if the constraints $\{f_i \geq 0 \mid i \in [r]\}$ are satisfied, then the identity $g \geq 0$ is also satisfied.
We say that this SoS proof has \emph{degree $\ell$} if for every set $S \subseteq [r]$, the polynomial $p^2_S \Pi_{i \in S} f_i$ has degree at most $\ell$.
If there is a degree-$\ell$ SoS proof that $\{f_i \geq 0 \mid i \in [r]\}$ implies $\{g \geq 0\}$, we write
\begin{equation}
    \{f_i \geq 0 \mid i \in [r]\} \sststile{\ell}{x}\{g \geq 0\}.
\end{equation}
We will sometimes drop the indeterminate in $\sststile{\ell}{x}$ when this causes no confusion.
For all polynomials $f,g\colon\R^m \to \R$ and for all coordinate-wise polynomials $F\colon \R^m \to \R^{m_F}$, $G\colon \R^m \to \R^{m_G}$, $H\colon \R^{m_H} \to \R^m$, we have the following inference rules.\footnote{
    This notation should be read in the following way: given the proofs above the bar line, we can derive the proof below the bar line.
}
\begin{figure}[h]
    \renewcommand{\arraystretch}{1.5}
    \begin{center}
    \begin{tabular}{c c}
        Addition Rule & Multiplication Rule \\
        \vspace{1em}
        $\displaystyle\frac{ \calA \sststile{\ell}{} \{f \geq 0, g \geq 0 \} } { \calA \sststile{\ell}{} \{f + g \geq 0\}}$
        & $\displaystyle\frac{ \calA \sststile{\ell}{} \{f \geq 0\},\quad \calA \sststile{\ell'}{} \{g \geq 0\}} { \calA \sststile{\ell+\ell'}{} \{f \cdot g \geq 0\}}$ \\
        Transitivity Rule & Substitution Rule \\
        $\displaystyle\frac{ \calA \sststile{\ell}{}  \calB,\quad \calB \sststile{\ell'}{} C}{ \calA \sststile{\ell \cdot \ell'}{} C}$
        & $\displaystyle\frac{\{F \geq 0\} \sststile{\ell}{} \{G \geq 0\}}{\{F(H) \geq 0\} \sststile{\ell \cdot \deg(H)} {} \{G(H) \geq 0\}}$
    \end{tabular}
    \end{center}
    \renewcommand{\arraystretch}{1}
\end{figure}

Sum-of-squares proofs allow us to deduce properties of pseudo-distributions that satisfy some constraints.
\begin{fact}[Soundness]
  \label{fact:sos-soundness}
  Let $\mu$ be a degree-$\ell$ pseudo-distribution.
  If $\mu$ is consistent with the set of degree-$d_A$ polynomial constraints $\calA$, denoted $\mu \sdtstile{d_A}{} \calA$, and there is a degree-$d_B$ sum-of-squares proof that $ \calA \sststile{d_B}{} \calB$, and $\ell \geq d_Ad_B$, then $\mu \sdtstile{d_A d_B}{}  \calB$.
\end{fact}

We also have a converse to \cref{fact:sos-soundness}: every property of low-level pseudo-distributions can be derived by low-degree sum-of-squares proofs.
\begin{fact}[Completeness]
  \label{fact:sos-completeness}
  Let $d \geq r \geq r'$.
  Suppose $ \calA$ is a collection of polynomial constraints with degree at most $r$, and $ \calA \sststile{}{x} \{ \sum_{i = 1}^m x_i^2 \leq 1\}$.   Let $\{g \geq 0 \}$ be a polynomial constraint.
  If every degree-$d$ pseudo-distribution that satisfies $D \sdtstile{r}{}  \calA$ also satisfies $D \sdtstile{r'}{} \{g \geq 0 \}$, then for every $\epsilon > 0$, there is a sum-of-squares proof $ \calA \sststile{d}{} \{g \geq - \epsilon \}$.
\end{fact}

\paragraph{Basic sum-of-squares proofs.}

Now, we recall some basic facts about sum-of-squares proofs.
First, any univariate polynomial inequality admits a sum-of-squares proof over the reals.
\begin{fact}[Univariate polynomial inequalities admit SoS proofs~\cite{laurent2009sums}]
\label{fact:univariate-sos-proofs}
Let $p$ be a polynomial of degree $d$.
If $p(x)\geq 0$ for all $x \geq 0$, we have $\sststile{d}{x} \Set{ p(x) \geq 0 } $.
If $p(x)\geq 0$ for all $x \in [a, b]$, then $\Set{ x\geq a, x\leq b } \sststile{d}{x} \Set{ p(x) \geq 0 }$.
\end{fact}

Second, if $p \geq 0$ and $p$ is a quadratic, then this admits a sum-of-squares proof.
\begin{fact}[Quadratic polynomial inequalities admit SoS proofs]
    \label{fact:nonnegative-quadratic}
Let $p$ be a polynomial in the indeterminates $x \in \R^m$ such that $p$ has degree $2$ and $p \geq 0$ for all $x \in \mathbb{R}^m$. Then $\sststile{2}{x} \Set{ p(x) \geq 0  }$. 
\end{fact}
\begin{proof}
Let $M$ be the unique $(m + 1) \times (m+1 )$ Hermitian matrix such that, for $v(x) = (1,x_1, \dots , x_m)^\dagger$,
\[
p(x_1, \dots , x_m) = v(x)^\dagger M v(x) \,.
\]
The inequality $p \geq 0$ implies that $M$ is PSD: consider a vector $v  = (v_1, \dots , v_{m+1}) \in \R^{m+1}$.
If $v_1 \neq 0$, then $v^\dagger M v = p(w) \geq 0$ for $w_1 = v_2/v_1, \dots , w_m = v_{m+1}/v_1$.
If $v_1 = 0$, then $v^\dagger M v = \lim_{c \to \infty} p(c \cdot w) \geq 0$ for $w_1 = v_2, \dots , w_m = v_{m+1}$.
This shows that $M$ must be PSD, so we can write $M = \sum_{i = 1}^{m+1} u_iu_i^\dagger$ for some vectors $u_i \in \R^{m+1}$.  Thus,
\[
p(x_1, \dots , x_m) = v(x)^\dagger M v(x) = \sum_{i = 1}^{m+1} \langle u_i, v(x)\rangle^2
\]
which is a degree-$2$ SoS polynomial and we are done.
\end{proof}

We also use the following basic sum-of-squares proofs.
For further details, we refer the reader to a recent monograph~\cite{fleming2019semialgebraic}.

\begin{fact}[Operator norm bound]
\label{fact:operator_norm}
For a symmetric matrix $A \in \R^{d \times d}$ and a vector $v \mathbb{R}^d$,
\[
\sststile{2}{v} \bracks[\Big]{ v^\dagger A v \leq \norm{A}\norm{v}^2 }.
\]
\end{fact}

\begin{fact}[Almost triangle inequality] \label{fact:sos-almost-triangle}
Let $f_1, f_2, \ldots, f_r$ be indeterminates. Then
\[
\sststile{2t}{f_1, f_2,\dots,f_r} \Set{ \parens[\Big]{\sum_{i\leq r} f_i}^{2t} \leq r^{2t-1} \parens[\Big]{\sum_{i =1}^r f_i^{2t}}}.
\]
\end{fact}

\begin{fact}[SoS Hölder's inequality]\label{fact:sos-holder}
Let $w_1, \ldots w_n$ be indeterminates and let $f_1,\ldots f_n$ be polynomials of degree $d$ in the variables $x \in \R^m$. 
Let $k$ be a power of 2.  
Then
\[
\Set{w_i^2 = w_i, \forall i\in[n] } \sststile{2kd}{x,w} \Set{  \parens[\Big]{\frac{1}{n} \sum_{i = 1}^n w_i f_i}^{k} \leq \parens[\Big]{\frac{1}{n} \sum_{i = 1}^n w_i}^{k-1} \parens[\Big]{\frac{1}{n} \sum_{i = 1}^n f_i^k}}. 
\]
\end{fact}

\begin{fact}[Almost square-root]
\label{fact:squared-value-to-magnitude}
Given a scalar indeterminate $v$, $   \set{v^2 \leq 1 } \sststile{2}{v} \set{  -1 \leq v\leq 1  }$.
\end{fact}
\begin{proof}
We know that $\set{ \Paren{1-v}^2 = 1 + v^2 -2v \geq 0 }$ and  $\set{ \Paren{1 + v}^2 = 1 + v^2 + 2v \geq 0 }$.
By assumption, $\Set{ 1-v^2 \geq 0}$.
So, from the addition rule we have $\Set{ 2 + 2v\geq0  }$ and  $\Set{ 2 - 2v \geq0  }$.
Rearranging yields the claim. 
\end{proof}

\section{Translating between polynomials and nested commutators}\label{sec:poly-commutators}
In this section, we relate nested commutators of matrices to polynomials of their associated eigenvalues.
We begin with the following basic observation mentioned in the technical overview.

\begin{lemma}[Nested commutators to eigen-polynomials] \label{lem:commutator-to-polynomial}
    For matrices $A,B \in \C^{n \times n}$, in the basis where $A$ is diagonal with entries $A_{ii} = \alpha_i$, $A B = B \circ \{\alpha_i\}_{ij}$ and $B A = B \circ \{\alpha_j\}_{ij}$.
    Consequently,
    \begin{align*}
        [A, B]_k = B \circ \braces{(\alpha_i - \alpha_j)^k}_{ij}.
    \end{align*}
    Further, by linearity, for a polynomial $q(x) = \sum_{k = 0}^{d} c_k x^k$,
    \begin{align*}
        \sum c_k [A, B]_k = B \circ \braces{q(\alpha_i - \alpha_j)}_{ij}
    \end{align*}
\end{lemma}

In light of the above, we make the following definition associating a polynomial to an expression involving commutators of matrices.

\begin{definition}[Univariate ``commutator polynomials'']
\label{def:polynomial-nested-commutator}
For a polynomial $p(x) = a_0 + a_1x + \dots + a_dx^d$, given square matrices $X,A$ of the same size, define
\[
p(X|A) = a_0A + a_1 [X,A]_1 + \dots + a_d [X,A]_d \,.
\]
\end{definition}

We will need a generalization of the above that associates bivariate polynomials $p(x,y)$ to expressions with nested commutators involving two matrices $X,Y$.
We will primarily be interested in the case where $X,Y$ commute or are close to commuting.
We begin by generalizing the nested commutator.

\begin{definition}[Bivariate nested commutators]
\label{def:bi-variate-nested-commutators}
Let $S \in \{0,1 \}^\ell$ and $X,Y,A \in \C^{\dims \times \dims}$ be matrices.  Consider a sequence $Z_1,Z_2,\dots, Z_\ell$ of length $\ell$  where each $Z_i \in \{ X,Y \}$ and $Z_i = X$ if and only if the $i$th entry of $S$ is $0$.  We define 
\[
[ (X,Y)_S ,A] = [Z_1, [Z_2, [ \dots [Z_\ell, A ] \dots ]]] \,.
\]
\end{definition}

For a monomial $x^iy^j$, we wish to associate it to a nested commutator $[(X, Y)_S, A]$ where the number of $0$'s and $1$'s in $S$ is $i$ and $j$, respectively.
There are many different such commutators, reflecting that $X$ and $Y$ need not commute.
We will show that when $X,Y$ are close to commuting, the nested commutators are also close, so the ordering in $S$ does not matter.
When $\abs{S} = 2$, this follows from the identity below.

\begin{fact}[Jacobi identity]\label{fact:switch:HH'}
    We have the identity $[X,[Y,A]] - [Y, [X, A]] = [[X,Y],A]$.
\end{fact}

We extend this to higher-order commutators.
\begin{lemma}[Reordering bivariate nested commutators]\label{lem:polynomial-equivalence}
For any two sequences $S,S' \in \{0,1 \}^\ell$ with the same number of $0$'s and $1$'s, let $t \leq \ell^2$ be the number of adjacent swaps needed to transform $S$ to $S'$.
Then there are some coefficients $c_1, \dots , c_t \in \{-1,1 \}$, and sequences $S_1,T_1, \dots , S_t, T_t$ where $\len(S_i) + \len(\abs{T_i}) = \ell - 2$ such that
\[
[(X,Y)_S,A] - [(X,Y)_{S'},A] = \sum_{i = 1}^t c_i \left[ (X,Y)_{S_i}, \left[[X,Y], \left[(X,Y)_{T_i},A \right]\right] \right]. 
\]
\end{lemma}
\begin{proof}
Consider when $S,S'$ differ exactly by a single swap of two adjacent elements.  In this case, by Fact~\ref{fact:switch:HH'}, the difference on the LHS is equal to exactly one term of the form
\[
\left[ (X,Y)_{S_i}, \left[[X,Y], \left[(X,Y)_{T_i},A \right]\right] \right] 
\]
where $S_i$ is the prefix up to the point where $S,S'$ differ and $T_i$ is the suffix.  Now we can repeatedly apply this to swap adjacent elements of $S$ until it matches $S'$.  Each of the residual terms is of the form given on the RHS so we are done.
\end{proof}

Now we define the correspondence between bivariate polynomials and bivariate nested commutators by arbitrarily choosing an order for the $S$ associated to every monomial.

\begin{definition}[Bivariate ``commutator polynomials'']\label{def:polynomial-translation}
Given a polynomial of degree $d$ in two variables $x,y$, $p(x,y) = \sum_{i + j \leq d}a_{ij}x^iy^j$, we define its associated matrix commutator polynomial for matrices $X,Y,A \in \C^{n \times n}$ as
\[
    p(X,Y|A) = \sum_{i + j \leq d} a_{ij} [X, [Y,A]_j ]_i
\]
\end{definition}

The key property that we will need is that certain polynomial identities (in the original bivariate polynomials) are essentially preserved under this translation. 
We begin by showing this for monomials.
The following fact shows a relation between a commutator polynomial on $A$ to a commutator polynomial on $B$.

\begin{fact} \label{eq:move-commutator}
For any Hermitian matrix $X$ and matrices $A,B,\rho$, we have the identity
\begin{equation*}
    \tr\Paren{ [X,A] B^\dagger \rho } - \tr\Paren{  A [X,B]^\dagger \rho }  = -\tr\Paren{AB^\dagger[X,\rho] } \,.
\end{equation*}
\end{fact}
\begin{proof}
\begin{align*}
\tr\Paren{ [X,A] B^\dagger \rho } - \tr\Paren{  A [X,B]^\dagger \rho }  & = \tr\Paren{ XAB^\dagger \rho - AXB^\dagger \rho  }  - \tr\Paren{ AB^\dagger X \rho -  A XB^\dagger \rho } \\
& = \tr\Paren{ XAB^\dagger \rho  }  - \tr\Paren{ AB^\dagger X \rho } \\
& = \tr\Paren{ AB^\dagger \rho X } - \tr\Paren{ AB^\dagger X \rho }\\
& = -\tr\Paren{AB^\dagger[X,\rho] } \qedhere
\end{align*}
\end{proof}

This can be extended to general monomials.

\begin{lemma}[Commutator monomial equivalences]
    \label{lem:monomial-equivalence}
    Let $p(x, y) = x^{i_1}y^{i_2}$, $q(x, y) = x^{j_1}y^{j_2}$, and $r(x, y) = p(x, y)q(x, y)$.
    Let $d = \deg(r)$.
    Then for some $\ell \leq d^2$, we can write
    \begin{align*}
        &\tr\parens[\Big]{
            p(X, Y | A) q(X, Y | B)^\dagger \rho
        } - \tr \parens[\Big]{
            A \cdot r(X, Y | B)^\dagger \rho
        } = \sum_{i=1}^\ell Z_i,
    \end{align*}
    where every $Z_i$ is one of the following three types of errors:
    \begin{enumerate}
        \item $\pm \tr\left( \left[(X,Y)_S, A \right] \left[(X,Y)_T, B \right]^\dagger [X , \rho]  \right)$, where $\len(S) + \len(T) = d - 1$;
        \item $\pm \tr\left( \left[(X,Y)_S, A \right] \left[(X,Y)_T, B \right]^\dagger [Y, \rho]  \right)$, where $\len(S) + \len(T) = d - 1$;
        \item $\pm \tr\left( A \left[ (X,Y)_S, \left[[X, Y], [(X,Y)_T, B] \right]\right]^\dagger  \rho \right)$, where $\len(S) + \len(T) = d - 2$.
    \end{enumerate}
\end{lemma}
\begin{proof}
Our goal is to express
\begin{align*}
    \tr\Paren{  [X,[Y,A]_{i_2}]_{i_1} [X,[Y,B]_{j_2}]_{j_1}^{\dagger} \rho}
    - \tr \parens[\Big]{ A[X,[Y,B]_{ i_2 + j_2}]_{i_1 + j_1}^{\dagger} \rho }
\end{align*}
as a sum of errors.
Observe that \cref{eq:move-commutator} allows us to remove one copy of $X$ or $Y$ from the commutator in front of $A$ and move it onto the commutator in front of $B$ at the cost of an error term of type $1$ or $2$.
Thus, we can repeatedly apply \cref{eq:move-commutator} to move all of the $X$'s and $Y$'s in the commutator in front of $A$ onto the commutator in front of $B$ and write 
\[
\tr\Paren{  [X,[Y,A]_{i_2}]_{i_1} [X,[Y,B]_{j_2}]_{j_1}^{\dagger} \rho} -  \tr\Paren{ A[Y, [X,[Y,B]_{j_2}]_{j_1 + i_1} ]_{i_2}^{\dagger} \rho} 
\]
as a sum of $i_1$ terms of type $1$ and $i_2$ terms of type $2$.
Next, we can apply \cref{lem:polynomial-equivalence} to reorder the sequence of $X$ and $Y$ on the commutator in front of $B$ at the cost of error terms of type $3$.  This allows us to write 
\[
 \tr \parens[\Big]{ A[Y, [X,[Y,B]_{j_2}]_{j_1 + i_1} ]_{i_2}^{\dagger} \rho}  - \tr \parens[\Big]{ A[X,[Y,B]_{ i_2 + j_2}]_{i_1 + j_2}^{\dagger} \rho }
\]
as a sum of $i_2(j_1 + i_1)$ terms of type $3$.
Together, this gives the desired bound.
\end{proof}

\begin{theorem}[Polynomial identities to nested commutator identities]
\label{thm:polynomial-equivalence}
    Consider a formal polynomial identity in two variables
    \[
    p_1(x,y)q_1(x,y) + \dots + p_k(x,y)q_k(x,y) = 0
    \]
    where each of the polynomials $p_i,q_i$ is $(d,C)$-bounded. Let $X, Y \in \C^{\dims \times \dims}$ be Hermitian matrices and $\rho, A,B \in \C^{\dims \times \dims}$ be arbitrary matrices.
    Then we can write
    \begin{align*}
        \tr\left( \left(p_1(X,Y|A)q_1(X,Y|B)^\dagger + \dots + p_k(X,Y|A)q_k(X,Y|B)^\dagger \right) \rho \right)
        = \sum_{i = 1}^{t} c_i Z_i,
    \end{align*}
    where $t \leq 4kd^6$,
    the coefficients $c_i$ satisfy $\abs{c_i} \leq C^2 2^{2d}$,
    and every $Z_i$ is one of the following three types of errors:
    \begin{enumerate}
        \item $\frac{\pm 1}{(\len(S) + \len(T))!} \tr\left( \left[(X,Y)_S, A \right] \left[(X,Y)_T, B \right]^\dagger [X , \rho]  \right)$, where $\len(S) + \len(T) \leq 2d$;
        \item $\frac{\pm 1}{(\len(S) + \len(T))!} \tr\left( \left[(X,Y)_S, A \right] \left[(X,Y)_T, B \right]^\dagger [Y, \rho]  \right)$, where $\len(S) + \len(T) \leq 2d$;
        \item $\frac{\pm 1}{(\len(S) + \len(T))!} \tr\left( A \left[ (X,Y)_S, \left[[X, Y], [(X,Y)_T, B] \right]\right]^\dagger  \rho \right)$, where $\len(S) + \len(T) \leq 2d$.
    \end{enumerate}
\end{theorem}
\begin{proof}
Let $r_\ell(x, y) = p_\ell(x, y)q_\ell(x, y)$.
By the assumed formal polynomial identity,
\begin{align*}
    \tr\left( \left(A r_1(X,Y|B)^\dagger + \dots + A r_k(X,Y|B)^\dagger \right) \rho \right) = 0.
\end{align*}
So, it suffices to express each product
\begin{align*}
    \tr\left( p_\ell(X,Y|A)q_\ell(X,Y|B)^\dagger \rho \right)
    - \tr\left( A r_\ell(X,Y|B)^\dagger \rho \right)
\end{align*}
as a linear combination of errors.
Let $p_\ell = \sum_{i_1,i_2} a_{\ell, i_1, i_2 } x^{i_1}y^{i_2}$ and $q_\ell = \sum_{ j_1, j_2}b_{\ell,j_1,j_2} x^{j_1}y^{j_2} $.  We can expand the above expression into its commutator monomials:
\begin{align*}
    &\tr\left( p_\ell(X,Y|A)q_\ell(X,Y|B)^\dagger \rho \right)
    - \tr\left( A r_\ell(X,Y|B)^\dagger \rho \right) \\
    &= \sum_{i_1, i_2, j_1, j_2}a_{\ell,i_1,i_2} \cdot b_{\ell, j_1, j_2} \parens[\Big]{\tr([X,[Y,A]_{i_2}]_{i_1} \cdot [X,[Y,B]_{j_2}]_{j_1}^{\dagger} \rho) - \tr(A\cdot [X,[Y,B]_{i_2 + j_2}]_{i_1 + j_1}^{\dagger} \rho)}\,.
\end{align*}
\cref{lem:monomial-equivalence} shows how to expand summand into error terms; note that the degree of the product is $i_1 + j_1 + i_2 + j_2 \leq 2d$.
There are at most $(2d)^2$ error terms per summand, and there are $d^4$ summands per product $p_\ell q_\ell$.
This gives a total bound on the number of error terms as $4kd^6$ as desired.
All that is left is to bound the size of the coefficients.
For an error term of any type associated to a particular summand, the coefficient is $a_{\ell,i_1,i_2}b_{\ell,j_1,j_2}$.
Let $S, T$ be the sequences associated to the error term.
Because $p_\ell$ and $q_\ell$ are $(d, C)$-bounded, we can bound
\begin{align*}
    \abs{a_{\ell,i_1,i_2}b_{\ell,j_1,j_2}} \leq \frac{C^2}{(i_1 + i_2)! (j_1 + j_2)!} \leq \frac{C^2 2^{2d}}{(i_1 + i_2 + j_1 + j_2)!} \leq \frac{C^2 2^{2d}}{(\len(S) + \len(T))!}.
\end{align*}
The last inequality holds because, in all error types, $\len(S) + \len(T)$ is at most the degree of the product.
\end{proof}

\section{Polynomial approximations of the exponential}\label{sec:poly-approx}

Our goal for this section is to construct a polynomial approximation to $e^x$ that satisfies the properties necessary for our algorithm.
In particular, we need a polynomial that approximates $e^x$ on an interval $[-\kappa, \kappa]$, but with the additional property that the polynomial does not grow too quickly outside the area of approximation.

\begin{definition}[Flat approximation of the exponential] \label{def:flat-exp}
    Given $\eps, \eta \in (0, 1)$ and $\kappa \geq 1$,  we say a polynomial $p(x)$ is a $(\kappa, \eta,\eps)$-\textit{flat exponential approximation} if
    \begin{enumerate}
        \item $|p(x) - e^x| \leq \eps $ for $x \in [-\kappa, \kappa]$, and
        \item $|p(x)| \leq \max(1, e^x)e^{\eta |x|}$ . 
    \end{enumerate}
\end{definition}

The canonical polynomial approximation of $e^x$ is its Taylor series truncation.

\begin{definition}
    Let $s_{\ell}(x) = \sum_{k = 0}^\ell \frac{x^k}{k!}$ denote the degree-$\ell$ truncation of the Taylor series of $e^x$.    
\end{definition}

\begin{remark} \label{rmk:taylor-fails}
The Taylor series truncation $s_\ell(x)$ does not satisfy our desired approximation properties from \cref{def:flat-exp}, even for large $\ell$.
This is because the truncation blows up too quickly on the negative tail: $\abs{s_\ell(-\ell)} \eqsim e^{\ell} > e^{\eta \ell}$.
The same issue occurs with conventional approximations of the exponential, like truncations of the Chebyshev series and QSVT-style ``bounded'' approximations of the Chebyshev series~\cite{gslw18,tt24}: these truncations at degree $\ell$ fail in the region around $-\ell$ (so, notably, increasing the degree does not improve the flatness parameter of these approximations).
\end{remark}

While standard techniques for approximating the exponential do not suffice, we construct a modified polynomial that is a flat approximation to the exponential. 

\begin{definition}[Iteratively truncated Taylor series]
\label{def:iterative-truncated-taylor-series}
For positive integers $k,\ell$, we define
\[
p_{k,\ell}(x) = s_{2\ell}(x/k) \cdot s_{4\ell}(x/k) \cdots s_{2^{k}\ell}(x/k) \,
\]
to be the product of $k$ geometrically increasing Taylor series truncations.  Next, we define a shifted integral of $p_{k,\ell}(x)$ as follows:
\[
q_{k,\ell}(x) = 1 + \int_{0}^x p_{k,\ell}(y) \diff y \,.
\]
The degree of $p_{k, \ell}(x)$ and $q_{k, \ell}$ is $(2^{k+1}-1)\ell$ and $(2^{k+1}-1)\ell+1$, respectively.
\end{definition}

It will also be important that $p_{k,\ell}$ is always positive and thus $q_{k,\ell}$ is monotonically increasing.
The main theorems that we show in this section are as follows:

\begin{theorem}[Flat approximations to the exponential]
\label{thm:exp-fancy-approx}
Given $\eps, \eta \in (0,1)$ and $\kappa \geq 1$,  let $k \geq 5/\eta$ and $\ell \geq 100(\kappa +\log k/\eps)$.  Then the polynomials $p_{k,\ell}, q_{k,\ell}$ are $(\kappa, \eta, \eps)$-\textit{flat exponential approximation}. 
\end{theorem}
\begin{theorem}[Monotonicity of the flat approximations]\label{thm:exp-monotone-approx}
For any positive integers $k, \ell$, we have the following:
\begin{equation*}
    \sststile{2^{k+1} \ell + 20 }{x,y} \Set{ 0.5(x - y)(1 + 0.25(x - y)^2) (q_{k,\ell}(x) - q_{k,\ell}(y)) - 0.00025(x - y)^2 p_{k,\ell}(x) \geq 0  } .
\end{equation*}
Moreover, the LHS is a  $(10^{2^k\ell},\,2^k \ell + 10,\,200^{2^k\ell})  $-bounded sum-of-squares polynomial in $x$ and $y$.
\end{theorem}

\begin{remark}
Note that \cref{thm:exp-monotone-approx} stems from the intuition that roughly we should have
\[
q_{k,\ell}(x) - q_{k,\ell}(y) = \int_y^x p_{k, \ell}(y) \diff y \geq (x - y)(p_{k,\ell}(x) - p_{k,\ell}(y)) \,.
\]
However, the above isn't quite true, so the additional terms in \cref{thm:exp-monotone-approx} are there to account for this.  The fact that the difference is not only non-negative, but also a sum of square polynomials will be crucial later on in the analysis of our algorithm.
\end{remark}

We will prove \cref{thm:exp-fancy-approx} in this section.  The proof of \cref{thm:exp-monotone-approx} is deferred to the appendix as it is quite long and computational.  We begin by establishing some basic facts about the truncated Taylor series and our polynomials $p$ and $q$.  The following fact follows from Taylor's theorem, which implies that for every $x \in \R$ there exists some $c \in [0,1]$ such that
\begin{align} \label{eq:exp-taylors-thm}
    e^x - s_\ell(x) = e^{cx}\frac{x^{\ell+1}}{(\ell+1)!}\,.
\end{align}
\begin{fact}[Bounds on truncated Taylor series of $e^x$]\label{fact:even-odd-truncations}
For $x \geq 0$, $e^x \geq s_{\ell}(x)$ for all $\ell$.  For $x < 0$, we have $s_{2\ell - 1}(x) \leq e^x \leq s_{2\ell}(x)$ for all $\ell$.
\end{fact}

\begin{corollary}[Even truncations are non-negative]\label{coro:positive-truncation}
For any $\ell \in \mathbb{N}$, for all $x \in \mathbb{R}$,   $s_{2\ell}(x) \geq 0$.  
\end{corollary}

\begin{lemma}[Taylor series truncation of $e^x$] \label{lem:exp-taylor}
Given $\eps \in (0,1)$ and $b \geq 0$,  let $\ell \geq 10b + \log(1/\eps)$. Then $\abs{s_\ell(x) - e^x} \leq \eps$ for $x \in [-b, b]$.
\end{lemma}
\begin{proof}
This statement also follows from Taylor's theorem~\eqref{eq:exp-taylors-thm}.
From this, for all $x$, the approximation error of the Taylor series truncation is
\begin{align*}
    \abs{s_{\ell-1}(x) - e^x} \leq \frac{e^{\abs{x}} \abs{x}^{\ell}}{\ell!}
    \leq e^{\abs{x}}\parens[\Big]{\frac{e\abs{x}}{\ell}}^\ell.
\end{align*}
Taking $\ell \geq e^2 \abs{x} + \log(e^{\abs{x}}/\eps)$ makes the above expression bounded by $\eps$.\footnote{
    This can be improved by a $\log\log$ factor~\cite[Lemma 59]{gslw18}, but this does not improve the degree for the parameter setting in our application.
}
The final bound comes from taking $\abs{x} \leq b$.
\end{proof}

\begin{claim}\label{claim:simple-coeff-bound}
The polynomial $p_{k,\ell}$ is $( (2^{k+1} - 1)\ell, 1)$-bounded and the polynomial $q_{k,\ell}$ is $( (2^{k+1} - 1)\ell + 1, 1)$-bounded.
\end{claim}
\begin{proof}
Note that the coefficients of the monomials in $p_{k,\ell}$ are all positive and are all less than the coefficients of the monomials in 
\[
\left(1 + \frac{(x/k)}{1!} + \frac{(x/k)^2}{2!} + \dots  \right)^k
\]
where we replace each truncated sum with the full infinite sum.  However the above is exactly equal to $1 + x + \frac{x^2}{2!} + \dots $.  Thus for any degree $d$, the coefficient of $x^d$ in $p_{k,\ell}(x)$ is at most $\frac{1}{d!}$ so $p_{k,\ell}$ is $((2^{k+1} - 1)\ell , 1)$-bounded since its degree is $(2^{k+1} - 1)\ell$.  Recall that $q_{k,\ell}$ is obtained by integrating $p$ and thus we immediately get the same bound on the coefficients of $q_{k,\ell}$ i.e. it is is $((2^{k+1} - 1)\ell + 1, 1)$-bounded and we are done.
\end{proof}

Now we move onto the proof of Theorem~\ref{thm:exp-fancy-approx}.  

\begin{proof}[Proof of \cref{thm:exp-fancy-approx}]
We prove the statement for $p_{k,\ell}$ first. 
By Lemma~\ref{lem:exp-taylor} (with $\eps \gets \eps/(k2^{2\kappa}) )$), we have that for all $x \in [-\kappa, \kappa]$
\[
\left\lvert \frac{s_{2^j\ell}(x/k)}{e^{x/k}} - 1 \right\rvert \leq \frac{\eps}{k \cdot e^{2\kappa}}
\]
and thus,
\[
1 - \frac{\eps}{e^{2\kappa}} \leq \frac{p_{k,\ell}(x)}{e^x}  -1 \leq 1 + \frac{\eps}{e^{2\kappa}} 
\]
so we have $|p_{k,\ell}(x) - e^x| \leq \eps/\kappa$ on the interval $[-\kappa, \kappa]$.  Now we immediately get the same guarantee for $q_{k,\ell}(x)$ by integrating the above.

Now we prove the second condition for $p_{k,\ell}$.  For $x \geq 0$, we simply incur the $e^x$ upper bound by \cref{fact:even-odd-truncations}.  For $x \leq 0$, let $j_0$ be the smallest positive integer such that $-x < 2^{j_0+1}k \ell$ then for all $1 \leq j < j_0$,
\[
s_{2^j\ell}(x/k) \leq \frac{(x/k)^{2^j\ell}}{(2^j\ell)!} \,.
\]

Also, for all $j > j_0 + 2$,
\begin{equation}
s_{2^j\ell}(x/k) = s_{ (2^j\ell) - 1}(x/k) + \frac{(x/k)^{2^j\ell}}{(2^j\ell)!} \leq e^{x/k} + \left(\frac{3x}{2^j\ell k} \right)^{2^j\ell}\leq  1
\end{equation}
where we used \cref{fact:even-odd-truncations}.  Finally, for $j_0 \leq j \leq j_0 + 2 $
we have
\begin{equation}
\label{eqn:weak-upper-bound-for intermediate-terms}
s_{2^j\ell}(x/k) = s_{(2^j \ell) - 1}(x/k) + \frac{(x/k)^{2^j\ell}}{(2^j\ell)!}  \leq 1 + \frac{(x/k)^{2^j\ell}}{(2^j\ell)!} \leq 2 \max\left( \frac{(x/k)^{2^j\ell}}{(2^j\ell)!} , 1\right)
\end{equation}
Overall, combining the above, we conclude
\begin{equation}
\label{eqn:combined-upper-bound-pkell}
p_{k,\ell}(x) = s_{2\ell}(x/k) \cdots s_{2^k\ell}(x/k) \leq 8 \prod_{j = 1}^k \max\left(\frac{(x/k)^{2^j\ell}}{(2^j\ell)!} , 1 \right) 
\end{equation}
where the factor of $8$ is because we need to apply \eqref{eqn:weak-upper-bound-for intermediate-terms} on at most $3$ terms.  Note that there must be some $k_0 \in \Set{j_0, j_{0}+1, j_0 + 2}$, such that
\[
\frac{(x/k)^{2^j\ell}}{(2^j\ell)!} > 1,
\]
exactly when $j \leq k_0$.  Then, we can upper bound Equation \eqref{eqn:combined-upper-bound-pkell} as follows:
\begin{equation}
\label{eqn:refined-upper-bound-pkell}
p_{k,\ell}(x) \leq \frac{8(x/k)^{(2^{k_0+1} - 2) \ell}}{(2\ell)! \cdots (2^{k_0}\ell)!} 
\end{equation}
Further, we have that
\[
(2\ell)! \cdots (2^{k_0}\ell)! \geq \frac{((2^{k_0+1} - 2)\ell)!}{4^{(2^{k_0+1} - 2)\ell}}.
\]
Substituting in Equation \eqref{eqn:refined-upper-bound-pkell}, we get 
\begin{equation}
    p_{k,\ell}(x) \leq \frac{(5x/k)^{(2^{k_0+1} - 2)\ell}}{((2^{k_0+1} - 2)\ell)!} \leq e^{5|x|/k} \,.
\end{equation}
The desired bound for $q$ then follows as well from the definition as $q$ is obtained by simply integrating $p$.
\end{proof}

\section{Accessing Gibbs states: the quantum piece}
\label{sec:accessing-gibbs-states}

We now describe how our algorithm uses its copies of $\rho$.
This is the only location in the algorithm where we access $\rho$; the rest of the algorithm is entirely classical.
By making measurements on copies of $\rho$, we can estimate expectations of observables $\tr(X\rho)$ for various choices of the matrix $X$.  These estimates will then be used in the learning algorithm.

\begin{lemma} \label{lem:estimator-helper}
    Let $X = UDU^\dagger \in \mathbb{C}^{\dims \times \dims}$ be a unitarily diagonalizable matrix, and suppose we are given copies of a quantum state with density matrix $\rho \in \mathbb{C}^{\dims \times \dims}$.
    Then we can estimate $\tr(X \rho)$ to $\eps\norm{X}$ error with probability $\geq 1-\delta$ using $\bigO{\frac{1}{\eps^2}\log\frac{1}{\delta}}$ copies of $\rho$.
    The running time is the number of samples times the cost of applying $U^\dagger$ to a quantum state and measuring in the computational basis.
\end{lemma}
\begin{proof}
Consider taking $\rho$ and applying $U^\dagger$ to form $U^\dagger \rho U$, then measuring in the computational basis.
We see the outcome $i$ with probability $\bra{u_i} \rho \ket{u_i}$, where $u_i$ is the $i$th column of $U$.
Let $z$ be the random variable attained by performing this measurement and taking $z = D_{i,i}$; then $\E[z] = \sum_i D_{i,i}\bra{u_i} \rho \ket{u_i} = \tr(X \rho)$, and $z$ is always bounded by $\max_i \abs{D_{i,i}} = \norm{X}$.
By a Chernoff bound we conclude that averaging $\bigO{\frac{1}{\eps^2}\log\frac{1}{\delta}}$ independent copies of this estimator gives the desired error bound.
\end{proof}

We wish to estimate $\tr(X\rho)$ for many different Pauli observables $X$, so we could simply run \cref{lem:estimator-helper} for each one.
However, we can use a classical shadows-like procedure~\cite{hkp20} to estimate them all at once.

\begin{lemma}[Computing expectations of observables of Gibbs states]\label{lem:basic-estimator}
    Let $H = \sum_{a=1}^\terms \lambda_a E_a$ be a $\locality$-local Hamiltonian on $\qubits$ qubits whose dual interaction graph $\graph$ has maximum degree $\degree$.
    Consider the set of $\ell$-$\graph$-local Paulis, $\locals_{\ell}$.
    There is a quantum algorithm that, with probability $\geq 1-\delta$, outputs estimates of $ \tr(P_1P_2P_3 \rho)$ to $\eps$ error for all $P_1, P_2 , P_3 \in \locals_\ell$.
    This algorithm uses $t = \bigO{\frac{4^{3\ell}}{\eps^2}\log\frac{\abs{\locals_\ell}^3}{\delta}}$ samples, $\bigO{\qubits t}$ quantum gates, and $\poly( t, \qubits, \abs{P_\ell} )$ classical post-processing time.
    
\end{lemma}
\begin{proof}
Consider the following procedure, where we take as input a copy of $\rho$, rotate into the eigenbasis of a Pauli matrix $P \in \locals$ chosen uniformly at random, and then measure in the computational basis to get a measurement outcome $b \in \{0,1\}^\qubits$.
Take $t = \bigO{\frac{4^{3\ell}}{\eps^2}\log\frac{\abs{\locals_\ell}^3}{\delta}}$ of these samples $(P, b)$ and denote the collection $\mathcal{S}$.

Now, consider estimating some $\tr(X \rho)$ where $X = UDU^\dagger \in \locals$ and $\supp(X) \leq 3\ell$; $P_1 P_2 P_3$ is such an $X$, up to a root of unity.
Since
\begin{align*}
    \tr(X \rho) = \tr_{\supp(X)}(X_{(\supp(X))}\tr_{[n] - \supp(X)}(\rho)),
\end{align*}
it suffices to consider $\rho$ on the support of $X$.
Let $\mathcal{T} \subset \mathcal{S}$ be the collection of samples whose Pauli matrix agrees with $X$ on the support of $X$.
Then we can use this subsample to run \cref{lem:estimator-helper} and estimate $\tr(X \rho)$, as all that is needed to use it are measurements in a basis where $X$ is diagonal.
Since the probability that any given sample is in a basis where $X$ is diagonal is $\geq 4^{-3\ell}$, with $t$ samples we can guarantee that we get an estimate $\tr(X \rho)$ to $\eps$ error with failure probability $\delta / \abs{\locals_\ell}^3$.
By union bound, we can then conclude that with our set of samples we can get $\eps$-good estimates of every $\tr (P_1 P_2 P_3 \rho)$ with probability $\geq 1-\delta$.

The running time can be seen to be $\poly(t, \qubits, \abs{\locals_\ell})$.
Note that, for any product of Paulis $P \in \locals$, we can apply the quantum gate to rotate into its eigenbasis with $\bigO{\qubits}$ gates, so the quantum gate complexity is $\bigO{\qubits t} = \bigO{\frac{4^{3\ell}\qubits}{\eps^2}\log\frac{\abs{\locals_\ell}^3}{\delta}}$.
\end{proof}
\section{Algorithm and analysis}
\label{sec:algorithm-and-analysis}

Now we are ready to present our learning algorithm.
We solve the following Hamiltonian learning problem:

\begin{problem}[Hamiltonian learning] \label{prob:main}
    Recalling the set-up described in \cref{def:hamiltonian,def:low-insersection-ham}, let $H = \sum_{a=1}^\terms \lambda_a E_a \in \mathbb{C}^{\dims \times \dims}$ be a $\locality$-local Hamiltonian on $\qubits$ qubits whose terms $E_a$ are known, distinct, non-identity Pauli operators, and whose coefficients $\lambda_a \in \R$ satisfy $\abs{\lambda_a} \leq 1$.
    Let $\degree$ denote the maximum degree of a vertex in the dual interaction graph associated to the Hamiltonian.
    Given $\eps, \delta > 0$, along with copies of the Gibbs state $\rho = \frac{\exp(-\beta H)}{\tr \exp(-\beta H)}$ corresponding to $H$ at a known inverse temperature $\beta > 0$, find estimates $\hat{\lambda}_a$ such that, with probability $\geq 1-\delta$, $\parens{\hat{\lambda}_a - \lambda_a}^2 \leq \eps^2$ for all $a \in [\terms]$.
\end{problem}

\begin{theorem}[Efficiently learning a quantum Hamiltonian]
\label{thm:main-ham-learning-theorem}
    Let $H = \sum_{a=1}^\terms \lambda_a E_a \in \mathbb{C}^{\dims \times \dims}$ be a $\locality$-local Hamiltonian on $\qubits$ qubits whose dual interaction graph has degree $\degree$ (as in \cref{prob:main}).
    Suppose we are given the terms $\{E_a\}_{a \in [\terms]}$, $\eps > 0$, $\delta > 0$, and copies of the Gibbs state $\rho$ at a known inverse temperature $\beta > 0$.
    Then there exists an algorithm that can output estimates $\tilde{\lambda}_a$ such that, with probability $\geq 1-\delta$, $\parens{\hat{\lambda}_a - \lambda_a}^2 \leq \eps^2$ for all $a \in [\terms]$.
    This algorithm uses
    \begin{equation*}
        \bigO[\Big]{(\terms^6/\eps^{e^{f(\locality, \degree)\beta}})\log(\terms/\delta) + (f(\locality, \degree)/(\beta^2\eps^2))\log(\terms/\delta)}
    \end{equation*}
    copies of the Gibbs state and running time
    \begin{equation*}
        \poly\Paren{ m,\log(1/\delta) }\cdot \Paren{1/\eps}^{e^{f(\locality,\degree)\beta}} + f(\locality, \degree)\cdot (\terms/(\beta^2\eps^2))\log(\terms/\delta),
    \end{equation*}
    where $f(\locality, \degree)$ is a positive function depending only on $\locality$ and $\degree$.
\end{theorem}

Specifically, we prove that, when $\beta > g(\locality, \degree)$ for some positive $g$, there exists an algorithm that succeeds with sample complexity $(\terms^6/\eps^{e^{f(\locality, \degree)\beta}})\log(\terms/\delta)$ and time complexity $\poly\Paren{ m,\log(1/\delta) }\cdot \Paren{1/\eps}^{e^{f(\locality,\degree)\beta}}$ for some positive $f(\locality, \degree)$.
We get \cref{thm:main-ham-learning-theorem} by taking $g$ to be the threshold at which the high-temperature algorithm~\cite{hkt21} stops working and then combining our complexity bounds with that of the high-temperature algorithm.
That $\beta$ is lower bounded by a constant is just a simplification; the same algorithm and analysis should work for any temperature upon appropriate adjustments.

We are concerned with the setting where $\locality$ and $\degree$ are constant (that is, when the Hamiltonian is low-intersection), so we do not try to optimize dependence on these parameters.
Further, we assume that $\qubits = \bigO{\terms}$; this is without loss of generality, as $\terms \geq \qubits / \locality$ is necessary for the support of the Hamiltonian to be all $\qubits$ qubits.

Our main theorem follows from showing that a low-degree sum-of-squares relaxation of a carefully chosen polynomial system can be rounded to obtain accurate estimates of the true coefficients $\{\lambda_a\}_{a \in [\terms]}$.
We first set up the system, and then build up to the full algorithm (\cref{algo:learning-arbitrary-hams}).

\paragraph{Presenting the polynomial system.}
We begin by setting up the precise polynomial system that we will then solve.  Let $\eps_0 = \frac{\eps^{10^{C_{\locality,\degree}\beta}}}{\terms^3}$ for some sufficiently large constant $C_{\locality,\degree}$ depending only on $\locality$ and $\degree$.  Next, let $\ell_0 = 2^{C_{\locality,\degree}\beta}  \log(1/\eps), \ell_1 = 4\locality$ and define $\calA = \locals_{4^{C_{\locality,\degree} \beta } \ell_0}$ and $\calB = \locals_{\ell_1}$ (as in \cref{def:dual-interaction-graph}).

Let $\lambda' = [\lambda_1', \lambda_2', \ldots, \lambda_\terms']$ be a set of indeterminates.  We will solve a polynomial system in these indeterminates, and then use it to extract estimates to the true coefficients.
In the system, and throughout, we denote $H' = \sum_{a} \lambda'_a E_a$.
This object $H'$ is merely notational convenience: we do not optimize with this exponential sized object. We only ever need to work with traces of $H'$ and thus our polynomial system admits a succinct (polynomial sized) representation. 

We begin by making measurements of $\rho$ to construct $\eps_0$-accurate estimates to $\tr\left(A_1A_2A_3\rho \right)$  for all $A_1, A_2 , A_3 \in \calA$ using \cref{lem:basic-estimator}.
Let $\wt{\tr}$ denote the map from $A_1A_2A_3\rho$ to our estimate of $\tr(A_1A_2A_3\rho)$, extended by linearity to linear combinations of such matrices that we have estimated.\footnote{
    To make this well-defined, we need to ensure that we only estimate a set of $A_1 A_2 A_3$ that are linearly independent.
    This simply involves removing duplicates up to phase.
}
Again, this is notational convenience for our polynomial system.

By \cref{coro:number-of-local-paulis}, $|\calA| \leq \terms(10\degree)^{4^{C_{\locality,\degree} \beta }\ell_0} \leq \terms (1/\eps)^{10^{C_{\locality,\degree}\beta}}$ so we can produce these estimates in running time $\poly_{\beta, \locality, \degree}(\terms,1/\eps)$ and sample complexity
\begin{equation}
    \bigO[\Bigg]{\frac{4^{3 \cdot 4^{C_{\locality, \degree} \beta}\ell_0}}{\eps_0^2}\log\frac{\abs{\locals_{4^{C_{\locality, \degree} \beta}\ell_0}}}{\delta}}
    = \bigO[\Bigg]{\frac{\terms^6}{\eps^{20^{C_{\locality, \degree}\beta}}}\log\frac{\terms}{\delta}}.
\end{equation}
This is the sample complexity of the entire algorithm, as this is all of the information we need of $\rho$.
Then, we write the following polynomial system:

\newcommand{\constraintsystem}{
    \calC_{\lambda'} =
        \left \{\begin{aligned}
            & \forall a \in [\terms] & -1 \leq \lambda_a' &\leq 1 \\
            & & H' = \sum_{a \in [\terms] } \lambda_a' \cdot E_a &\\
            & \forall A \in \calA &  \abs{\wt{\tr}\Paren{ A \Paren{ H'\rho - \rho H' } } }^2  &\leq \eps_0^2 \\
            & \forall B_1, B_2 \in \calB,
            &  \abs*{\wt{\tr}\Paren{ B_2  q_{C_{\locality,\degree}\beta, \ell_0}\Paren{-\beta H'| B_1} \rho }  -\wt{\tr}\Paren{ B_1B_2\rho }}^2
            & \leq   \eps^2 \\
        \end{aligned}\right \},
}
\begin{equation}\label{eqn:const-system}
    \constraintsystem
\end{equation}
where $q$ is the polynomial from \cref{def:iterative-truncated-taylor-series} extended to commutators using~\cref{def:polynomial-translation}.
With this set of parameters, by \cref{thm:exp-fancy-approx},
\begin{gather}
    \deg(q_{C_{\locality,\degree}\beta, \ell_0}) \leq 2 \cdot 4^{C_{\locality,\degree}\beta}\log(1/\eps) + 1 \text { and} \label{eq:q-main-degree}\\
    q_{C_{\locality,\degree}\beta, \ell_0} \text{ is a } \parens[\Big]{0.001 \cdot 2^{C_{\locality,\degree}\beta} \log(1/\eps) , \frac{5}{C_{\locality,\degree}\beta} , 0.001\eps}\text{-flat exponential approximation.} \label{eq:q-main-flatness}
\end{gather}

\paragraph{Representing the polynomial system.}
It is important to understand how the above can be represented as a polynomial system in the $\lambda'$ with real coefficients and polynomial (in $\terms$) size.  The key is that we can write the expressions inside the traces in the form $\left(\sum_{\alpha} (\lambda')^{\alpha} M_{\alpha}\right) \rho$ where $\alpha$ ranges over low-degree monomials and $M_{\alpha}$ are rescaled Pauli matrices, which we can explicitly and succinctly represent.  Then by linearity of $\wt{\tr}$, we can write the system by plugging in our estimates for $\wt{\tr}(M_{\alpha}\rho)$.  For example, for the first constraint, we can write
\[
    \wt{\tr}(A_1A_2 (H' \rho - \rho H')) = \wt{\tr}( (A_1A_2H' - H'A_1A_2)\rho) = \sum_{a} \lambda_a' (\wt{\tr}(A_1A_2E_a\rho) - \wt{\tr}( E_aA_1A_2 \rho ) )\,.
\]
We can then plug in our estimates $\wt{\tr}(A_1A_2E_a\rho)$ and $\wt{\tr}(E_aA_1A_2\rho)$ obtained from \cref{lem:basic-estimator}.
These trace estimates may be complex-valued, so that $\wt{\tr}(A_1A_2E_a\rho) - \wt{\tr}( E_aA_1A_2 \rho ) = \psi_a + \ii \chi_a $ for some $\psi_a, \chi_a \in \R$.
We can rewrite this,
\begin{align*}
    \abs[\Big]{ \sum_{a} \lambda_a' (\chi_a + \ii \psi_a) }^2 &\leq \eps_0^2 \\
    \iff
    \parens[\Big]{ \sum_{a} \lambda_a' \chi_a }^2 + \parens[\Big]{\sum_a \lambda_a' \psi_a}^2 &\leq \eps_0^2,
\end{align*}
revealing that we have a polynomial constraint in $\lambda'$ with real coefficients.  For the second constraint, $q_{C_{\locality,\degree}\beta, \ell_0}\Paren{-\beta H'|B_1}$ is a linear combination of nested commutators $[E_{a_\ell},[E_{a_{\ell-1}}, [\dots[E_{a_1}, B_1]]]]$ where $\ell \leq 2 \cdot 4^{C_{\locality,\degree}\beta}\log(1/\eps) + 1$, so by \cref{lem:pauli-commutator}, such a commutator is an element of $\calA$.
The coefficient associated with this nested commutator is $(-\beta)^\ell\lambda_{a_\ell}'\dots \lambda_{a_1}'$, so that for some polynomials $p_A(\lambda')$ we can write
\[
    \wt{\tr}\Paren{ B_2  q_{C_{\locality,\degree}\beta, \ell_0}\Paren{-\beta H'| B_1} \rho } = \sum_{A \in \calA} p_A(\lambda') \wt{\tr}\left( B_2A \rho \right).
\]
Since we can plug in our estimates with $\wt{\tr}$ as before, this expression is well-defined.
We can further conclude that the associated constraint can be written as a polynomial constraint in $\lambda'$ with real coefficients.
Recall that $|\calA| \leq  \terms (1/\eps)^{10^{C_{\locality,\degree}\beta}}$.
Furthermore, by \cref{lem:count-monomials}, there are at most $\terms(10\degree)^{4^{C_{\locality,\degree}\beta }\ell_0} \leq \terms (1/\eps)^{10^{C_{\locality,\degree}\beta}}$ distinct monomials in the $\lambda'$ that appear in all of the $p_A$ and thus this entire representation has size $\poly(\terms, 1/\eps)$.

\begin{mdframed}
\begin{algorithm}[Learning a Hamiltonian from Gibbs states, \cref{thm:main-ham-learning-theorem}]
    \label{algo:learning-arbitrary-hams}\mbox{}
    \begin{description}
    \item[Input:] Description of $\locality$-local Hamiltonian terms $\{E_a\}_{a \in [\terms]}$ with dual interaction graph degree $\degree$; accuracy and failure probability parameters $\eps, \delta \in (0,1)$; inverse temperature $\beta > 0$; $(\terms^6/\eps^{e^{f(\locality, \degree)\beta}})\log(\terms/\delta)$ copies of the Gibbs state $\rho = \frac{\exp\Paren{-\beta H} }{ \tr\exp\Paren{-\beta H} }$ for an unknown Hamiltonian $H = \sum_a \lambda_a E_a$.
    
    \item[Operation:]\mbox{}
    \begin{enumerate}
    \item Set $\eps_0 = \frac{\eps^{10^{\const\beta} }}{m^3}$ where $\const$ is a sufficiently large constant depending only on $\locality,\degree$;
    \item Set $\ell_0 = 2^{\const\beta}  \log(1/\eps), \ell_1 = 4\locality$;
    \item Define $\calA = \locals_{4^{\const\beta} \ell_0}, \calB = \locals_{\ell_1}$;
    \item For all $A \in \calA$, compute estimates $\wt{\tr} (A \rho)$ of $\tr\Paren{ A \rho }$ to $\eps_0$ error using \cref{lem:basic-estimator};
    \item Consider the following constraint system: 
    \begin{equation*}
        \constraintsystem
    \end{equation*}
    \item Compute a degree-$\bigO{2^{\const\beta}\ell_0}$  pseudo-distribution $\mu$ consistent with $\calC_{\lambda'}$;
    \item Set $\hat{\lambda} = \pexpecf{\mu}{ \lambda' }$;
    \end{enumerate}
    \item[Output:]  $\hat{\lambda}$ such that with probability at least $1-\delta$,  $\abs{\hat{\lambda}_a - \lambda_a}^2 \leq \eps^2$ for all $a \in [\terms]$.
    \end{description}
\end{algorithm}
\end{mdframed}

\paragraph{Analyzing the polynomial system.}

We have now described the bulk of the algorithm (\cref{algo:learning-arbitrary-hams}).
To obtain \cref{thm:main-ham-learning-theorem}, we prove the following intermediate lemmas.  In light of Lemma~\ref{lem:basic-estimator}, we may assume that all of our estimates are accurate to within $\eps_0$ with high probability.
\begin{assumption}[Measuring Gibbs states]\label{assumption:accurate-estimates}
For all $A_1, A_2, A_3 \in \calA$, our estimate satisfies
\[
|\wt{\tr}( A_1A_2 A_3 \rho ) - \tr(A_1A_2A_3 \rho) | \leq \eps_0 \,.
\]   
\end{assumption} 
The first lemma states that when our estimates are accurate, the true coefficients i.e. $\lambda' = \lambda$ are a feasible solution to the system.

\begin{lemma}[Feasibility of the constraint system]
\label{lem:feasibility}
If Assumption~\ref{assumption:accurate-estimates} holds, then the constraints in the system $\calC_{\lambda'}$ in Equation \eqref{eqn:const-system} are satisfied when $\lambda' = \lambda$ (so that $H' = H$).
\end{lemma}

Next, we need to prove soundness.  The proof of soundness is broken into two key steps.  First, we prove that any solution to the system must have $H' = \sum_a \lambda'_a E_a$ nearly commute with $H$.

\begin{remark}[On sum-of-squares proof degree]
All of our sum-of-squares proofs will be of degree $\log(1/\eps)^{2^{f(\locality,\degree)\beta} }$.
For the sake of brevity, we omit the precise degrees from all subsequent statements.
\end{remark}

\begin{lemma}[SoS finds an \textit{almost-commuting} state]
\label{lem:sos-almost-commuting}
    Assume that \cref{assumption:accurate-estimates} holds.  Let $H' = \sum_a \lambda'_a E_a$ and write $\ii[H,H'] = \sum \gamma_b F_b$ for some $2\locality$-$\graph$-local terms $F_b$ where each of $\gamma_b$ are linear expressions in the $\lambda_a'$ (such a representation exists by Lemma~\ref{lem:pauli-commutator}).  Then 
    \begin{equation*}
        \Set{\calC_{\lambda'} } \sststile{}{\lambda'} \Set{ -e^{\calO_{\locality,\degree}(\beta)} \terms^{1.5}\sqrt{\eps_0} \leq \gamma_b  \leq e^{\calO_{\locality,\degree}(\beta)} \terms^{1.5}\sqrt{\eps_0} } \,.
    \end{equation*}
\end{lemma}

Using the above lemma, we can derive identifiability of each parameter in the sum-of-squares proof system:

\begin{lemma}[A sum-of-squares proof of identifiability]
\label{lem:sos-identifiability}
Suppose \cref{assumption:accurate-estimates} holds.  Given the constraint system $\calC_{\lambda'}$, for any term $a^* \in [\terms]$, 
\begin{equation*}
    \Set{\calC_{\lambda'} } \sststile{}{\lambda'} \Set{ (\lambda_{a^*} - \lambda_{a^*}')^2 \leq 2^{\const\beta} \eps }.
\end{equation*}
\end{lemma}

Combining the above lemmas, we can almost prove \cref{thm:main-ham-learning-theorem}.  We first show how to immediately get the same learning guarantee with the same sample complexity but a slightly worse running time.  Later on, we show in \cref{sec:faster-solver} how to improve this runtime and complete the proof of the stronger \cref{thm:main-ham-learning-theorem}.

\begin{theorem}[Weaker version of \cref{thm:main-ham-learning-theorem}]\label{thm:weak-ham-learning-theorem}
    Let $H = \sum_{a=1}^\terms \lambda_a E_a \in \mathbb{C}^{\dims \times \dims}$ be a $\locality$-local Hamiltonian on $\qubits$ qubits whose dual interaction graph has degree $\degree$ (as in \cref{prob:main}).
    Suppose we are given the terms $\{E_a\}_{a \in [\terms]}$, $\eps > 0$, $\delta > 0$, and copies of the Gibbs state $\rho$ at a known inverse temperature $\beta > (10(\degree+1))^{-10}$.
    Then there exists an algorithm that can output estimates $\tilde{\lambda}_a$ such that, with probability $\geq 1-\delta$, $\parens{\hat{\lambda}_a - \lambda_a}^2 \leq \eps^2$ for all $a \in [\terms]$.
    This algorithm uses
    \begin{equation*}
        \bigO[\Big]{(\terms^6/\eps^{e^{f(\locality, \degree)\beta}})\log(\terms/\delta)}
    \end{equation*}
    copies of the Gibbs state and running time
    \begin{equation*}
        \Paren{\terms/\eps}^{\log(1/\eps) \cdot \exp(f(\locality,\degree)\cdot\beta)}.
    \end{equation*}
    where $f(\locality, \degree)$ is a positive function depending only on $\locality$ and $\degree$.
\end{theorem}

\begin{proof}[Proof of Theorem~\ref{thm:weak-ham-learning-theorem}]

Given $\bigO{(\terms^6/\eps^{e^{f(\locality, \degree)\beta}})\log(\terms/\delta)}$ copies of the Gibbs state $\rho$, it follows from \cref{lem:basic-estimator} that with probability at least $1-\delta$, \cref{assumption:accurate-estimates} holds.
Conditioning on this event, it follows from \cref{lem:feasibility} that the constraint system $\calC_{\lambda'}$ is feasible. Given a degree $t = \Omega\Paren{ 2^{\const \beta} \log(1/\eps) } = \Omega\Paren{ 2^{\const \beta} \log(1/\eps) }$ pseudo-distribution $\mu$, where $\const$ is a fixed universal constant that only depends on $\locality, \degree$,  it follows from \cref{lem:sos-identifiability} (and redefining $\eps$ appropriately as $(\eps/2^{\const\beta})^2$) that for all $a \in [\terms]$,
\begin{equation*}
      \Paren{\lambda_{a^*} - \pexpecf{\mu}{\lambda'_{a^*} }}^2 \leq \pexpecf{\mu}{ \Paren{\lambda_{a^*} - \lambda'_{a^*} }^2 } \leq \eps^2 .
\end{equation*}

The running time is dominated by computing a degree-$t$ pseudo-distribution over $\terms$ indeterminates and $\abs{\calA} + \abs{\calB} \leq  \terms^{\mathcal{O}(1) } \Paren{ 1/\eps }^{\exp(\bigO{ \const \beta })}$ many constraints. It follows from~\cref{fact:eff-pseudo-distribution} that the overall running time is at most 
\begin{equation*}
    \Paren{ \terms/\eps }^{ \log(1/\eps) \exp(\bigO{\const \beta } )},
\end{equation*}
which concludes the proof. 
\end{proof}

\begin{remark}[Measurements can be made local] \label{rmk:local}
    The algorithm as stated requires knowing estimates to three-body observables as in \cref{assumption:accurate-estimates}.
    However, we actually only need single-body observables of the form $\tr(A \rho)$ where $A \in \locals_{10^{C_{\locality, \degree} \beta} \log(1/\eps)}$.
    This suffices for the constraints involving $\wt{\tr}\parens{ A \Paren{ H'\rho - \rho H' } } = \wt{\tr}([A, H'] \rho)$, since $[A, H']$ can be expanded into such observables.
    As for the constraints involving $\wt{\tr}\Paren{ B_2  q_{C_{\locality,\degree}\beta, \ell_0}\Paren{-\beta H'| B_1} \rho }  -\wt{\tr}\Paren{ B_1B_2\rho }$, we only need those constraints which are used in the proof of \cref{lem:sos-identifiability}.
    This set of constraints can be inferred from the proof of \cref{lem:sos-test-functions}, where we take $B_1 \in \locals_{\locality}$ and $B_2$ to be a Pauli appearing in the decomposition of $[H - H', B_1] + 0.25[H - H',B_1]_3$.
    This means that we only need constraints where $B_2$ and $B_1$ are spatially close together, and therefore we only need estimates to local observables.
\end{remark}
\section{A proof of feasibility (\texorpdfstring{\cref{lem:feasibility}}{Lemma 6.4})}

Our goal in this section is to prove \cref{lem:feasibility}, which states that our constraint system \cref{eqn:const-system} is satisfied by the ground truth, $\lambda' = \lambda$, when \cref{assumption:accurate-estimates} is satisfied.
Since we are working with the \emph{ground truth} Hamiltonian $H$, we can relax, as our proofs need not be in the sum-of-squares system.

We first use the assumption to bound the error of replacing $\wt{\tr}$ with $\tr$ in the constraints of the sum-of-squares system.
This constraint error comes from the sampling error of our estimates of observables in \cref{sec:accessing-gibbs-states}.
\begin{lemma}\label{lem:error-in-constraints}
    When \cref{assumption:accurate-estimates} holds, for all $B_1, B_2 \in \calB$,
    \begin{align*}
        \abs{ \wt{\tr}(B_1B_2\rho) - \tr(B_1B_2\rho) } &\leq \tfrac12\eps \\
        \abs{ \wt{\tr}(B_2q_{\const\beta, \ell_0}(-\beta H| B_1) \rho) - \tr(B_2q_{\const\beta, \ell_0}(-\beta H| B_1) \rho) } & \leq \tfrac12 \eps
    \end{align*}
\end{lemma}
\begin{proof}
\cref{assumption:accurate-estimates} implies that
\[
    \abs{ \wt{\tr}(B_1B_2\rho) - \tr(B_1B_2\rho) } \leq 0.5\eps
\]
for all $B_1, B_2$. Next, by \cref{lem:cluster-expansion} and \cref{claim:simple-coeff-bound}, we can write $q_{\const\beta, \ell_0}(-\beta H| B_1)$ as a linear combination of the form
\[
    \sum_{A \in \calA} c_A A
\]
where $\calA \subset \locals_{2^{\const\beta + 1}  \ell_0 \locality}$ and the coefficients satisfy $\sum_{A \in \calA}|c_A| \leq ((1 + \beta) \degree)^{2^{\const\beta + 5}\ell_0}$.  Thus, conditioned on the event that \cref{assumption:accurate-estimates} holds, we have
\begin{equation*}
    | \wt{\tr}(B_2q_{\const\beta, \ell_0}(-\beta H| B_1) \rho) - \tr(B_2q_{\const\beta, \ell_0}(-\beta H| B_1) \rho)|  \leq   0.5\eps  \,. \qedhere
\end{equation*}
\end{proof}

We now prove that the first two constraints are satisfied, when $\wt{\tr}$ is replaced by $\tr$.

\begin{lemma}[Feasibility of commutators]
\label{lem:feasibility-of-commutators}
Let $A$ be a matrix and let $\rho$ be a Gibbs state associated with the $\locality$-local Hamiltonian $H = \sum_{a\in [\terms]} \lambda_a E_a$.
Then for all $a\in [\terms]$, $\abs{\lambda_a} \leq 1$ and $\tr(A(H\rho - \rho H)) = 0$.
\end{lemma}
\begin{proof}
By assumption, $\abs{\lambda_a} \leq 1$.  Next, since $H$ commutes with $\rho$, for all $A$,
\begin{equation*}
    \tr(A(H\rho - \rho H)) = 0. \qedhere
\end{equation*}
\end{proof}

It remains to prove that $H$ satisfies the final set of constraints in \cref{eqn:const-system}, which states that $\tr\parens{B_2 q(-\beta H \vert  B_1) \rho} \approx \tr(B_1B_2 \rho)$ when $q$ is a good flat approximation to the exponential $e^x$ (recall \cref{def:flat-exp}).
To provide intuition for this statement, consider working in the basis where $H$ is diagonal and its diagonal entries (and eigenvalues) are $\sigma_i$.
Then
\begin{align*}
    \tr(\paren{B_2 q(-\beta H \vert B_1) - B_1 B_2} \rho)
    &= \tr(\paren{B_2 (B_1 \circ \braces{q(-\beta (\sigma_i - \sigma_j))}_{ij}) - B_1B_2} \rho) \\
    &\approx \tr(\paren{B_2 (B_1 \circ \braces{e^{-\beta(\sigma_i - \sigma_j)}}_{ij}) - B_1B_2} \rho) \\
    &= \tr(\paren{B_2 \rho B_1 \rho^{-1} - B_1 B_2} \rho) \\
    &= \tr(B_2 \rho B_1 - B_1 B_2 \rho) \\
    &= 0,
\end{align*}
where the second equality uses that $q$ is a good approximation to $e^x$.
In the following lemma, we make the above intuition precise, showing that when $q$ is a sufficiently good flat exponential approximation, this derivation indeed holds up to some small error.

\begin{lemma}[Feasibility of the polynomial approximation]
\label{lem:feasibility-of-poly-approx}
Let $H$ be a $\locality$-local Hamiltonian with interaction degree $\degree$ and let $\rho$ be a Gibbs state of $H$ at inverse temperature $\beta$.
For $A, B \in \locals_{4\locality}$ and a polynomial $q$ that is a $(\kappa, \eta, \eps)$-flat exponential approximation where 
\begin{enumerate}[label=\textup{(\alph*)}]
    \item $\kappa \geq C_1 \beta \log(1/\eps)$ for some sufficiently large constant $C_1$ depending only on $\locality,\degree$,
    \item $\eta < 1/(C_2\beta)$ for some sufficiently large constant $C_2$ depending only on $\locality, \degree$,
\end{enumerate}
we have 
\begin{equation*}
 |\tr\Paren{ B q\Paren{-\beta H | A}  \rho } - \tr\Paren{AB\rho}| \leq \tfrac12\eps \norm{A} \norm{B}.
\end{equation*}
\end{lemma}
\begin{proof}
By a change of basis, we take $H$ to be diagonal with $H_{ii} = \sigma_i$.
Consequently, $\rho$ is also diagonal with $\rho_{ii} = e^{-\beta \sigma_i} / \tr e^{-\beta H}$, and
\begin{equation} \label{eqn:feasibility-pt-1}
    \begin{aligned}[b]
    \tr\Paren{B q(-\beta H | A) \rho} - \tr\Paren{AB \rho}
    &= \tr\Paren{B \cdot q(-\beta H | A) \cdot \rho} - \tr\Paren{B \rho A} \\
    &= \frac{1}{\tr e^{-\beta H}} \sum_{i, j = 1}^{\dims} \parens[\Big]{B_{ji} [q(-\beta H | A)]_{ij} e^{-\beta \sigma_j} - B_{ji} A_{ij} e^{- \beta \sigma_i}} \\
    &= \frac{1}{\tr e^{-\beta H}} \sum_{i, j = 1}^{\dims} \parens[\Big]{B_{ji} A_{ij} q(\beta (\sigma_j - \sigma_i)) e^{-\beta \sigma_j} - B_{ji} A_{ij} e^{- \beta \sigma_i}} \\
    &= \frac{1}{\tr e^{-\beta H}} \sum_{i, j = 1}^{\dims} B_{ji}A_{ij}\parens[\Big]{q(\beta(\sigma_j - \sigma_i)) - e^{\beta (\sigma_j - \sigma_i)}} e^{-\beta \sigma_j},
    \end{aligned}
\end{equation}
where the first equality follows from cyclicity of trace, the second uses its definition, the third uses \cref{def:polynomial-nested-commutator} and \cref{lem:commutator-to-polynomial}, and the fourth groups relevant terms together.
Now, we use that $q$ is a $(\kappa, \eta, \eps)$-flat exponential approximation to bound this quantity.
Let $L = \kappa/(3\beta)$ and let $S_a = \braces{i \in [\dims] : \sigma_i \in [aL, (a+1)L)}$, so that $[\dims] = \bigsqcup_a S_a$ (where the index goes from $-\infty$ to $\infty$.)
Continuing from \cref{eqn:feasibility-pt-1},
\begin{equation}
\label{eqn:splitting-into-level-sets}
    \begin{aligned}[b]
        & \frac{1}{\tr e^{-\beta H}} \sum_{i, j = 1}^{\dims} B_{ji}A_{ij}\parens[\Big]{q(\beta(\sigma_j - \sigma_i)) - e^{\beta (\sigma_j - \sigma_i)}} e^{-\beta \sigma_j} \\
        & = \frac{1}{\tr e^{-\beta H}} \sum_a \sum_b \sum_{i \in S_a} \sum_{j \in S_b}B_{ji} A_{ij}\parens[\Big]{q(\beta(\sigma_j - \sigma_i)) - e^{\beta (\sigma_j - \sigma_i)}} e^{-\beta \sigma_j} \\
        &= \sum_\alpha \frac{1}{\tr e^{-\beta H}}\sum_a \sum_{i \in S_a} \sum_{j \in S_{a + \alpha}} B_{ji} A_{ij}\parens[\Big]{q(\beta(\sigma_j - \sigma_i)) - e^{\beta (\sigma_j - \sigma_i)}} e^{-\beta \sigma_j}
    \end{aligned}
\end{equation}
Now, we bound \eqref{eqn:splitting-into-level-sets} by breaking into cases depending on the value of $\alpha$.
First, suppose $-2 \leq \alpha \leq 2$.
For $i \in S_a$ and $j \in S_{a + \alpha}$, the input to $q$ is bounded by $\kappa$: $\beta \abs{\sigma_j - \sigma_i} < \beta(\alpha + 1)L \leq \kappa$.
So it follows from the approximation guarantee (\cref{def:flat-exp}) that
\begin{equation}
    \abs{ q(\beta(\sigma_j - \sigma_i)) - e^{\beta(\sigma_j - \sigma_i)}} \leq \eps,
\end{equation}
and thus
\begin{equation}\label{eq:submatrix-case1}
\begin{aligned}[b]
    & \abs[\Big]{\sum_{i \in S_a}\sum_{j \in S_{a+\alpha}} B_{ji} A_{ij}\parens[\Big]{q(\beta(\sigma_j - \sigma_i)) - e^{\beta (\sigma_j - \sigma_i)}} e^{-\beta \sigma_j}} \\
    & \leq \eps \sum_{i \in S_a}\sum_{j \in S_{a+\alpha}} \abs{B_{ji} A_{ij}}e^{-\beta \sigma_j} \\
    & \leq \eps  \parens[\Big]{\max_{j\in S_{a + \alpha}} \sum_{ i \in S_a} |B_{ji}A_{ij}| } \sum_{j \in S_{a + \alpha}} e^{-\beta \sigma_j} \\
    &\leq \eps \parens[\Big]{\max_{j \in S_{a + \alpha}} \sum_{i \in S_a} |B_{ji}|^2 \sum_{i \in S_a} |A_{ij}|^2}^{1/2} \sum_{j \in S_{a + \alpha}} e^{-\beta \sigma_j} \\
    &\leq \eps \norm{A_{S_{a}, S_{a+\alpha}}}\norm{B_{S_a,S_{a+\alpha}}} \sum_{j \in S_{a + \alpha}} e^{-\beta \sigma_j},
\end{aligned}
\end{equation}
where $A_{S_a,S_b}$ and $B_{S_a, S_b}$ are the submatrices of $A$ and $B$ indexed by $i \in S_a$ and $j \in S_b$ (in the eigenbasis of $H$).
Next, suppose $\alpha \geq 3$.
For $i \in S_a$ and $j \in S_{a + \alpha}$, $\sigma_j > \sigma_i$, so by the guarantee on $q$ (\cref{def:flat-exp}), $\abs{q(\beta(\sigma_j - \sigma_i))} \leq \max(1, e^{\beta(\sigma_j - \sigma_i)}) e^{\eta \abs{\beta(\sigma_j - \sigma_i)}} = e^{(1+\eta) \beta(\sigma_j - \sigma_i)}$.
Thus,
\begin{equation}\label{eq:submatrix-case2}
\begin{aligned}[b]
    &\abs[\Big]{\sum_{i \in S_a}\sum_{j \in S_{a+\alpha}} B_{ji} A_{ij} \parens[\Big]{q(\beta(\sigma_j - \sigma_i)) - e^{\beta (\sigma_j - \sigma_i)}} e^{-\beta \sigma_j}}  \\
    &\leq  \sum_{i \in S_a}\sum_{j \in S_{a+\alpha}} \abs{B_{ji} A_{ij}} \parens[\Big]{e^{(1+\eta) \beta(\sigma_j - \sigma_i)} + e^{\beta(\sigma_j - \sigma_i)}} e^{-\beta \sigma_j} \\
    &=  \sum_{i \in S_a}\sum_{j \in S_{a+\alpha}} \abs{B_{ji} A_{ij}} \parens[\Big]{e^{\eta \beta(\sigma_j - \sigma_i)} + 1} e^{-\beta \sigma_i} \\
    &\leq  (e^{ \eta \beta (\alpha + 1) L} + 1) \parens[\Big]{\max_{i \in S_{a}} \sum_{ j \in S_{a + \alpha}} |B_{ji}A_{ij}| } \sum_{i \in S_{a}} e^{-\beta \sigma_i}  \\
    &\leq (e^{ \eta \beta (\alpha + 1) L} + 1) \norm{A_{S_{a}, S_{a+\alpha}}}\norm{B_{S_a,S_{a+\alpha}}} \sum_{i \in S_{a}} e^{-\beta \sigma_i} \,.
\end{aligned}
\end{equation}
Finally, suppose $\alpha \leq -3$.
For $i \in S_a$ and $j \in S_{a+\alpha}$, we have $\sigma_j < \sigma_i$, so by the same guarantee on $q$ (\cref{def:flat-exp}), $\abs{q(\beta(\sigma_j - \sigma_i))} \leq e^{\eta \beta(\sigma_i - \sigma_j)}$.
Thus,
\begin{equation}\label{eq:submatrix-case3}
\begin{aligned}[b]
    &\abs[\Big]{ \sum_{i \in S_a}\sum_{j \in S_{a+\alpha}} B_{ji} A_{ij} \parens[\Big]{q(\beta(\sigma_j - \sigma_i)) - e^{\beta (\sigma_j - \sigma_i)}} e^{-\beta \sigma_j} } \\
    &\leq \sum_{i \in S_a}\sum_{j \in S_{a+\alpha}} \abs{B_{ji} A_{ij}} (e^{-\beta \sigma_j}e^{\eta \beta (\abs{\alpha} +1)L} + e^{-\beta \sigma_i})\\
    &\leq \norm{B_{S_a,S_{a+\alpha}}}\norm{A_{S_{a}, S_{a+\alpha}}}\parens[\Big]{\sum_{i \in S_{a}} e^{-\beta \sigma_i}  + e^{\eta \beta (|\alpha| + 1)L} \sum_{j \in S_{a + \alpha}} e^{-\beta \sigma_j} } \,.
\end{aligned}
\end{equation}
By \cref{lem:akl}, for an absolute constant $c > 1$ depending only on $\locality,\degree$,
\begin{equation} \label{eqn:feasibility-akl}
\begin{split}
    \norm{A_{S_a, S_{a+\alpha}}} &\leq \norm{A} e^{c-(\abs{\alpha} - 1)L/c} \\
    \norm{B_{S_a, S_{a+\alpha}}} &\leq \norm{B} e^{c-(\abs{\alpha} - 1)L/c} \,.
\end{split}
\end{equation}
We choose $\eta = 1/(C_2 \beta)$ small enough to ensure that $\eta \beta(\abs{\alpha} + 1) \leq (\abs{\alpha} - 1)/c$ whenever $\abs{\alpha} \geq 3$.
Combining \cref{eq:submatrix-case1,eq:submatrix-case2,eq:submatrix-case3,eqn:feasibility-akl} with the above, we get the following bound on \cref{eqn:splitting-into-level-sets}.
\begin{align}
    &\abs*{\sum_\alpha \frac{1}{\tr e^{-\beta H}}\sum_a \sum_{i \in S_a} \sum_{j \in S_{a + \alpha}} B_{ji} A_{ij}\parens[\Big]{q(\beta(\sigma_j - \sigma_i)) - e^{\beta (\sigma_j - \sigma_i)}} e^{-\beta \sigma_j}} \tag*{from \cref{eqn:splitting-into-level-sets}} \\
    &\leq \frac{1}{\tr e^{-\beta H}}\sum_a\Big\lparen
        \sum_{\alpha=-2}^2 \eps \norm{A_{S_{a}, S_{a+\alpha}}}\norm{B_{S_a,S_{a+\alpha}}} \sum_{j \in S_{a + \alpha}} e^{-\beta \sigma_j} \tag*{by \cref{eq:submatrix-case1}} \\
        &\hspace{4em} + \sum_{\alpha = 3}^\infty (e^{ \eta \beta (\alpha + 1) L} + 1) \norm{A_{S_{a}, S_{a+\alpha}}}\norm{B_{S_a,S_{a+\alpha}}} \sum_{i \in S_{a}} e^{-\beta \sigma_i} \tag*{by \cref{eq:submatrix-case2}}\\
        &\hspace{4em} + \sum_{\alpha = -3}^{-\infty} \norm{B_{S_a,S_{a+\alpha}}}\norm{A_{S_{a}, S_{a+\alpha}}}\parens[\Big]{\sum_{i \in S_{a}} e^{-\beta \sigma_i}  + e^{\eta \beta (|\alpha| + 1)L} \sum_{j \in S_{a + \alpha}} e^{-\beta \sigma_j} }\Big\rparen \tag*{by \cref{eq:submatrix-case3}}\\
    &\leq \frac{\norm{A}\norm{B}}{\tr e^{-\beta H}}\sum_a\Big\lparen
        \sum_{\alpha=-2}^2 \eps \sum_{j \in S_{a + \alpha}} e^{-\beta \sigma_j} \nonumber\\
        &\hspace{4em} + \sum_{\alpha = 3}^\infty (e^{ \eta \beta (\alpha + 1) L} + 1) e^{2(c-(\abs{\alpha} - 1)L/c)} \sum_{i \in S_{a}} e^{-\beta \sigma_i} \nonumber\\
        &\hspace{4em} + \sum_{\alpha = -3}^{-\infty} e^{2(c-(\abs{\alpha} - 1)L/c)}\parens[\Big]{\sum_{i \in S_{a}} e^{-\beta \sigma_i}  + e^{\eta \beta (|\alpha| + 1)L} \sum_{j \in S_{a + \alpha}} e^{-\beta \sigma_j} }\Big\rparen \tag*{by \cref{eqn:feasibility-akl}}\\
    &= \norm{A}\norm{B}\Big\lparen
        \sum_{\alpha=-2}^2 \eps \nonumber\\
        &\hspace{4em} + \sum_{\alpha = 3}^\infty (e^{ \eta \beta (\alpha + 1) L} + 1) e^{2(c-(\abs{\alpha} - 1)L/c)} \nonumber\\
        &\hspace{4em} + \sum_{\alpha = -3}^{-\infty} e^{2(c-(\abs{\alpha} - 1)L/c)}\parens{1 + e^{\eta \beta (|\alpha| + 1)L} }\Big\rparen \tag*{by $\sum_a \sum_{i \in S_{a + \alpha}} e^{-\beta \sigma_i} = \tr(e^{-\beta H})$}\\
    &\leq (5+e^{2c})\norm{A} \norm{B} \parens[\Big]{\eps + \sum_{\alpha, |\alpha| \geq 3} (1 + e^{\eta \beta (|\alpha| + 1)L}) e^{-2(\abs{\alpha} - 1)L/c} } \nonumber\\
    &\leq (5+e^{2c}) \tr(e^{-\beta H}) \norm{A} \norm{B} \left(\eps + \sum_{\alpha, |\alpha| \geq 3}  2e^{-(|\alpha| - 1)L/c}  \right) \nonumber\\
    &\leq \tfrac12\eps \norm{A} \norm{B}
    \tag*{taking $L =\kappa/(3\beta) \geq  C_1 \log(1/\eps)$}
\end{align}
The last step is where we need $C_1$ to be sufficiently large in terms of $\locality,\degree$.
Plugging this back into the relation at the beginning, we conclude
\[
|\tr\parens{B q(-\beta H |  A) \rho} - \tr\parens{AB \rho} | \leq \tfrac12\eps \norm{A}\norm{B}
\]
as desired.
\end{proof}

Now we can complete the proof of \cref{lem:feasibility}.
\begin{proof}[Proof of \cref{lem:feasibility}]
The first two constraints are satisfied by \cref{lem:error-in-constraints,lem:feasibility-of-commutators}:
\begin{align}
    \abs{\wt{\tr}(A_1 A_2(H \rho - \rho H))}^2
    \leq 2\abs{\wt{\tr}(A_1 A_2(H \rho - \rho H))}^2 + 2\abs{\tr(A_1 A_2(H \rho - \rho H))}^2
    \leq \tfrac12\eps^2.
\end{align}
As for the final constraint, by \cref{thm:exp-fancy-approx}, and as stated in \eqref{eq:q-main-flatness}, the polynomial $q_{C_{\locality,\degree}\beta, \ell_0}$ is a flat exponential approximation with parameters
\[
    \parens[\Big]{ 0.001 \cdot 2^{\const\beta} \log(1/\eps) , \frac{5}{\const\beta} , 0.001\eps } \,.
\]
Thus, by \cref{lem:feasibility-of-poly-approx}, we have
\[
    \abs{\tr(B_2q_{\const\beta, \ell_0}(-\beta H | B_1) \rho - \tr(B_1B_2\rho)) } \leq 0.1\eps
\]
for all $B_1, B_2 \in \calB$.
Combining the above with \cref{lem:error-in-constraints} immediately implies that all constraints are satisfied.
\end{proof}
\section{A sum-of-squares proof that the commutator is small (\texorpdfstring{\cref{lem:sos-almost-commuting}}{Lemma 6.6})}\label{sec:commutator-is-small}

In this section, we prove \cref{lem:sos-almost-commuting}.  We will only use a subset of the constraints, namely the commutator constraints involving $\wt{\tr}(A(H'\rho - \rho H'))$, to prove this.  The proof will involve two parts.  First, we will show that the constraints imply that  $\tr([H,H'](H'\rho - \rho H'))$ is small.  Then, we show that $\tr([H,H'](H'\rho - \rho H'))$  being small actually implies that $[H,H']$ must be small.

\begin{lemma}[Trace inner product is bounded]\label{lem:commutator-bound1}
Under \cref{assumption:accurate-estimates}, we have
\begin{equation*}
   \Set{ \calC_{\lambda'} } \sststile{}{\lambda'} \Set{ \left\lvert \tr\paren{[H, H'](H'\rho  - \rho H' )} \right\rvert   \leq m^2 (10^\locality\degree)^5 \eps_0  } \,.
\end{equation*}
\end{lemma}
\begin{proof}
By \cref{lem:pauli-commutator}, we can write $[H,H'] = \sum_{F \in \locals_{2\locality}} f_a(\lambda') F$ where each $f_a$ is a linear functions of $\lambda_1', \dots , \lambda_m'$.  Furthermore we have the following properties
\begin{itemize}
    \item Each $f_a$ has nonzero coefficients on at most $\degree$ terms
    \item All of the coefficients in each of the $f_a$ have magnitude at most $1$
\end{itemize}  
The total number of terms $F \in \locals_{2\locality}$ is at most $\terms(10^\locality \degree)^2$ (recall \cref{coro:number-of-local-paulis}).
We have the constraint $-\eps_0 \leq \wt{\tr}\Paren{ F (H'\rho - \rho H') } \leq \eps_0$ for each $F \in \locals_{2\locality}$, and thus summing over the $F$'s 
\begin{equation}
\label{eqn:upper-bound-trace-commutator-estimate}
\begin{split}
     \Set{ \calC_{\lambda'} } \sststile{}{\lambda'} \Biggl\{
        \wt{\tr}\Paren{ [H, H'] (H'\rho - \rho H')  }
        & = \wt{\tr}\parens[\Big]{ \sum_{F \in \locals_{2\locality}} f_a(\lambda') F \Paren{ H' \rho - \rho H' } } \\
        & \leq \degree \sum_{F \in \calP_{2\locality} } \wt{\tr}\Paren{ F \Paren{H'\rho - \rho H'}}   \leq m 10^{2\locality}\degree^3  \eps_0 \Biggr\}  ,
\end{split}
\end{equation}
where the first inequality follows from the constraint that  $-1 \leq \lambda'_i \leq 1$ and that each $f_a$ has at most $\degree$ non-zeros, and the second inequality follows from there being at most $\terms(10^\locality\degree)^2$ terms $F \in \calP_{2\locality}$ and the corresponding trace being at most $\eps_0$ due to the constraints. Similarly, we can derive
\begin{equation}
\label{eqn:lower-bound-trace-commutator-estimate}
\begin{split}
     \Set{ \calC_{\lambda'} } \sststile{}{\lambda'} \braces[\Bigg]{ \wt{\tr}\Paren{ [H, H'] (H'\rho - \rho H')  }  \geq -\terms 10^{2\locality}\degree^3  \eps_0 }.
\end{split}
\end{equation}


Next, again expanding $H$ and $H'$ in the Pauli basis, we have
\begin{equation}
\label{eqn:expansion-quadratic-in-lambda-true}
    \sststile{}{\lambda'} \Set{\tr\Paren{[H, H'](H'\rho  - \rho H' )} = \sum_{ F_1, F_2 \in \locals_{2\locality}} g_{F_1, F_2}(\lambda') \tr(F_1F_2\rho)   },
\end{equation}
where $g_{F_1,F_2}(\lambda')$ are quadratic functions of the $\lambda'$ each with at most $\degree$ nonzero monomials and coefficients between $\pm 2$. Similarly, 
\begin{equation}
\label{eqn:expansion-quadratic-in-lambda-estimate}
    \sststile{}{\lambda'} \Set{\wt{\tr}\Paren{[H, H'](H'\rho  - \rho H' )} = \sum_{F_1, F_2 \in \locals_{2\locality}} g_{F_1, F_2}(\lambda') \wt{\tr}\Paren{ F_1 F_2\rho }   }.
\end{equation}
Using Assumption~\ref{assumption:accurate-estimates} and the constraint that $-1 \leq \lambda_i' \leq 1$,
we have
\begin{equation}
\begin{split}
    \Set{\calC_{\lambda'} } \sststile{}{\lambda'} \Biggl\{ & \tr\Paren{ [H, H'](H'\rho  - \rho H' )  }  \\
    & \leq \sum_{F_1 , F_2 \in \locals_{2\locality} } g_{F_1, F_2}(\lambda') \wt{\tr}\Paren{F_1 F_2 \rho}  + \eps_0 \cdot \sum_{F_1 , F_2 \in \locals_{2\locality} } g_{F_1, F_2}(\lambda')    \\
    & =  \wt{\tr}\Paren{ [H, H'](H'\rho  - \rho H' )  } + \eps_0 \sum_{F_1 , F_2 \in \calP_{2\locality} } g_{F_1, F_2}(\lambda') \\
    & \leq \terms^2 (10^\locality\degree)^5 \eps_0 \Biggr\} ,
\end{split}
\end{equation}
where the first inequality follows from \eqref{eqn:expansion-quadratic-in-lambda-true}, \eqref{eqn:expansion-quadratic-in-lambda-estimate} and 
 \cref{assumption:accurate-estimates} and the second follows from \eqref{eqn:upper-bound-trace-commutator-estimate}, the fact that $g$ has coefficients between $[-2,2]$ and there are at most $\terms^2 (10^\locality\degree)^4$ terms $F_1,F_2 \in \calP_{2\locality}$.  Similarly,  using the lower bound estimate in \eqref{eqn:lower-bound-trace-commutator-estimate} we have 
\begin{equation}
     \Set{\calC_{\lambda'} } \sststile{}{\lambda'} \Biggl\{  \tr\Paren{ [H, H'](H'\rho  - \rho H' )  }  \geq - m^2 (10^\locality\degree)^5 \eps_0 \Biggr\},
\end{equation}
which completes the proof.
\end{proof}

Next, we move on to the second part of the proof.  First, we prove the following inequality relating   $\tr\paren{[H, H'](H'\rho  - \rho H')}$ to $\tr\Paren{ (\ii[H', H ] )^2 \rho }$. We note that we consider the expression $\ii[H', H]$ to ensure that the matrix is Hermitian and the coefficients of the $\lambda_i'$ are real. 

\begin{lemma}[Lower bounding the commutator at arbitrary temperature]\label{lem:commutator-bound2}
Given $0<\beta$,  $H' = \sum_{i \in [m]} \lambda_i' E_i $,  $H = \sum_{i \in [m] } \lambda_i E_i$ and $\rho = e^{-\beta H}$, we have  
    \begin{equation*}
     \sststile{}{\lambda'} \Set{  \frac{\beta}{1+ 2\beta \norm{H} }  \tr\Paren{(\ii[H', H ] )^2 \rho }    \leq  \tr\Paren{ [H', H]  \Paren{ H'\rho - \rho H' } } } .
\end{equation*}
\end{lemma}
\begin{proof}

Consider the basis where $H$ is diagonal and let its eigenvalues be $\sigma_i$.  Let $Z = \tr(e^{-\beta H})$.  Then, 
\begin{equation}
\label{eqn:expand-h'-rho-rho-h'}
    \begin{split}
    \sststile{}{\lambda'} \Biggl\{ \tr\paren{[H', H](H' \rho  - \rho H' )}
    &= \tr\paren{(H'H - HH')(H' \rho - \rho H')} \\
    &= \frac{1}{Z}\tr\paren{(H' \circ \braces{\sigma_j - \sigma_i}_{ij})(H' \circ \braces{e^{-\beta \sigma_j} - e^{-\beta \sigma_i}}_{ij})} \\ 
    &= \frac{1}{Z} \sum_{i, j} H_{ij}'H_{ji}'(\sigma_j - \sigma_i)(e^{-\beta \sigma_i} - e^{-\beta \sigma_j}) \\
    &= \frac{1}{Z} \sum_{i, j} \abs{H_{ij}'}^2(\sigma_j - \sigma_i)(e^{-\beta \sigma_i} - e^{-\beta \sigma_j}) \\
    &= \frac{1}{Z} \sum_{i, j} \abs{H_{ij}'}^2(\sigma_j - \sigma_i)(1 - e^{-\beta (\sigma_j - \sigma_i)})e^{-\beta \sigma_i}\Biggr\}
    \end{split}
\end{equation}

By a similar argument,
\begin{equation}
\label{eqn:expand-h'-squared-rho}
    \begin{split}
        \sststile{}{\lambda'} \Biggl \{ \tr\Paren{(\ii[H', H ] )^2 \rho } 
    &= \frac{1}{Z}\tr\paren{(H' \circ \braces{\sigma_j - \sigma_i}_{ij})(H' \circ \braces{(\sigma_i - \sigma_j)e^{-\beta \sigma_j}}_{ij})} \\
    &= \frac{1}{Z} \sum_{i, j} \abs{H_{ij}'}^2(\sigma_j - \sigma_i)^2 e^{-\beta \sigma_i} \Biggr\} \,.
    \end{split}
\end{equation}

We observe that, since $e^x \geq 1 + x$, $e^{-x} \leq \frac{1}{1+x}$, and so
\begin{align*}
    x(1 - e^{-\beta x}) \geq \abs{x}(1 - e^{-\beta \abs{x}}) \geq \frac{\beta x^2}{1+\abs{\beta x}}.
\end{align*}
Applying this inequality with $x = \sigma_j - \sigma_i$, which are constants in the sum-of-squares proof system, and substituting back into \eqref{eqn:expand-h'-rho-rho-h'}, we have,

\begin{equation}
    \begin{split}
        \sststile{}{\lambda'} \Biggl\{ \tr\paren{[H', H](H' \rho  - \rho H')}
    &= \frac{1}{Z} \sum_{i, j} \abs{H_{ij}'}^2(\sigma_j - \sigma_i)(1 - e^{-\beta (\sigma_j - \sigma_i)})e^{-\beta \sigma_i} \\
    &\geq \frac{1}{Z} \sum_{i, j} \abs{H_{ij}'}^2(\sigma_j - \sigma_i)^2\frac{\beta}{1 + \abs{\beta}\abs{\sigma_j - \sigma_i}}e^{-\beta \sigma_i} \\
    &\geq \frac{\beta}{1 + 2 \beta \norm{H}}\frac{1}{Z} \sum_{i, j} \abs{H_{ij}'}^2(\sigma_j - \sigma_i)^2e^{-\beta \sigma_i} \\
    &= \frac{\beta}{1 + 2 \beta \norm{H}}\tr\Paren{(\ii[H', H ] )^2 \rho }\Biggr\}  ,
    \end{split}
\end{equation}
where the second inequality follows from $\abs{\sigma_j - \sigma_i} \leq 2 \norm{H}$ and the last equality follows from \eqref{eqn:expand-h'-squared-rho}. 
\end{proof}

\begin{lemma}\label{lem:commutator-bound3}
Let $H' = \sum_{a=1}^{\terms} \lambda'_a E_a$ and write $\ii[H,H'] = \sum_{b = 1}^{\abs{\locals_{2\locality}}} \gamma_b F_b$ as a linear combination of $2\locality$-local Pauli matrices $F_b$ where each of the $\gamma_b$ are linear expressions in the $\lambda_a'$.  Then, for each $\gamma_b$, we have
\[
\Set{ \calC_{\lambda'} }  \sststile{}{\lambda'} \Set{ \gamma_b^2 \leq e^{\calO_{\locality,\degree}(\beta)} \tr((\ii[H',H])^2 \rho) } .
\]
\end{lemma}
\begin{proof}
Note that $\ii[H',H]$ is Hermitian and is $2\locality$-local.  Thus, by \cref{lem:large-marginals}, we have that for any real values for the $\lambda_i'$, the inequality 
\[
\gamma_b^2 \leq e^{\calO_{\locality,\degree}(\beta)} \tr((\ii[H',H])^2 \rho)
\]
holds.  Now both sides of the above are quadratic expressions in the $\lambda_i'$ so by \cref{fact:nonnegative-quadratic}, the difference between the two sides can be written as a sum of squares of linear functions of the $\lambda_i'$.  
Thus, 
\begin{equation*}
    \Set{ \calC_{\lambda'} }  \sststile{}{\lambda'}\Set{ \gamma_b^2  \leq e^{\calO_{\locality,\degree}(\beta)} \tr((\ii[H',H])^2 \rho)  }
\end{equation*}
as desired.
\end{proof}

Now we can finish the proof of \cref{lem:sos-almost-commuting} by combining the previous lemmas.

\begin{proof}[Proof of \cref{lem:sos-almost-commuting}]
 Combining \cref{lem:commutator-bound1}, \cref{lem:commutator-bound2}, and \cref{lem:commutator-bound3}, we get 
 \[
 \Set{ \calC_{\lambda'} }  \sststile{}{\lambda'} \Set{ \gamma_b^2 \leq e^{\calO_{\locality,\degree}(\beta)} \terms^3 \eps_0 } 
 \]
 since $\norm{H} \leq m$ (and we can adjust the $O_{\locality,\degree}(1)$ in the exponent to absorb the other factors).  It follows from~\cref{fact:squared-value-to-magnitude} that
 \[
 \Set{ \gamma_b^2 \leq e^{\calO_{\locality,\degree}(\beta)} \terms^3 \eps_0 }   \sststile{}{\lambda'} \Set{ -e^{\calO_{\locality,\degree}(\beta)} \terms^{1.5}\sqrt{\eps_0} \leq \gamma_b \leq e^{\calO_{\locality,\degree}(\beta)} \terms^{1.5}\sqrt{\eps_0} }
 \]
and this is exactly the desired statement. 
\end{proof}

\section{A sum-of-squares proof of identifiability (\texorpdfstring{\cref{lem:sos-identifiability}}{Lemma 6.7})}
\label{sec:sos-identifiability}

In this section, we prove \cref{lem:sos-identifiability}.  At a high level, we will rely on properties of $q_{C_{\locality,\degree}\beta, \ell_0}$ proven in Section~\ref{sec:poly-approx}.  However, since we are working with commutator polynomials, we will need to invoke the translation between polynomials and commutators in Section~\ref{sec:poly-commutators} at each step.  First, it is critical that $H'$ and $H$ almost commute so that the ``error" terms that appear in the translation (those on the RHS of Theorem~\ref{thm:polynomial-equivalence}) are small.  We make this precise in the following subsection.

\subsection{Bounding error terms: polynomials to nested commutators }

We begin by showing the following lemma:

\begin{lemma}[Bounding switched commutators of type 3]\label{lem:commutator-coeff-bound}
Let $S,T \in \{0,1\}^*$ be arbitrary sequences of length at most $\ell$.  Let $A \in \locals_{\locality}$. 
 Then we can write
\[
\ii^{|S| + |T|}\left[ (H,H')_{S}, \left[[H',H], \left[(H,H')_{T},A \right]\right] \right] = \sum_{G_c \in \calP_{4(\ell+1)\locality} } \zeta_c G_c
\]
where $\zeta_c$ are polynomials in the $\lambda'$ and the terms $G_c $ are in $ \locals_{4(\ell + 1)\locality}$ and have distance at most $4(\ell + 1)\locality$ from the support of $A$.  If Assumption~\ref{assumption:accurate-estimates} holds then  
\begin{equation*}
\begin{split}
 \Set{\calC_{\lambda'} } \sststile{}{\lambda'} \Bigl\{ & - (|S| + |T| + 1)! (2\degree)^{4\ell}e^{\calO_{\locality,\degree}(\beta)} m^{1.5}\sqrt{\eps_0}   \\
 & \leq \zeta_c \leq  (|S| + |T| + 1)! (2\degree)^{4\ell}e^{\calO_{\locality,\degree}(\beta)} m^{1.5}\sqrt{\eps_0} \Bigr\} \,.
\end{split}
\end{equation*}
Note that the $\ii^{|S| + |T|}$ is to make the expression Hermitian so that the $\zeta_c$ are real polynomials.
\end{lemma}
\begin{proof}
We can write $\ii[H,H'] = \sum_{b} \gamma_b F_b$ for $F_b \in \locals_{2\locality}$ and by \cref{lem:sos-almost-commuting}
\begin{equation}\label{eq:commutator-small}
\Set{\calC_{\lambda'} } \sststile{}{\lambda'} \Set{ -e^{\calO_{\locality,\degree}(\beta)} m^{1.5}\sqrt{\eps_0} \leq \gamma_b \leq e^{\calO_{\locality,\degree}(\beta)} m^{1.5}\sqrt{\eps_0} } \,.
\end{equation}
For notational convenience, for an index $j \in \{1,2, \dots , |S| + |T|\}$, let $\lambda_{a, [j]}$ be equal to $\lambda_a$ if the $j$th entry of the sequence $S T$ (concatenated) is $0$ and equal to $\lambda_a'$ otherwise.  Now we can apply \cref{lem:cluster-expansion} (with $\locality \leftarrow 2\locality$ and $\degree \leftarrow 10\degree^2$, since $[H,H']$ is $2\locality$-local) to write
\begin{equation*}
\begin{split}
&\ii^{|S| + |T|}\left[ (H,H')_{S}, \left[[H',H], \left[(H,H')_{T},A \right]\right] \right] \\ &= 2^{|S| + |T|  + 1}\sum_{a_1, \dots , a_{|S| + |T|}, b }   c_{a_1, \dots , a_{|S| + |T|}, b} \cdot \Paren{ \prod_{j \in [\abs{S}+\abs{T}]}  \lambda_{a_j. [j]  } }  \cdot \gamma_b \cdot  A_{a_1, \dots , a_{|S| + |T|}, b}
\end{split}
\end{equation*}
where $  c_{a_1, \dots , a_{|S| + |T|}, b} \in \pm 1, \pm \ii$ and the sum has at most $(|S| + |T| + 1)! \degree^{2(|S| + |T| + 1)}$ terms and each of the terms $ A_{a_1, \dots , a_{|S| + |T|}, b} \in \locals_{4(\ell + 1)\locality}$ and has distance at most $4(\ell + 1)\locality$ from the support of $A$.  Thus, we can rewrite the original commutator in the form $\sum_{G_c \in \locals_{4(\ell+1)\locality}} \zeta_c G_c$ where each of the $\zeta_c$ is a polynomial with real coefficients (because the original commutator is Hermitian) in the variables $\lambda'$ of degree at most $2\ell + 2$.

Since we have the constraint $-1 \leq \lambda_a' \leq 1$ and \eqref{eq:commutator-small} and also we know that $-1 \leq \lambda_a \leq 1$, combining over all of the terms in the above sum, 
\[
\Set{\calC_{\lambda'} } \sststile{2\ell + 2}{\lambda'} \left\{ \zeta_C \leq  (|S| + |T| + 1)! (2\degree)^{2(|S| + |T| + 1)}e^{\calO_{\locality,\degree}(\beta)} m^{1.5}\sqrt{\eps_0} \right\} \,.
\]  
The lower bound can be obtained in a similar manner. 
\end{proof}

\begin{lemma}[Bounding switched commutators of type 1 and 2]\label{lem:commutator-coeff-bound2}
Let $S \in \{0,1 \}^*$ be an arbitrary sequence of length at most $\ell$.  Let $A \in \locals_{\locality}$.  Then we can write   
\[
\ii^{|S|}\left[ (H,H')_{S}, A \right] = \sum \zeta_c G_c
\]
where the terms $G_c \in \locals_{(\ell + 1)\locality}$ and have distance at most $(\ell + 1)\locality$ from the support of $A$.  Furthermore,
\[
 \Set{\calC_{\lambda'} } \sststile{2\ell + 2}{\lambda'} \left\{ |\zeta_C| \leq  (|S| +  1)! (4\degree)^{\ell} \right\} \,.
\]
Note that the $\ii^{|S|}$ ensures that the expression is Hermitian so that the $\zeta_c$ are real.
\end{lemma}
\begin{proof}
The proof is the same as \cref{lem:commutator-coeff-bound}.  We don't need \cref{lem:sos-almost-commuting} at all and just need to use the constraint that $-1 \leq \lambda_a' \leq 1$ and the fact that $-1 \leq \lambda_a \leq 1$ to bound the coefficients.
\end{proof}

We will also need the following lemma that bounds the effect of the sampling error.  This is analogous to \cref{lem:error-in-constraints} but we now need that it is enforced for all potential choices of $\lambda'$.

\begin{lemma}[Estimated expectations are close to true expectations]\label{lem:error-in-sos-constraints}
Under \cref{assumption:accurate-estimates}, for all $B_1, B_2 \in \calB$,
\[
    \Set{\calC_{\lambda'} } \sststile{}{\lambda'} \left\{ \abs{\wt{\tr}(B_2q_{C_{\locality,\degree}\beta, \ell_0}(-\beta H'| B_1) \rho) - \tr(B_2q_{C_{\locality,\degree}\beta, \ell_0}(-\beta H'| B_1) \rho)}^2 \leq   0.01\eps^2 \right\}
\]
\end{lemma}
\begin{proof}
By \cref{lem:cluster-expansion} and \cref{claim:simple-coeff-bound}, we can write $q_{C_{\locality,\degree}\beta, \ell_0}(-\beta H'| B_1)$ as a linear combination of the form $\sum_{A \in \calA} c_A A$
where $\calA \subset \locals_{2^{C_{\locality,\degree}\beta + 1}  \ell_0 \locality}$ and the coefficients $c_A$ are polynomials in the $\lambda'$ of degree at most $q_{C_{\locality,\degree}\beta, \ell_0}$.  Furthermore, the sum of the magnitudes of all of the coefficients of all of the $c_A$ is at most $((1 + \beta) \degree)^{2^{C_{\locality,\degree}\beta + 5}\ell_0}$.  We can then write
\[
\sststile{}{}
\Set{ \wt{\tr}(B_2q_{C_{\locality,\degree}\beta, \ell_0}(-\beta H' \mid  B_1) \rho) - \tr(B_2q_{C_{\locality,\degree}\beta, \ell_0}(-\beta H'\mid B_1) \rho)  = \sum_{A \in \calA} c_A ( \wt{\tr}(B_2A\rho) - \tr(B_2A \rho))   }.
\]
Thus, conditioned on the event that \cref{assumption:accurate-estimates} holds and using the constraints that $-1 \leq \lambda' \leq 1$, we get the desired bound.
\end{proof}

\subsection{No small local marginals in sum-of-squares}\label{sec:main-identifiability-proof}

Now we move onto the main proof.  At a high level, we will show that if $H, H'$ are far, then there must be certain choices for $B_1, B_2$ in the constraint system $\Set{\calC_{\lambda'} }$ that ``witness'' this and thus the constraints will be violated.  We break up the desired statement into a series of inequalities.

First, we can take a suitable linear combination of the constraints to derive that the following quantities must be small.  Note that in the expressions below, $B$ and $[H - H',B] + 0.25 [H- H', B]_3$ are playing the role of our ``witness''.  

\begin{lemma}[Nested commutator polynomials are bounded in SoS]\label{lem:sos-test-functions}
Let $q$ denote the polynomial $q_{C_{\locality,\degree}\beta,\ell_0}$.
Under \cref{assumption:accurate-estimates}, we have for any $B \in \locals_\locality$,
\begin{equation*}
\Set{\calC_{\lambda'} } \sststile{}{\lambda'}   \left\{ \abs*{\tr\Paren{([H - H', B] + 0.25[H - H', B]_3)(q(-\beta H'|B) - q(-\beta H| B))\rho} } \leq (2\degree)^{12} \eps^2  \right\}
\end{equation*}
\end{lemma}

\begin{proof}
Fix any $B \in \locals_\locality$.  By \cref{lem:pauli-commutator}, we have
\begin{equation*}
    \sststile{}{\lambda'} \Set{ \ii\Paren{ [H - H', B] + 0.25[H - H', B]_3 } = \sum_{F_b \in \locals_{4\locality}} \gamma_b F_b },
\end{equation*}
 for some $\gamma_b$ that are degree-$3$ polynomials in the indeterminates $(\lambda'_i)$'s with real coefficients (because the expression is Hermitian) in the $\lambda'$.  Since $-1 \leq \lambda_i \leq 1$ and we also have the constraint that $-1 \leq \lambda_i' \leq 1$, we get
\[
\Set{\calC_{\lambda'} } \sststile{}{\lambda'} \Set{ -(2\degree)^3 \leq \gamma_b \leq (2\degree)^3} \,.
\]
Also, at most $\degree^3$ of the $\gamma_b$ are nonzero.
We recall the following statements: for all $B_1, B_2 \in \calB$,
\begin{align*}
    \Set{\calC_{\lambda'} } &\sststile{}{\lambda'} \Set{
        \Abs{ \wt{\tr}\Paren{B_2q(-\beta H' \mid B_1) \rho} - \tr \Paren{B_2q(-\beta H' \mid  B_1) \rho} }^2 \leq   0.01\eps^2
    } \tag{\cref{lem:error-in-sos-constraints}} \\
    &\sststile{}{\lambda'} \Set{
        \Abs{  \wt{\tr}(B_2q(-\beta H \mid  B_1) \rho) - \tr(B_2q(-\beta H \mid B_1) \rho)}^2   \leq   0.01\eps^2
    }, \tag{\cref{lem:error-in-constraints}} \\
    &\sststile{}{\lambda'} \Set{
        \Abs{   \wt{\tr}(B_1B_2\rho) - \wt{\tr}\Paren{ B_2q(-\beta H \mid B_1) \rho } }^2   \leq   \eps^2
    } \tag{\cref{lem:feasibility}}
\end{align*}
Using the (last) constraint in the system $\calC_{\lambda'}$ and the above, we deduce
\begin{equation*}
    \Set{\calC_{\lambda'} } \sststile{}{\lambda'} \left\{ \abs{ \tr(B_2q_{C_{\locality,\degree}\beta, \ell_0}(-\beta H'| B_1) \rho) - \tr(B_2q_{C_{\locality,\degree}\beta, \ell_0}(-\beta H| B_1) \rho) }^2 \leq   1.44\eps^2 \right\}    \,.
\end{equation*}
Now we can plug in $B_1 = B$ and $B_2 = F_b$ into the above and then take a linear combination with coefficients equal to the $\gamma_b$.  Using the properties of the $\gamma_b$ at the beginning, we conclude  
\[
    \Set{\calC_{\lambda'} } \sststile{}{\lambda'} \left\{ \left\lvert \tr\parens{([H - H', B] + 0.25[H - H', B]_3)(q(-\beta H'|B) - q(\beta H| B))\rho} \right\rvert \leq (2\degree)^6 \eps  \right\}
\]
as desired.
\end{proof}

On the other hand, we will show using the properties of $q_{C_{\locality,\degree}\beta, \ell_0}$ in Theorem~\ref{thm:exp-monotone-approx} that the expression on the LHS in the above is actually lower bounded by some function of $[H-H', B]$.  The first step of this is the following lemma.

\begin{lemma}[Key polynomial identity for nested commutators in sum-of-squares]\label{lem:key-sos-identity}
Let 
\[
\begin{split}
X(B) &= \tr\left(([H - H', B] + 0.25[H - H', B]_3)(q_{C_{\locality,\degree}\beta, \ell_0}(-\beta H|B) - q_{C_{\locality,\degree}\beta, \ell_0}(-\beta H'| B))\rho \right)  \\ &\quad - \frac{1}{2000} \left( \tr\parens{[H - H', B]_2  p_{C_{\locality,\degree}\beta, \ell_0}(-\beta H|B) \rho} \right) \,.
\end{split}
\]
Under \cref{assumption:accurate-estimates}, we have for any $B \in \locals_\locality$,
    \begin{align*}
       \Set{\calC_{\lambda'} } \sststile{}{\lambda'} \{ \textsf{Re}(X(B))    \geq  - \eps \} \\
        \Set{\calC_{\lambda'} } \sststile{}{\lambda'} \{ |\textsf{Im}(X(B))| \leq \eps \}
    \end{align*}
\end{lemma}
\begin{proof}
By \cref{thm:exp-monotone-approx}, we can write the following polynomial equality in two formal variables $x,y$ 
\begin{equation}
    ((x-y) + 0.25(x- y)^3)(q_{C_{\locality,\degree}\beta, \ell_0}(x) - q_{C_{\locality,\degree}\beta, \ell_0}(y)) - \frac{1}{2000} (x-y)^2  p_{C_{\locality,\degree}\beta, \ell_0}(x) = \sum_{j = 1}^{m} r_j(x,y)^2
\end{equation}
where $m \leq 10^{2^{C_{\locality,\degree}\beta}\ell_0}$ and each of the polynomials $r_j$ is $(2^{C_{\locality,\degree}\beta}\ell_0 + 10, 200^{2^{C_{\locality,\degree}\beta}\ell_0} )$-bounded.  Also note that since $H,H',B$ are Hermitian, for any polynomial $q$
\begin{equation}
\sststile{}{\lambda'} \Set{ q(H,H'\mid B)^\dagger = q(-H, -H'\mid B) } \,.
\end{equation}
Thus, we can apply \cref{thm:polynomial-equivalence} to write the following formal identity:
\begin{equation}\label{eq:translate-identity}
\begin{split}
\sststile{}{\lambda'} \Biggl\{ 
&\tr\left(([H - H', B] + 0.25[H - H', B]_3)(q_{C_{\locality,\degree}\beta, \ell_0}(-\beta H|B) - q_{C_{\locality,\degree}\beta, \ell_0}(-\beta H'| B))\rho\right) \\ &\quad -     \frac{1}{2000} \tr\parens{[H - H', B]_2  p_{C_{\locality,\degree}\beta, \ell_0}(-\beta H|B) \rho}  \\ &= \sum_{j = 1}^m \tr( r_j(H,H'|B) r_j(H,H'|B)^\dagger \rho  ) + D \Biggr\}, 
\end{split}
\end{equation}
where $D$ is a sum of four types of terms given on the RHS of Theorem~\ref{thm:polynomial-equivalence}.  By \cref{claim:simple-coeff-bound}, all of the polynomials in the original identity are 
$( 2^{C_{\locality,\degree}\beta + 1}\ell_0 , 200^{2^{C_{\locality,\degree}\beta}\ell_0} )$-bounded and thus \cref{thm:polynomial-equivalence} combined with applying  \cref{lem:commutator-coeff-bound} and \cref{lem:commutator-coeff-bound2} to the individual terms in $D$ gives us that
\[
\begin{split}
\Set{\calC_{\lambda'} } &\sststile{}{\lambda'} \left\{ |\textsf{Re}(D)| \leq  2^{4^{C_{\locality,\degree}\beta}\ell_0} m^{1.5} \sqrt{\eps_0} \right\} \\
\Set{\calC_{\lambda'} } &\sststile{}{\lambda'} \left\{ |\textsf{Im}(D)| \leq 2^{4^{C_{\locality,\degree}\beta}\ell_0} m^{1.5} \sqrt{\eps_0} \right\} \,.
\end{split}
\]
Now by definition, $2^{4^{C_{\locality,\degree}\beta}\ell_0} m^{1.5} \sqrt{\eps_0} \leq 0.1\eps$.  Finally, it remains to note that 
\[
\Set{\calC_{\lambda'} } \sststile{}{\lambda'} \left\{ \tr( r_j(H,H'|B) r_j(H,H'|B)^\dagger \rho  ) \geq 0 \right\} \,.
\]
This is because we can rewrite the LHS above as $\norm{ \rho^{1/2} r_j(H,H'|B)}_F^2$ which is a sum of squares in the real and imaginary parts of the entries of the matrix $\rho^{1/2} r_j(H,H'|B)$.  This matrix has entries that are polynomials with complex coefficients in the $\lambda'$.  Separating out the real and imaginary parts, the real and imaginary parts of the entries are thus polynomials in the $\lambda'$ with real coefficients.  So overall, $\norm{ \rho^{1/2} r_j(H,H'|B)}_F^2$ is a sum of squares of polynomials in the $\lambda'$.  
Combining everything with \eqref{eq:translate-identity} completes the proof.
\end{proof}

Next, we analyze the subtracted term $\tr\parens{[H - H', B]_2  p_{C_{\locality,\degree}\beta, \ell_0}(-\beta H|B) \rho}$ in \cref{lem:key-sos-identity}.  We use the property that $p_{C_{\locality,\degree}\beta, \ell_0}$ is a good approximation to the exponential to relate it to a much simpler expression.

\begin{lemma}[Derivative is uniformly lower bounded ]\label{lem:sos-relate-to-square} Define
\[
Y(B) = \tr\parens{[H - H', B]_2p_{C_{\locality,\degree}\beta, \ell_0}(-\beta H | B)\rho} - \tr\parens{(\ii[H - H', B])^2 \rho} \,.
\]
Then for any $B \in \locals_\locality$
    \begin{align*}
         \Set{\calC_{\lambda'} } \sststile{}{\lambda'} \{ \textsf{Re}(Y(B)) \geq   -  (2\degree)^4 \eps  \} \\
          \Set{\calC_{\lambda'} } \sststile{}{\lambda'} \{ |\textsf{Im}(Y(B))| \leq (2\degree)^4 \eps \} \,.  
    \end{align*}
\end{lemma}
\begin{proof}
Let 
\[
Z(B) = \tr\parens{[H - H', B]_2p_{C_{\locality,\degree}\beta, \ell_0}(-\beta H | B)\rho} - \tr\parens{B[H - H', B]_2 \rho} \,.
\]
We can write $[H - H',B]_2 = \sum_{F_b \in \locals_{3\locality}}\gamma_b F_b$ for some $\gamma_b$ that are degree-$2$ polynomials with real coefficients (because the expression is Hermitian) in the $\lambda'$.  Since $-1 \leq \lambda_i \leq 1$ and we have the constraint that $-1 \leq \lambda_i' \leq 1$, we get 
\begin{equation}\label{eq:commutator-coeff-bound}
\Set{\calC_{\lambda'} } \sststile{}{\lambda'} \{ -(2\degree)^2 \leq \gamma_b \leq (2\degree)^2\} \,.
\end{equation}
Also the number of nonzero $\gamma_b$ is at most $2\degree^2$.  Now by \cref{thm:exp-fancy-approx} the polynomial $p_{C_{\locality,\degree}\beta, \ell_0}$ is a weak exponential approximation with parameters
\[
\left( 0.001 \cdot 2^{C_{\locality,\degree}\beta} \log(1/\eps) , \frac{5}{C_{\locality,\degree}\beta} , 0.001\eps \right) \,.
\]
and thus by \cref{lem:feasibility-of-poly-approx}, we have
\[
|\tr(F p_{C_{\locality,\degree}\beta, \ell_0}(-\beta H | B) \rho - \tr(B F\rho)  | \leq 0.1\eps
\]
for all $F \in \locals_{3\locality}$.  Note that the above is simply a numerical inequality and does not involve any variables of the SoS system.  Now taking a linear combination of the above over all $F \in \locals_{3\locality}$ given by the decomposition $[H - H',B]_2 = \sum_{F_b \in \locals_{3\locality}}\gamma_b F_b$ and using \eqref{eq:commutator-coeff-bound}, we conclude
\begin{equation}\label{eq:sos-bound-z}
\begin{split}
         \Set{\calC_{\lambda'} } \sststile{}{\lambda'} \{ \textsf{Re}(Z(B)) \geq   -  \degree^4 \eps  \} \\
          \Set{\calC_{\lambda'} } \sststile{}{\lambda'} \{ |\textsf{Im}(Z(B))| \leq \degree^4 \eps \} 
\end{split}
\end{equation}
Next, note that 
\[
Y(B) - Z(B) = \tr( B [H - H', B] [ \rho, H'] ) 
\]
and writing $B [H - H', B]$ as a linear combination of elements of $\locals_{2\locality}$ and using that $-1 \leq \lambda_a \leq 1$ and the constraint $-1 \leq \lambda_a' \leq 1$,  we get (from the commutator constraints in $\Set{\calC_{\lambda'} }$) 
\[
\begin{split}
\Set{\calC_{\lambda'} } \sststile{}{\lambda'} \{ |\textsf{Re}(Y(B) - Z(B))| \leq \eps  \} \\
\Set{\calC_{\lambda'} } \sststile{}{\lambda'} \{ |\textsf{Im}(Y(B) - Z(B))| \leq  \eps  \} 
\end{split}
\]
and combining this with \eqref{eq:sos-bound-z} completes the proof.
\end{proof}

\begin{lemma}[No small local marginals for $H-H'$]\label{lem:sos-aaks-v2}
Fix a $B \in \locals_\locality$.  Write $\ii[H - H',B] = \sum_{F_b \in \locals_{2\locality}} \gamma_b F_b$ as a linear combination of $2\locality$-local Pauli matrices $F_b$ where each of the $\gamma_b$ are linear expressions in the $\lambda'$.  Then, for each $\gamma_b$, we have
\[
\Set{ \calC_{\lambda'} }  \sststile{}{\lambda'} \Set{ \gamma_b^2 \leq e^{\calO_{\locality,\degree}(\beta)} \tr((\ii[H - H',B])^2 \rho) } .
\]
\end{lemma}
\begin{proof}
The proof is exactly the same as the proof of \cref{lem:commutator-bound3}.
\end{proof}

\begin{lemma}[Selecting each coefficient in Hamiltonian]\label{lem:commutator-witness}
Let $X = \sum_{E_a \in \locals_\locality} x_a E_a$ be written as a linear combination of $\locality$-local Pauli matrices.  Then for any $E_a \neq I$, there exist $B \in \locals_\locality$ and $P \in \locals_{2\locality}$ whose supports intersect with $E_a$ such that 
\[
\sststile{}{} \Set{ \left( \frac{1}{2} \tr( \ii[X, B] P ) \right)^2 = x_a^2  } \,,
\]
\end{lemma}

\begin{proof}
Note that for any $A_1, A_2 , B \in \locals_\locality$ such that $[A_1, B], [A_2, B] \neq 0$ and $A_1 \neq A_2$, then $[A_1, B] \neq [A_2 , B]$.  In other words, taking the commutator with a fixed Pauli matrix is injective as long as the commutator is nonzero.  Now since $E_a \neq I$, clearly we can choose $B \in \locals_\locality$ with the same support as $E_a$ such that $[B,E_a] \neq 0$.  By \cref{lem:pauli-commutator}, we can write
$\ii[X,B]$ as a linear combination of $2\locality$-local Pauli matrices and the coefficient of $\ii[B, E_a]/2$ (which is a $2\locality$-local Pauli matrix itself) is $\pm 2 x_a$.  Thus, taking $P$ to be $\ii[B, E_a]/2$ gives the desired equality. 
\end{proof}

Now we can complete the proof of \cref{lem:sos-identifiability}.

\begin{proof}[Proof of \cref{lem:sos-identifiability}]
Combining \cref{lem:sos-test-functions}, \cref{lem:key-sos-identity}, \cref{lem:sos-relate-to-square} gives
\[
\Set{ \calC_{\lambda'} }  \sststile{}{\lambda'} \{ \tr\parens{(\ii[H - H', B])^2 \rho} \leq (10\degree)^6 \eps \} \,.
\]
Note that $\ii[H-H',B]$ is Hermitian so the expression on the LHS is a real polynomial in the $\lambda'$.  Now fix an index $a$ and we will analyze $\lambda_a - \lambda_a'$.  By \cref{lem:commutator-witness}, we can find $B \in \locals_\locality$ and $F_b \in \locals_{2\locality}$ such that the coefficient of $F_b$ in the representation $\ii[H - H',B] = \sum_{F_b \in \locals_{2\locality}}\gamma_b F_b$ is $\pm 2(\lambda_a - \lambda_a')$.  Then \cref{lem:sos-aaks-v2} implies that 
\[
\Set{ \calC_{\lambda'} }  \sststile{}{\lambda'} \{ (\lambda_a - \lambda_a')^2 \leq 2^{C_{\locality,\degree}\beta} \eps \} 
\]
as long as $C_{\locality,\degree}$ is sufficiently large in terms of $\locality, \degree$ and this completes the proof.
\end{proof}

\section{A faster algorithm}\label{sec:faster-solver}

We now complete the proof of Theorem~\ref{thm:main-ham-learning-theorem}.  The key to obtaining the faster runtime is observing that actually only a specific family of monomials appear in all of the sum-of-squares proofs in the previous sections.  The key observation is formally stated below.

\begin{definition}[Relevant monomials via cluster expansion]\label{def:limited-monomials}
We say a monomial in the variables $\lambda'$, say $\lambda_{a_1}' \lambda_{a_2} \cdots \lambda_{a_c}'$, is relevant if there exist three clusters 
$(E_{a_1}^{[i]}, \dots , E_{a_C}^{[i]} )$ for $i \in \{1,2, 3 \}$, $E_{a_j}^{[i]} \in \locals_\locality$ such that
\begin{itemize}
    \item $C \leq 10 \cdot 4^{C_{\locality,\degree}\beta}\log(1/\eps)$
    \item $\{E_{a_1}, \dots , E_{a_c} \} \subseteq \bigcup_i \{ E_{a_1}^{[i]}, \dots , E_{a_C}^{[i]}  \}$ as unordered multisets
\end{itemize} 
We use $\calL$ to denote the set of all relevant monomials in the $\lambda'$.
\end{definition}

\begin{lemma}\label{lem:limited-monomials}
Fix a $E_a \in \locals_\locality$.  Then we can write
\[
2^{C_{\locality,\degree}\beta}\eps - (\lambda_a - \lambda_a')^2 = \sum_{i} r_i(\lambda')^2  \prod_{g \in S_i} g(\lambda')
\]
where the $r_i$ are polynomials in the $\lambda'$ and $S_i$ are subsets of constraints in $\Set{ \calC_{\lambda'}}$ (written in the form $g(\lambda') \geq 0$ ).  Furthermore, each product of constraints $\prod_{g \in S_i} g(\lambda')$ involves
\begin{itemize}
    \item At most one commutator constraint
    \item At most one polynomial approximation constraint 
    \item A product of the constraints $\lambda_a + 1 \geq 0, \lambda_a -1 \leq 0$  where the $\lambda_a$ that appear form a relevant monomial
\end{itemize}
\end{lemma}
\begin{proof}
This follows from examining the proofs in Section~\ref{sec:commutator-is-small} and Section~\ref{sec:main-identifiability-proof}, using \cref{lem:cluster-expansion} to characterize all of the monomials that appear whenever we expand a nested commutator in $H$ and $H'$.  Note we can ensure that all of the constraints being multiplied in the sets $S_i$ are distinct because if any constraint is multiplied twice, we get a term of the form $g(\lambda')^2$ and can instead fold it into the $r_i(\lambda')^2$.   
\end{proof}

Now using \cref{lem:count-monomials}, we can count the total number of relevant monomials.

\begin{lemma}\label{lem:count-limited-monomials}
The total number of distinct relevant monomials is at most $m \cdot (1/\eps)^{10^{C_{\locality,\degree}\beta}}$. 
\end{lemma}
\begin{proof}
This follows immediately from \cref{lem:count-monomials}.    
\end{proof}

Finally, we have the following theorem from Steurer and Tiegel~\cite{st21} that allows us to solve sum-of-squares systems with running time depending only on the number of monomials that appear in the proofs.  

\begin{theorem}[Degree reduction via linearization \cite{st21}]\label{thm:solve-limited-sos-system}
Let $p: \R^n \rightarrow \R$ be a multivariate polynomial of degree at most $t$.  Suppose that there exists a system of polynomial inequalities $\calA = \{ q_1 \geq 0, \dots , q_m \geq 0 \}$  such that $\calA \sststile{x}{t} \{ p(x) \geq 0 \}$ and further assume that this proof can be written in the form
\[
\sum_{i} r_i(x)^2 \prod_{j \in S_i} q_j(x)
\]
where $S_i \subseteq [m]$ and there are at most $M$ distinct sets $S_i$.  Also assume that the number of distinct monomials that appear in $p(x)$ is at most $N$.  Then we can write a polynomial system $\calA'$ in $x$ and some additional auxiliary variables such that
\begin{itemize}
\item $\calA'$ is feasible whenever $\calA$ is feasible
\item $\calA' \sststile{}{} \{ p(x) \geq 0 \}$
\end{itemize}
and we can compute a pseudoexpectation satisfying $\calA'$ in time $O(m + M + (tN)^3 )$.
\end{theorem}

\begin{proof}[Proof of \cref{thm:main-ham-learning-theorem}]
The proof is exactly the same as the proof of \cref{thm:weak-ham-learning-theorem} except using \cref{lem:count-limited-monomials} and \cref{lem:limited-monomials} to bound the complexity of the final sum-of-squares proof and using \cref{thm:solve-limited-sos-system} to solve the  sum-of-squares system in running time $\poly(m, \log(1/\delta), (1/\eps)^{2^{f(\locality, \degree)\beta}} )$.   
\end{proof}

\section*{Acknowledgments}
\addcontentsline{toc}{section}{Acknowledgments}

AB is supported by Ankur Moitra's ONR grant and the NSF TRIPODS program (award DMS-2022448). AL is supported in part by an NSF GRFP and a Hertz Fellowship.  AM is supported in part by a Microsoft Trustworthy AI Grant, an ONR grant and a David and Lucile Packard Fellowship.  ET is supported by the NSF GRFP (DGE-1762114) and the Miller Institute for Basic Research in Science, University of California Berkeley.

\newpage
\printbibliography

\appendix
\section{Proof of \texorpdfstring{\cref{aaks-marginals}}{Theorem 2.13}} \label{app:aaks}

We recall the setup from \cref{aaks-marginals}.
We have $H$, a $\locality$-local Hamiltonian with dual interaction graph $\graph$ with max degree $\degree$, along with $A = \sum_b \sigma_b P_b$, a $\locality$'-local operator where $P_b$ are products of Pauli matrices and $-1 \leq \sigma_b \leq 1$ and whose dual interaction graph has max degree $\degree'$.
For $\beta > 0$, $\rho$ is the corresponding Gibbs state of $H$.
We wish to show a lower bound on $\angles{A^2}$.
\begin{align*}
    \angles{A^2} = \tr(A^2 \rho)
    \geq \max_{i \in [\qubits]} \parens[\Big]{c\tr(A_{(i)}^2/\dims)}^{6 + c'\beta}.
\end{align*}

\begin{remark}
    \cite[Theorem 33]{aaks20} proves that
    \begin{equation*}
    \angles{A^2} = \tr(A^2 \rho) \geq \max_{i \in \Lambda} \tr(A_{(i)}^2 / \dims)^{\beta^{\bigOmega{1}}}
    \end{equation*}
    for \emph{quasi-local} operators.
\end{remark}

\subsection{Reproving lemmas}

First, we show that this quantity is ``protected'' by local unitary operations, meaning it doesn't change too much when we apply a local unitary to $A$.
How this works in AAKS is as follows.
There are two claims: Claim 37 and Claim 38.
Claim 37 only uses AKL for local operators, while Claim 38 uses AKL for quasi-local operators (and relies on Claim 37).

\begin{lemma}[{\cite[Claim 37]{aaks20}}] \label{lem:aaks-37}
    Let $U$ be a unitary supported on $S \subset [\qubits]$.
    Then, for any $M$ with $\norm{M} \leq 1$,
    \begin{align*}
        \angles{(U^\dagger M U)^2} \leq \bigOt[\Big]{(\degree + 1)\locality\abs{S}\parens[\big]{4e^{\frac{2\abs{S}}{\locality}}}^{\frac{2(\degree + 1)\locality\beta}{1 + 2(\degree+1)\locality\beta}} \angles{M^2}^{\frac{1}{1 + 2(\degree+1)\locality \beta}}}.
    \end{align*}
\end{lemma}
Thinking of $\degree$ and $\locality$ as being constants and choosing $c, c'$ appropriately, the inequality becomes
\begin{align*}
        \angles{(U^\dagger M U)^2} \leq e^{c\abs{S}} \angles{M^2}^{\frac{1}{1 + c'\beta}}.
\end{align*}
The $\bigOt{\cdots}$ gets folded into the values of $c, c'$.
\begin{proof}
The overview of the proof is as follows.
We want to bound the influence of $U$ on the expression $\tr(UM^2U\rho)$, so we can consider it in terms of the eigenbasis of $H$.
The off-diagonal pieces are bounded by \cref{lem:akl} (applying it with $\onespin = \degree + 1$ and $R = (\degree + 1)\abs{S}$):
\begin{align} \label{eq:akl-in-aaks-h}
    \norm{\Pi_{[\sigma + \Delta, \infty]}^{(H)}U\Pi_{[-\infty, \sigma]}^{(H)}}
    \leq \norm{U}e^{-\frac{1}{4(\degree + 1)\locality}(\Delta - 4(\degree + 1)\abs{S})}
    = e^{-\frac{\Delta}{4(\degree+1)\locality}}e^{\frac{\abs{S}}{\locality}}.
\end{align}
We split our expression into the on-diagonal and off-diagonal pieces.
\begin{align*}
    \angles{(U^\dagger M U)^2}
    &= \tr((U^\dagger M U)^2 \rho) \\
    &= \sum_i \tr((U^\dagger M U)^2 \Pi_{[i, i+1)}^{(H)}\rho) \\
    &= \sum_i \fnorm{U^\dagger M U \Pi_{[i, i+1)}^{(H)}\sqrt{\rho}}^2 \\
    &= \sum_i \fnorm{M U \Pi_{[i, i+1)}^{(H)}\sqrt{\rho}}^2 \\
    &= \sum_i \fnorm{M (\Pi_{(-\infty, i - \Delta)}^{(H)} + \Pi_{[i-\Delta,i+1+\Delta)}^{(H)} + \Pi_{[i+1+\Delta,\infty)}^{(H)}) U \Pi_{[i, i+1)}^{(H)}\sqrt{\rho}}^2 \\
    &\leq 2\sum_i \parens[\Big]{\underbrace{\fnorm{M (\Pi_{(-\infty, i - \Delta)}^{(H)} + \Pi_{[i+1+\Delta,\infty)}^{(H)}) U \Pi_{[i, i+1)}^{(H)}\sqrt{\rho}}^2}_{(\textsc{off}_i)} + \underbrace{\fnorm{M \Pi_{[i-\Delta,i+1+\Delta)}^{(H)} U \Pi_{[i, i+1)}^{(H)}\sqrt{\rho}}^2}_{(\textsc{on}_i)}}
\end{align*}
We first bound $(\textsc{on}_i)$.
\begin{align*}
    (\textsc{on}_i)
    &= \fnorm{M \Pi_{[i-\Delta,i+1+\Delta)}^{(H)} U \Pi_{[i, i+1)}^{(H)}\sqrt{\rho}}^2 \\
    &= \fnorm{M \Pi_{[i-\Delta,i+1+\Delta)}^{(H)}}^2 \norm{U \Pi_{[i, i+1)}^{(H)}\sqrt{\rho}}^2 \\
    &\leq \fnorm{M \Pi_{[i-\Delta,i+1+\Delta)}^{(H)}}^2 e^{-\beta i} \\
    &\leq \fnorm{M \Pi_{[i-\Delta,i+1+\Delta)}^{(H)}\sqrt{\rho}}^2 e^{\beta \Delta}
\end{align*}
Next, we bound $(\textsc{off}_i)$.
Here, we use that $\norm{M} \leq 1$ and \cref{eq:akl-in-aaks-h}.
\begin{align*}
    (\textsc{off}_i)
    &= \fnorm{M (\Pi_{(-\infty, i - \Delta)}^{(H)} + \Pi_{[i+1+\Delta,\infty)}^{(H)}) U \Pi_{[i, i+1)}^{(H)}\sqrt{\rho}}^2 \\
    &\leq \norm{M}^2 \norm{(\Pi_{(-\infty, i - \Delta)}^{(H)} + \Pi_{[i+1+\Delta,\infty)}^{(H)}) U \Pi_{[i, i+1)}^{(H)}}^2 \fnorm{\Pi_{[i, i+1)}^{(H)}\sqrt{\rho}}^2 \\
    &\leq \norm{(\Pi_{(-\infty, i - \Delta)}^{(H)} + \Pi_{[i+1+\Delta,\infty)}^{(H)}) U \Pi_{[i, i+1)}^{(H)}}^2 \fnorm{\Pi_{[i, i+1)}^{(H)}\sqrt{\rho}}^2 \\
    &\leq 2(e^{-\frac{\Delta}{4(\degree+1)\locality}}e^{\frac{\abs{S}}{\locality}})^2 \fnorm{\Pi_{[i, i+1)}^{(H)}\sqrt{\rho}}^2 \\
    &= 2e^{\frac{2\abs{S}}{\locality}}e^{-\frac{\Delta}{2(\degree+1)\locality}}\fnorm{\Pi_{[i, i+1)}^{(H)}\sqrt{\rho}}^2 \\
    &= 2e^{\frac{2\abs{S}}{\locality}}e^{-\frac{\Delta}{2(\degree+1)\locality}}\fnorm{\Pi_{[i, i+1)}^{(H)}\sqrt{\rho}}^2
\end{align*}
The remaining terms can be summed up nicely.
\begin{align*}
    \angles{(U^\dagger M U)^2}
    &\leq 2\sum_{i}((\textsc{off}_i) + (\textsc{on}_i)) \\
    &\leq 2\sum_{i}(2e^{\frac{2\abs{S}}{\locality}}e^{-\frac{\Delta}{2(\degree+1)\locality}} \fnorm{\Pi_{[i, i+1)}^{(H)}\sqrt{\rho}}^2 + \fnorm{M \Pi_{[i-\Delta,i+1+\Delta)}^{(H)}\sqrt{\rho}}^2 e^{\beta \Delta}) \\
    &= 4e^{\frac{2\abs{S}}{\locality}}e^{-\frac{\Delta}{2(\degree+1)\locality}} \fnorm{\sqrt{\rho}}^2 + (2\Delta + 1)\fnorm{M\sqrt{\rho}}^2 e^{\beta \Delta} \\
    &= 4e^{\frac{2\abs{S}}{\locality}}e^{-\frac{\Delta}{2(\degree+1)\locality}} + (2\Delta + 1)e^{\beta \Delta} \angles{M^2} \\
    &= \angles{M^2}e^{\beta\Delta}\parens[\Big]{\frac{4e^{\frac{2\abs{S}}{\locality}}}{\angles{M^2}}e^{-\Delta(\frac{1}{2(\degree+1)\locality} + \beta)} + 2\Delta + 1}
\end{align*}
This holds for every $\Delta \geq 0$.
We choose 
\begin{gather*}
    \Delta
    = \frac{1}{\frac{1}{2(\degree+1)\locality} + \beta}\log \frac{4e^{\frac{2\abs{S}}{\locality}}}{\angles{M^2}}.
\end{gather*}
Since $\angles{M^2} \leq 1$, $\Delta$ is indeed non-negative.
With this choice of $\Delta$, we have
\begin{align*}
    \angles{(U^\dagger M U)^2}
    &\leq \angles{M^2}e^{\beta\Delta}\parens[\Big]{\frac{4e^{\frac{2\abs{S}}{\locality}}}{\angles{M^2}}e^{-\Delta(\frac{1}{2(\degree+1)\locality} + \beta)} + (2\Delta + 1)} \\
    &= \angles{M^2}e^{\beta\Delta}(2\Delta + 2) \\
    &= \angles{M^2}(2\Delta + 2) \parens[\Big]{\frac{4e^{\frac{2\abs{S}}{\locality}}}{\angles{M^2}}}^{\frac{\beta}{\frac{1}{2(\degree+1)\locality} + \beta}} \\
    &= (2\Delta + 2) \parens[\big]{4e^{\frac{2\abs{S}}{\locality}}}^{\frac{2(\degree + 1)\locality\beta}{1 + 2(\degree+1)\locality\beta}} \angles{M^2}^{\frac{1}{1 + 2(\degree+1)\locality \beta}} \\
    &= \bigOt[\Big]{(\degree + 1)\locality\abs{S}\parens[\big]{4e^{\frac{2\abs{S}}{\locality}}}^{\frac{2(\degree + 1)\locality\beta}{1 + 2(\degree+1)\locality\beta}} \angles{M^2}^{\frac{1}{1 + 2(\degree+1)\locality \beta}}}.
\end{align*}
In the last line above, we are careful to pull out factors of $(\degree + 1)\locality$ and $\abs{S}$ to deal with regimes where $\abs{S}/k$ and $\beta$ are $\ll 1$ and $\gg 1$.
\end{proof}

\begin{lemma}[{\cite[Claim 38 + Corollary 39]{aaks20}}] \label{lem:aaks-38}
    Let $U$ be a unitary supported on $S \subset [\qubits]$ and let $A$ be the local operator as defined previously.
    Then, for all $\gamma > 0$,
    \begin{align*}
        \angles{(U^\dagger A U)^2} \leq \gamma^2 + e^{c\abs{S}}\parens[\Big]{\angles{\Pi_{(-\infty, -\gamma) \cup (\gamma, \infty)}^{(A)}}^{\frac{1}{2(1+c'\beta)}} + \frac{\gamma + 1}{\gamma^3} \angles{A^2}}
    \end{align*}
    where $c, c'$ are constants that depend on $\degree, \locality, \degree', \locality'$.
\end{lemma}
Our eventual goal is to get an upper bound in terms of $\angles{A^2}^{\Theta(\frac{1}{1+\beta})}$; since
\begin{equation} \label{eq:aaks-projector-mass}
    \angles{\Pi_{(-\infty, -\gamma) \cup (\gamma, \infty)}^{(A)}} \leq \gamma^{-2} \angles{A^2},
\end{equation}
by setting $\gamma$ appropriately we can do this.
\begin{remark}
    The original result proves this result for \emph{quasi-local} operators, instead of the local operator case under consideration for us.
    They attain
    \begin{align}
        \angles{(U^\dagger A U)^2} \leq \gamma^2 + \frac{1}{\gamma} e^{\bigO{\abs{S}}}\angles{\Pi_{(-\infty, -\gamma) \cup (\gamma, \infty)}^{(A)}}^{\bigO{1/\beta}} + \bigO{\abs{S}^6\frac{1}{\gamma^4}\angles{A^2}}
    \end{align}
    The final $\frac{1}{\gamma^4}$ is not present in the statement, but is incurred in \cite[Eq.\ 121]{aaks20}.
\end{remark}
\begin{proof}
Throughout, $c$'s denote positive constants that depend on $k, k', \degree, \degree'$.
We use \cref{lem:akl} to conclude the exact same bound as used in \cref{eq:akl-in-aaks-h}, but this time for $A$.
\begin{align} \label{eq:akl-in-aaks-a}
    \norm{\Pi_{[\sigma + \Delta, \infty]}^{(A)}U\Pi_{[-\infty, \sigma]}^{(A)}}
    \leq \norm{U}e^{-\frac{1}{4(\degree' + 1)\locality'}(\Delta - 4(\degree' + 1)\abs{S})}
    = e^{-\frac{\Delta}{4(\degree'+1)\locality'}}e^{\frac{\abs{S}}{\locality'}}.
\end{align}
We split up our expression in pieces based on the eigenbasis of $A$.
Let $\gamma$ be a parameter that we choose later, and let $\Pi_i^{(A)} = \Pi_{[(-i-1)\gamma, -i\gamma) \cup [i\gamma, (i+1)\gamma)}^{(A)}$.
\begin{align*}
    \angles{(U^\dagger A U)^2}
    &= \tr(U^\dagger A^2 U \rho) \\
    &= \sum_{i\geq 0} \tr(U^\dagger A \Pi_i^{(A)} A U \rho) \\
    &\leq \sum_{i\geq 0} (i+1)^2\gamma^2 \tr(U^\dagger \Pi_i^{(A)} U \rho) \\
    &= \sum_{i\geq 0} (i+1)^2\gamma^2 \fnorm{\Pi_i^{(A)} U \sqrt{\rho}}^2 \\
    &\leq \gamma^2 + \gamma^2\sum_{i\geq 1} (i+1)^2 \fnorm{\Pi_i^{(A)} U \sqrt{\rho}}^2 \\
    &\leq \gamma^2 + \gamma^2\sum_{i\geq 1} (i+1)^2 \parens[\Big]{\underbrace{\sum_{j \geq 0}\fnorm{\Pi_i^{(A)} U \Pi_j^{(A)} \sqrt{\rho}}}_{(\textsc{term}_i)}}^2
\end{align*}
By \cref{eq:akl-in-aaks-a}, we can conclude
\begin{align}
    \fnorm{\Pi_i^{(A)} U \Pi_j^{(A)} \sqrt{\rho}}
    &\leq \norm{\Pi_i^{(A)} U \Pi_j^{(A)}} \fnorm{\Pi_j^{(A)}\sqrt{\rho}}
    \leq 16e^{-\frac{\gamma\abs{j-i}}{4(\degree'+1)\locality'}}e^{\frac{\abs{S}}{\locality'}} \fnorm{\Pi_j^{(A)}\sqrt{\rho}} \\
    &\leq e^{-c_0\gamma\abs{j-i}}e^{c_1\abs{S}} \fnorm{\Pi_j^{(A)}\sqrt{\rho}} \label{eq:off-bound-nonzero}
\end{align}
We don't want to incur dependence on $\fnorm{\Pi_0^{(A)} \sqrt{\rho}}^2 = \angles{\Pi_0^{(A)}}$.
When $j = 0$, we can use \cref{lem:aaks-37} to get a bound that still depends on $\angles{I - \Pi_0^{(A)}}$.
\begin{align}
    \fnorm{\Pi_i^{(A)} U \Pi_0^{(A)} \sqrt{\rho}}
    &\leq \fnorm{\Pi_i^{(A)} U \sqrt{\rho}} + \fnorm{\Pi_i^{(A)} U (\id - \Pi_0^{(A)}) \sqrt{\rho}} \nonumber\\
    &\leq \sqrt{e^{c_2\abs{S}}\angles{\Pi_i^{(A)}}^{\frac{1}{1+c_3\beta}}} + \fnorm{(\id - \Pi_0^{(A)}) \sqrt{\rho}} \nonumber\\
    &\leq e^{c_2\abs{S}/2}\angles{I - \Pi_0^{(A)}}^{\frac{1}{2(1+c_3\beta)}} \label{eq:off-bound-zero}
\end{align}
With both of these bounds in hand, we can now bound $(\textsc{term}_i)$.
We use \cref{eq:off-bound-nonzero,eq:off-bound-zero} to get 
\begin{align*}
    (\textsc{term}_i)
    &= \sum_{j \geq 0}\fnorm{\Pi_i^{(A)} U \Pi_j^{(A)} \sqrt{\rho}} \\
    &\leq \fnorm{\Pi_i^{(A)} U \Pi_0^{(A)} \sqrt{\rho}} + \sum_{j \geq 1} e^{-c_0\gamma\abs{j-i}}e^{c_1\abs{S}} \fnorm{\Pi_j^{(A)}\sqrt{\rho}} \tag*{by \cref{eq:off-bound-nonzero}}\\
    &\leq \sqrt{e^{c_2\abs{S}/2}\angles{I - \Pi_0^{(A)}}^{\frac{1}{2(1+c_3\beta)}}e^{-c_0\gamma\abs{0-i}}e^{c_1\abs{S}} \fnorm{\Pi_0^{(A)}\sqrt{\rho}}} + \sum_{j \geq 1} e^{-c_0\gamma\abs{j-i}}e^{c_1\abs{S}} \fnorm{\Pi_j^{(A)}\sqrt{\rho}} \tag*{by \cref{eq:off-bound-zero,eq:off-bound-nonzero}}\\
    &\leq e^{(\frac{c_1}{2} + \frac{c_2}{4})\abs{S}}\angles{I - \Pi_0^{(A)}}^{\frac{1}{4(1+c_3\beta)}}e^{-\frac{c_0}{2}\gamma i} + \sum_{j \geq 1} e^{-c_0\gamma\abs{j-i}}e^{c_1\abs{S}} \fnorm{\Pi_j^{(A)}\sqrt{\rho}} \tag*{by $\fnorm{\Pi_0^{(A)}\sqrt{\rho}} \leq 1$}\\
    &\leq e^{c_4\abs{S}}\parens[\Bigg]{\angles{I - \Pi_0^{(A)}}^{\frac{1}{4(1+c_3\beta)}}e^{-\frac{c_0}{2}\gamma i} + \parens[\Big]{\sum_{j' \geq 1} e^{-c_0\gamma \abs{j'-i}}}^{\frac12}\parens[\Big]{\sum_{j \geq 1} e^{-c_0\gamma \abs{j-i}} \fnorm{\Pi_j^{(A)}\sqrt{\rho}}^2}^{\frac12}} \tag*{by Cauchy-Schwarz and taking $c_4 = \max(\frac{c_1}{2} + \frac{c_2}{4}, c_1)$} \\
    &\leq e^{c_4\abs{S}}\parens[\Big]{\angles{I - \Pi_0^{(A)}}^{\frac{1}{4(1+c_3\beta)}}e^{-\frac{c_0}{2}\gamma i} + c_5\sqrt{1 + \frac{1}{\gamma}}\parens[\Big]{\sum_{j \geq 1} e^{-c_0 \gamma \abs{j-i}} \fnorm{\Pi_j^{(A)}\sqrt{\rho}}^2}^{\frac12}} \tag*{taking $c_5$ appropriately}\\
    (\textsc{term}_i)^2
    &\leq 2e^{2c_4\abs{S}}\parens[\Big]{\angles{I - \Pi_0^{(A)}}^{\frac{1}{2(1+c_3\beta)}}e^{-c_0\gamma i} + c_5^2\parens[\Big]{1 + \frac{1}{\gamma}}\parens[\Big]{\sum_{j \geq 1} e^{-c_0 \gamma \abs{j-i}} \fnorm{\Pi_j^{(A)}\sqrt{\rho}}^2}} \\
    &\leq e^{c_6\abs{S}}\parens[\Big]{\angles{I - \Pi_0^{(A)}}^{\frac{1}{2(1+c_3\beta)}}e^{-c_0\gamma i} + \parens[\Big]{1 + \frac{1}{\gamma}}\parens[\Big]{\sum_{j \geq 1} e^{-c_0 \gamma \abs{j-i}} \fnorm{\Pi_j^{(A)}\sqrt{\rho}}^2}} \tag*{taking $c_6$ appropriately}
\end{align*}
Returning to what we originally wanted to bound,
\begin{align*}
    &\angles{(U^\dagger A U)^2} \\
    &\leq \gamma^2 + \gamma^2\sum_{i\geq 1} (i+1)^2 (\textsc{term}_i)^2 \\
    &\leq \gamma^2 + \gamma^2e^{c_6\abs{S}}\sum_{i\geq 1} (i+1)^2 \parens[\Big]{\angles{I - \Pi_0^{(A)}}^{\frac{1}{2(1+c_3\beta)}}e^{-c_0\gamma i} + \parens[\Big]{1 + \frac{1}{\gamma}}\parens[\Big]{\sum_{j \geq 1} e^{-c_0 \gamma \abs{j-i}} \fnorm{\Pi_j^{(A)}\sqrt{\rho}}^2}} \\
    &= \gamma^2 + e^{c_6\abs{S}}\angles{I - \Pi_0^{(A)}}^{\frac{1}{2(1+c_3\beta)}}\sum_{i\geq 1} \gamma^2(i+1)^2 e^{-c_0\gamma i} + e^{c_6\abs{S}}\frac{\gamma + 1}{\gamma} \sum_{j \geq 1} j^2\fnorm{\Pi_j^{(A)}\sqrt{\rho}}^2 \sum_{i \geq 1} \gamma^2(\tfrac{i+1}{j})^2 e^{-c_0\gamma \abs{j-i}} \\
    &\lesssim \gamma^2 + c_7e^{c_6\abs{S}}\angles{I - \Pi_0^{(A)}}^{\frac{1}{2(1+c_3\beta)}} + c_7e^{c_6\abs{S}}\frac{\gamma + 1}{\gamma} \sum_{j \geq 1} j^2\fnorm{\Pi_j^{(A)}\sqrt{\rho}}^2 \tag*{taking $c_7$ such that $c_7 \geq \sum_{i \geq 1} \gamma^2(i+1)^2 e^{-c_0 \gamma\abs{i - j}}$ for all $j$}\\
    &= \gamma^2 + e^{c_8\abs{S}}\parens[\Big]{\angles{I - \Pi_0^{(A)}}^{\frac{1}{2(1+c_3\beta)}} + \frac{\gamma + 1}{\gamma^3} \angles{A^2}} \tag*{by $\sum_{j \geq 1} (\gamma j)^2 \fnorm{\Pi_j^{(A)}\sqrt{\rho}}^2 \leq \angles{A^2}$}
\end{align*}
\end{proof}

\subsection{Applying the lemma}

\begin{lemma}[{\cite[Claim 36]{aaks20}}] \label{lem:aaks-36}
    There exists a unitary $U$ supported on $\supp(A_{(i)})$ such that
    \begin{align*}
        \tr(A_{(i)}^2 /\dims) \leq \tr((U^\dagger A_{(i)} U)^2 \rho) = \angles{(U^\dagger A_{(i)} U)^2}
    \end{align*}
\end{lemma}
\begin{proof}
Take $U$ to be the unitary that sends the $a$th largest eigenvector of $A_{(i)}^2$ to the $a$th largest eigenvector of $\tr_{[\qubits] \setminus \supp(A_{(i)})}\rho$.
\ewin{Write this better.}
\end{proof}

\begin{proof}[Proof of \cref{aaks-marginals}]
We consider $A_{(i)}$, and let $U_*$ be the unitary whose existence is shown by \cref{lem:aaks-36}.
Since $A_{(i)}$ is supported on $(\degree' + 1)\locality'$ qubits (\cref{eq:localizing-local}), so is $U_*$.
\begin{align*}
    \tr(A_{(i)}^2/\dims)
    &\leq \angles{(U_*^\dagger A_{(i)} U_*)^2} \tag*{by \cref{lem:aaks-36}} \\
    &= \fnorm{U_*^\dagger A_{(i)} U_* \sqrt{\rho}}^2 \\
    &= \fnorm[\Big]{U_*^\dagger \parens[\Big]{A - \int \diff\mu_i(U) U^\dagger A U} U_* \sqrt{\rho}}^2 \\
    &\leq 2\fnorm{U_*^\dagger A U_* \sqrt{\rho}}^2 + 2\parens[\Big]{\int \diff\mu_i(U) \fnorm{U_*^\dagger U^\dagger A_{(i)} U U_* \sqrt{\rho}}}^2 \\
    &\leq 4\gamma^2 + 4e^{c(\degree' + 1)\locality'}\parens[\Big]{\angles{\Pi_{(-\infty, -\gamma) \cup (\gamma, \infty)}^{(A)}}^{\frac{1}{2(1+c'\beta)}} + \frac{\gamma + 1}{\gamma^3} \angles{A^2}} \tag*{by \cref{lem:aaks-38}} \\
    &\leq 4\gamma^2 + 4e^{c(\degree' + 1)\locality'}\parens[\Big]{\gamma^{-\frac{1}{1+c'\beta}}\angles{A^2}^{\frac{1}{2(1+c'\beta)}} + \frac{\gamma + 1}{\gamma^3} \angles{A^2}} \tag*{by \cref{eq:aaks-projector-mass}}
\end{align*}
We set $\gamma^2 = \tr(A_{(i)}^2/\dims) / 8$ and $c'' = 8e^{c(\degree' + 1)\locality'}$ so that this implies
\begin{align*}
    \tr(A_{(i)}^2/\dims) = 8\gamma^2
    &\leq 4\gamma^2 + c''\parens[\Big]{\gamma^{-\frac{1}{1+c'\beta}}\angles{A^2}^{\frac{1}{2(1+c'\beta)}} + \frac{\gamma + 1}{\gamma^3} \angles{A^2}} \\
    \frac{4}{c''}\gamma^2 &\leq \gamma^{-\frac{1}{1+c'\beta}}\angles{A^2}^{\frac{1}{2(1+c'\beta)}} + \frac{\gamma + 1}{\gamma^3} \angles{A^2}
\end{align*}
So, either
\begin{align}
    \tfrac{2}{c''}\gamma^2 &\leq \gamma^{-\frac{1}{1+c'\beta}}\angles{A^2}^{\frac{1}{2(1+c'\beta)}} \nonumber\\
    (\tfrac{2}{c''})^{2(1+c'\beta)}\gamma^{2 + 4(1+c'\beta)} &\leq \angles{A^2} \label{eq:aaks-final}
\end{align}
or
\begin{align*}
    \frac{2}{c''}\gamma^2 &\leq \frac{\gamma + 1}{\gamma^3} \angles{A^2} \\
    c''' \gamma^5 &\leq \angles{A^2}
\end{align*}
Above, we used that $\gamma = \tr(A_{(i)}^2/\dims)/8 \leq (\degree + 1)/8$, so it can be folded into the constant.
\cref{eq:aaks-final} gives us our final bound.
\end{proof}

\section{Proof of \texorpdfstring{\cref{thm:exp-monotone-approx}}{Theorem 4.6}}\label{appendix:poly-approx}

Here we prove \cref{thm:exp-monotone-approx}.  As the proof is quite long and computational, we break it into several manageable steps. First, we show that the even and odd truncations $s_{2\ell - 1}(x), s_{2\ell}(x)$ can be related as follows: 

\begin{lemma}[Even truncations are bounded]\label{claim:compare-to-derivative}
For all $\ell \in \mathbb{N}$,  for all  $x \in \mathbb{R}$, 
\begin{equation*}
     \abs{ s_{2\ell-1}(x) } <  99 s_{2\ell}(x) . 
\end{equation*}
\end{lemma}
\begin{proof}
First, note that whenever $s_{2\ell - 1}(x) \geq 0$, we clearly have $s_{2\ell}(x) \geq |s_{2\ell - 1}(x)|$ and the desired inequalities clearly hold.  Thus, it suffices to consider when $s_{2\ell - 1}(x) < 0$ which only happens when $x \leq -0.1\ell$.  It remains to show that when $x \leq -0.1\ell$, $s_{2\ell}(x) + \frac{s_{2\ell -1}(x)}{99} \geq 0$. 
Let $f(x) = 0.99s_{2\ell}(x) + 0.01 s_{2\ell -1}(x) - e^x$.  Note that $f(0) = 0$ and $f^{(k)}(0) = 0$ for $ k \leq 2\ell - 1$ and $f^{(2\ell)}(x) = 0.99 - e^x$ where $f^{(k)}(x)$ denotes the $k$th derivative of $f$.  Now we can use the fundamental theorem of calculus to write
\[
f(x) = \int_{0}^x f^{(2\ell)}(y) \frac{(x - y)^{2\ell - 1}}{(2\ell - 1)!} \diff y = \int_0^x (0.99 - e^y)\frac{(x - y)^{2\ell - 1}}{(2\ell - 1)!} \diff y \,.
\]
We can rearrange the above as
\[
f(-x) = \int_0^x (0.99- e^{-y})\frac{(x - y)^{2\ell - 1}}{(2\ell - 1)!} \diff y \,.
\]
It suffices to prove that for $x \geq 0.1\ell$ that $f(-x) \geq 0$ since then we get that $0.99s_{2\ell}(-x) + 0.01 s_{2\ell -1}(-x) \geq e^{-x} \geq 0$.  Now, we analyze the integral on the RHS of the above.  We have
\[
\begin{split}
&\int_0^x (0.99- e^{-y})(x - y)^{2\ell - 1} \diff y \\ &\geq \int_{0.05}^{0.1} (x - y)^{2\ell - 1} (0.99 - e^{-y}) \diff y - \int_{0}^{0.02} (e^{-y} - 0.99) (x - y)^{2\ell - 1}\diff y  \\ &\geq 0.0019 (x - 0.1)^{2\ell - 1} - 0.0002 x^{2\ell - 1}
\end{split} 
\]
but since $x \geq 0.1 \ell$, 
\[
\frac{(x - 0.1)^{2\ell - 1}}{x^{2\ell - 1}} = \left( 1 - \frac{0.1}{x}\right)^{2\ell - 1} > \frac{1}{9}
\]
and thus we conclude $f(-x) \geq 0$ and we are done.
\end{proof}

Next, we present a few basic facts about polynomials that allow us to construct bounded coefficient sum of squares representations.

\begin{claim}\label{claim:bounded-roots}
Let $p(x) = a_nx^n + a_{n-1}x^{n-1} + \dots + a_0$ be a polynomial.  Let $C > 0$ be  a constant such that for all $i \in [n]$,
\[
|a_i| \geq \frac{|a_{i-1}|}{C} \,. 
\]
Then all of the (complex) roots of $p$ have magnitude at most $2C$.
\end{claim}
\begin{proof}
Let $z \in \C$ have $|z| \geq 2C$.  Then for all $i$,
\[
|a_{i}z^{i}| \leq \frac{C^{n-i}}{|z|^{n-i}} \cdot |a_nz^n| \leq \frac{|a_nz^n|}{2^{n- i}}
\]
and thus,
\[
|a_n z^n| > |a_{n-1}z^{n-1}| + \dots + |a_0|
\]
so we cannot have
\[
p(z) = a_nz^n + a_{n-1}z^{n-1} + \dots + a_0 = 0 \,.
\]
Thus, it is impossible for $z$ to be a root of $p$.
\end{proof}

As a straight-forward corollary, we show that each of the following admit a bounded coefficient sum-of-squares decomposition.

\begin{corollary}[Bounded coefficient Polynomials]\label{coro:bouned-sos-coeffs1}
Let $\ell$ be a positive integer and $-0.01 \leq c \leq 0.01$ be a real number.  Then, the polynomial
\[
s_{2\ell}(x) + c s_{2\ell -1}(x)
\]
is a $(2^\ell, \ell, 10^\ell)$-bounded SoS polynomial in $x$ (see Definition~\ref{def:sos-bounded-polynomial}).    
\end{corollary}
\begin{proof}
Note that all complex roots of $s_{2\ell}(x) + c s_{2\ell - 1}(x)$ have magnitude at most $5\ell$ by Claim~\ref{claim:bounded-roots}.  Also, by Claim~\ref{claim:compare-to-derivative}, none of the roots are real so by the Fundamental Theorem of Algebra, they come in conjugate pairs $z_1, \overline{z_1}, \dots , z_{\ell}, \overline{z_{\ell}}$.  Then, we can write 
\begin{equation}
s_{2\ell}(x) + c \ s_{2\ell - 1}(x){99} = \frac{1}{(2\ell)!} \prod_{j \in [\ell ]} (x - z_j)(x - \overline{z_j})  = \frac{1}{(2\ell)!} \prod_{j \in [\ell]} ((x - \Real(z_j))^2 + \Imag(z_j)^2)
\end{equation}
and now we can expand out the product above to get a sum of $2^\ell$ squares of polynomials of degree at most $\ell$.  Now we bound the coefficients of each of them.  For each monomial $x^k$, its coefficient has magnitude at most 
\begin{equation}
\binom{\ell}{k}\frac{(5\ell)^{\ell - k}}{\sqrt{(2l)!}} \leq  \frac{\ell^k}{k!} \frac{(10\ell)^{\ell - k}}{ \ell^{\ell}} = \frac{10^\ell}{k!} 
\end{equation}
and thus, we conclude that $s_{2\ell}(x) + c s_{2\ell - 1}(x)$ is a $(2^\ell, \ell, 10^\ell)$-bounded sum-of-squares polynomial.
\end{proof}

\begin{claim}\label{claim:sos-integration}
Let $p(x, y , t)$ be a polynomial such that for all $t \in [0,1]$, it is a $(k,d,C)$-bounded SoS polynomial in $x,y$ (after plugging in a real value for $t$).  Then the polynomial 
\[
r(x,y) = \int_{0}^1 p(x,y, t) \diff t
\]
is a $(3d^2, d, \sqrt{k}C)$-bounded SoS polynomial in $x,y$.
\end{claim}
\begin{proof}
We for each $t \in [0,1]$, there are some $(d,C)$-bounded polynomials $q_{1,t}(x,y), \dots , q_{k,t}(x,y)$ so that we can write
\[
p(x,y,t) = q_{1,t}(x,y)^2 + \dots + q_{k,t}(x,y)^2 \,.
\]
Let $v(x,y)$ be the vector of monomials 
\[v(x,y) = \left(1,x,y, \frac{x^2}{2!}, \frac{xy}{2!} , \frac{y^2}{2!}, \dots , \frac{x^d}{d!}, \frac{x^{d-1}y}{d!}, \dots , \frac{y^d}{d!} \right) \,. \]
Then we can associate each $q_{i,t}$ with a vector $u_{i,t}$ such that $q_{i,t}(x,y) = v(x,y)^\top  u_{i,t}$.  Since the $q_{i,t}$ are $(d,C)$-bounded, we know that all entries of $u_{i,t}$ are at most $C$.  Define the matrix
\[
M(t) = \sum_{i = 1}^k u_{i,t} u_{i,t}^\top \,.
\]
Then $M(t)$ is PSD and 
\[
p(x,y,t) = v(x,y)^\top M(t) v(x,y) \,.
\]
Now we can write
\[
r(x,y) = \int_{0}^1 p(x,y, t) \diff t = v(x,y)^\top \left( \int_{0}^1 M(t) \diff t \right) v(x,y) \,.
\]
We know that all of the entries of $R(t) = \int_{0}^1 M(t) \diff t $ are at most $kC^2$ and also it is a $\binom{d+2}{2} \times \binom{d+2}{2}$ matrix so we can write it as 
\[
R(t)= \sum_{i = 1}^{\binom{d+2}{2}} u_i u_i^\top
\]
for some vectors $u_i$ whose entries are at most $\sqrt{k}C$.  Now we can write 
\[
r(x,y) = \sum_{i = 1}^{\binom{d+2}{2}} (v(x,y)^\top u_i )^2
\]
and each of the $v(x,y)^\top u_i $ are $(d, \sqrt{k}C)$-bounded polynomials in $x,y$.  Thus, $r(x,y)$ is a $(3d^2, d, \sqrt{k}C) $-bounded sum-of-squares polynomial, as desired.
\end{proof}

Now we move onto the main proof of \cref{thm:exp-monotone-approx}.  We have the following basic identities.

\begin{fact}[Integration Identities]
\label{fact:derivative-int-substitution-fact}
For a function $p: \R \rightarrow \R$,
   \begin{equation*}
       p(z+a) - p(z-a) = a \int_{0}^{1} (p'(z+ ta) + p'(z - ta)) \diff t 
   \end{equation*} 
and 
\begin{equation*}
       p(z+a) + p(z-a) =  2p(z) +  a \int_{0}^{1} (p'(z+ ta) + p'(z - ta)) \diff t 
   \end{equation*} 
\end{fact}
\begin{proof}
Both of the equations follow immediately from the Fundamental Theorem of Calculus.
\end{proof}

Recall that the key inequality that we need to prove is the following.

\begin{lemma}[Modified Gradient Identity with bounded coeffcieints]
\label{lem:modified-gradient-identity}
For all positive integers $k,\ell$ and real numbers $x,y$,
\[
0.5(x - y)(1 + 0.25(x - y)^2) (q_{k,\ell}(x) - q_{k,\ell}(y)) - 0.00025(x - y)^2 p_{k,\ell}(x) \geq 0 \,.
\]
Furthermore, the LHS is a $(10^{2^k\ell} , 2^k \ell + 10, 200^{2^k\ell})  $-bounded SoS polynomial (in the variables $x,y$).
\end{lemma}

First, we prove the following algebraic identity which will be used to rewrite the expression in Lemma~\ref{lem:modified-gradient-identity} in a form for which it is easy to prove non-negativity.

\begin{lemma}[Relating $q$ to first and second derivative]\label{claim:manip-of-gradient-identity}
Let $k,\ell$ be positive integers.  Let
\[
r(x,y) = 0.5(x - y)(1 + 0.25(x - y)^2) (q_{k,\ell}(x) - q_{k,\ell}(y)) - 0.00025(x - y)^2 ( p_{k,\ell}(x) + p_{k,\ell}(y)) \,.
\]
Let $z = (x + y)/2$ and $a = (x - y)/2$.  We have the equality
\begin{equation}
\label{eqn:expanding-rxy}
\begin{split}
r(x,y) &= \underbrace{ \int_{0}^{1} \Paren{ \Paren{0.998 a^2 + a^4 } p_{k,\ell}(z + ta ) - 0.001a^3 (2-t) p'_{k,\ell}(z + ta) } \diff t }_{\eqref{eqn:expanding-rxy}.(1)}  \\
& \hspace{0.2in} + \underbrace{ \int_{0}^{1} \Paren{ \Paren{0.998 a^2 + a^4 } p_{k,\ell}(z - ta ) - 0.001a^3 (2-t) p'_{k,\ell}(z - ta) } \diff t}_{\eqref{eqn:expanding-rxy}.(2)}  
\end{split}
\end{equation}
\end{lemma}
\begin{proof}

Recalling that $p_{k,\ell}$ is the derivative of $q_{k,\ell}$, we can write

\begin{equation}
\label{eqn:writing-rxy-with-z-a}
    \begin{split}
r(x,y)&=(a + a^3)(q_{k,\ell}(z + a) - q_{k,\ell}(z - a)) - 0.001a^2(p_{k,\ell}(z + a) + p_{k,\ell}(z - a)) \\ &= (a^2 + a^4)\int_{0}^1 p_{k,\ell}(z + ta) + p_{k,\ell}(z - ta) \diff t \\
& \hspace{0.2in} - 0.002a^2 p_{k,\ell}(z) -  0.001a^3 \int_{0}^1 p'_{k,\ell}(z + ta) - p'_{k,\ell}(z - ta) \diff t,
\end{split}
\end{equation}
where  $p'$ is the derivative of $p$ and the equality follows from applying Fact~\ref{fact:derivative-int-substitution-fact} to $q_{k,\ell}$ and $p_{k,\ell}$. Next, observe 
\begin{equation}
\begin{split}
    & 2p_{k,\ell}(z)  = 2 \int_{0}^{1} p_{k, \ell}(z) \diff t  \\
    & = 2 \int_{0}^{1} p_{k, \ell}(z) \diff t  - \int_{0}^{1} \Paren{ p_{k,\ell }(z + ta ) +  p_{k, \ell}(z - ta) } + \int_{0}^{1} \Paren{ p_{k,\ell }(z + ta ) +  p_{k,\ell}(z - ta) } \\
    & = \int_{0}^{1} \Paren{  p_{k, \ell} ( z + ta ) -  p_{k, \ell}(z) } \diff t  - \int_{0}^{1} \Paren{ p_{k,\ell }(z) - p_{k,\ell} (z - ta) } \diff t  + \int_{0}^{1} \Paren{ p_{k,\ell }(z + ta ) +  p_{k,\ell}(z - ta) } \diff t \\
    & = \int_{0}^{1} a \Paren{ \int_{0}^{t} \Paren{ p'_{k,\ell}(z + sa)   - p'_{k,\ell}(z - sa) } \diff s } \diff t  + \int_{0}^{1} \Paren{ p_{k,\ell }(z + ta ) +  p_{k,\ell}(z - ta) } \diff t 
\end{split}
\end{equation}

Substituting the expanded expression for $2p_{k,\ell}$ back into Equation~\eqref{eqn:writing-rxy-with-z-a}, 
\[
\begin{split}
    r(x,y) & = \Paren{0.998a^2  + a^4}  \int_{0}^1 \Paren{ p_{k,\ell}(z + ta) + p_{k,\ell}(z - ta) } \diff t      \\
    & \hspace{0.2in} - 0.001 a^2 \int_{0}^{1} a \Paren{ \int_{0}^{t} \Paren{ p'_{k,\ell}(z + sa)   - p'_{k,\ell}(z - sa) } \diff s } \diff t \\
    & \hspace{0.2in} -  0.001a^3 \int_{0}^1 p'_{k,\ell}(z + ta) - p'_{k,\ell}(z - ta) \diff t\\
    & = \Paren{0.998a^2  + a^4}  \int_{0}^1 \Paren{ p_{k,\ell}(z + ta) + p_{k,\ell}(z - ta) } \diff t      \\ 
    & \hspace{0.2in} - 0.001 a^3    \int_{0}^{1} \Paren{1 - s}  \Paren{ p'_{k,\ell}(z + sa)   - p'_{k,\ell}(z - sa) } \diff s   \\
    & \hspace{0.2in} -  0.001a^3 \int_{0}^1 \Paren{ p'_{k,\ell}(z + ta) - p'_{k,\ell}(z - ta) } \diff t \\
    & = \Paren{0.998a^2  + a^4}  \int_{0}^1 \Paren{ p_{k,\ell}(z + ta) + p_{k,\ell}(z - ta) } \diff t      \\  
    & \hspace{0.2in} - 0.001 a^3    \int_{0}^{1} \Paren{2 - t}  \Paren{ p'_{k,\ell}(z + ta)   - p'_{k,\ell}(z - ta) } \diff t , \\
    & =  \underbrace{ \int_{0}^{1} \Paren{ \Paren{0.998 a^2 + a^4 } p_{k,\ell}(z + ta ) - 0.001a^3 (2-t) p'_{k,\ell}(z + ta) } \diff t }_{\eqref{eqn:expanding-rxy}.(1)}  \\
    & \hspace{0.2in} + \underbrace{ \int_{0}^{1} \Paren{ \Paren{0.998 a^2 + a^4 } p_{k,\ell}(z - ta ) - 0.001a^3 (2-t) p'_{k,\ell}(z - ta) } \diff t}_{\eqref{eqn:expanding-rxy}.(2)}  
\end{split}
\]
where the second equality follows from switching the order of the integral and observing that the coefficient for a fixed $s$ is simply $1-s$. 
\end{proof}

In light of Lemma~\ref{claim:manip-of-gradient-identity}, it remains to prove that the expressions \eqref{eqn:expanding-rxy}.(1) and \eqref{eqn:expanding-rxy}.(2) are nonnegative and can be written as bounded SoS polynomials in the variables $x,y$.

\begin{lemma}\label{claim:bounded-sos-polynomial-piece}
Let $z = (x + y)/2$ and $a = (x - y)/2$.  Then the expressions \eqref{eqn:expanding-rxy}.(1)  and \eqref{eqn:expanding-rxy}.(2) are nonnegative and can be written as a $(10 \cdot (2^k \ell)^2 , 2^k \ell + 10, 150^{2^k\ell}) $-bounded SoS polynomials in the variables $x,y$.
\end{lemma}
\begin{proof}
We focus on the expression \eqref{eqn:expanding-rxy}.(1).  The argument for \eqref{eqn:expanding-rxy}.(2)  is exactly the same. First, we show that expression \eqref{eqn:expanding-rxy}.(1) is non-negative. Using $1+a^2 = \frac{ (1+a)^2 + (1-a)^2}{2}$ and $a = \frac{ (1+a)^2 - (1-a)^2 }{4}$, 

\begin{equation}
\label{eqn:expanding-rxy-term1}
    \begin{split}
      \eqref{eqn:expanding-rxy}.(1) & = a^2 \cdot \int_{0}^{1} \Paren{ \Paren{0.998 + a^2 } p_{k,\ell}(z + ta ) - 0.001a (2-t) p'_{k,\ell}(z + ta) } \diff t   \\
      &  =  a^2 \cdot \int_{0}^{1} (2-t) \Paren{ \frac{ \Paren{1 + a^2 } }{8} p_{k,\ell}(z + ta ) -  \frac{ 0.001 ((1+a)^2 - (1-a)^2 )}{ 4}  p'_{k,\ell}(z + ta) } \diff t \\
      & \hspace{0.2in} +  a^2 \cdot \int_{0}^1 \Paren{ (0.998 + a^2) - \frac{ (2-t) (1+a^2) }{8 } p_{k,\ell}(z+ta) } \diff t  \\  
      & =   a^2 \cdot \int_{0}^{1} \frac{  (2-t) (1 + a)^2  }{16}  \Paren{    p_{k,\ell}(z + ta ) -  0.004 p'_{k,\ell}(z + ta)   } \diff t  \\
      & \hspace{0.2in}  + a^2 \cdot \int_{0}^{1} \frac{  (2-t) (1 - a)^2  }{16}  \Paren{    p_{k,\ell}(z + ta ) +  0.004 p'_{k,\ell}(z + ta)   } \diff t ,\\
      & \hspace{0.2in}  + a^2 \cdot \int_{0}^1 \Paren{ (0.998 + a^2) - \frac{ (2-t) (1+a^2) }{8 } } p_{k,\ell}(z+ta)  \diff t .
    \end{split}
\end{equation}



To prove nonnegativitiy, it suffices to show that $a_{k,\ell}(x) = p_{k,\ell}(x) - 0.004 p'_{k,\ell}(x)$ and $b_{k,\ell}(x) = p_{k,\ell}(x) +  0.004 p'_{k,\ell}(x)$ are non-negative for all $x \in \mathbb{R}$. 
Recall, $p_{k,\ell}(x) =  \prod_{j \in [k]} s_{2^j \ell} (x/k)$. Therefore, by product rule,
\begin{equation*}
\begin{split}
    p'_{k,\ell}(x) & = \sum_{j \in [k] }  \frac{ s'_{2^j \ell }(x/k) }{k}  \cdot \prod_{j'\neq j \in [k]} s_{2^{j'} \ell} (x/k)  \\
    & = \sum_{j \in [k] }  \frac{ s_{(2^j \ell)  - 1 }(x/k) }{k}  \cdot \prod_{j'\neq j \in [k]} s_{2^{j'} \ell} (x/k)   .
\end{split}
\end{equation*}
Therefore, 
\begin{equation}
    a_{k, \ell}(x) = \sum_{j \in [k]} \Paren{ \frac{ s_{2^j \ell }(x/k) - 0.004 \cdot s_{(2^j \ell) - 1}(x/k)  }{k}  }   \prod_{j'\neq j \in [k]} s_{2^{j'} \ell} (x/k)  \geq 0 
\end{equation}
since for all $j \in [k]$,  $s_{2^j \ell}(x/k) - 0.004 s_{ (2^j \ell)-1 }(x/k) \geq 0$ by Corollary~\ref{coro:bouned-sos-coeffs1} and $s_{2^{j'}  \ell  }(x/k) \geq 0$ by Corollary \ref{coro:positive-truncation}. Similarly, 
\begin{equation}
    b_{k,\ell}(x) = \sum_{j \in [k]} \Paren{ \frac{ s_{2^j \ell }(x/k) +  0.004 \cdot s_{(2^j \ell) - 1}(x/k)  }{k}  }   \prod_{j'\neq j \in [k]} s_{2^{j'} \ell} (x/k)  \geq 0
\end{equation}
since for all $j \in [k]$,  $s_{2^j \ell}(x/k) + 0.004 s_{ (2^j \ell)-1 }(x/k) \geq 0$ by Corollary~\ref{coro:bouned-sos-coeffs1}.

Now we have proved that the expression is nonnegative, we show that it is a SoS polynomial with bounded coefficients.  By Corollary~\ref{coro:bouned-sos-coeffs1} the polynomials  
\[ s_{2^j \ell}(x/k) \pm 0.004 s_{ (2^j \ell)-1 }(x/k) , s_{2^j \ell}(x/k) \] are all
$(2^{2^{j-1}\ell}, 2^{j-1}\ell, 10^{2^{j-1}\ell})$-bounded SoS polynomials.  Thus, by Claim~\ref{claim:basic-composition-properties}, $a_{k,\ell}(x)$  and $b_{k,\ell}(x)$ are both $(k \cdot 2^{2^{k}\ell} , 2^k \ell, 90^{2^k\ell}) $-bounded SoS polynomials.  Since $z + ta = \frac{1 + t}{2} x + \frac{1 - t}{2} y$, by Claim~\ref{claim:basic-composition-properties} again, we get that for any real number $t \in [0,1]$, 
\[
\frac{  (2-t) (1 + a)^2  }{16}  \Paren{    p_{k,\ell}(z + ta ) -  0.004 p'_{k,\ell}(z + ta)   }
\]
is a $(k \cdot 2^{2^{k}\ell} , 2^k \ell, 100^{2^k\ell}) $-bounded SoS polynomial in the variables $x,y$ after substituting $z = (x+y)/2, a = (x-y)/2$.  We can use a similar argument for the other terms in the expression in \eqref{eqn:expanding-rxy-term1}.  Then we can use Claim~\ref{claim:sos-integration} to bound the integral over $t$ and deduce that the expression \eqref{eqn:expanding-rxy}.(1) is a $(10 \cdot (2^k \ell)^2 , 2^k \ell + 10, 150^{2^k\ell}) $-bounded SoS polynomial in the variables $x,y$.
\end{proof}

Now we can complete the proof of Lemma~\ref{lem:modified-gradient-identity}.
\begin{proof}[Proof of Lemma~\ref{lem:modified-gradient-identity}]
Note that the expression in Lemma~\ref{lem:modified-gradient-identity} is equal to 
\[
r(x,y) + 0.00025(x-y)^2 p_{k,\ell}(y) \,.
\]
Now by the definition of $p_{k,\ell}$, Corollary~\ref{coro:bouned-sos-coeffs1}, and \cref{claim:basic-composition-properties}, we get that $p_{k,\ell}(y)$ is is a $(2^{2^{k}\ell} , 2^k \ell, 90^{2^k\ell}) $-bounded SoS polynomial.  \cref{claim:bounded-sos-polynomial-piece} and \cref{claim:manip-of-gradient-identity} imply that $r(x,y)$ is a $(20 \cdot (2^k \ell)^2 , 2^k \ell + 10, 150^{2^k\ell}) $-bounded SoS polynomial so overall, $r(x,y) + 0.00025(x-y)^2 p_{k,\ell}(y)$ is a $(10^{2^k\ell} , 2^k \ell + 10, 200^{2^k\ell})  $-bounded SoS polynomial (and thus also nonnegative) and we are done.
\end{proof}

Finally, this completes the proof of Theorem~\ref{thm:exp-monotone-approx}.
\begin{proof}[Proof of Theorem~\ref{thm:exp-monotone-approx}]
The desired result follows immediately from Lemma~\ref{lem:modified-gradient-identity}.
\end{proof}

\end{document}